\documentclass[a4paper,11pt]{article}
\usepackage{test,subcaption,framed,booktabs,multirow}

\usepackage[inline,shortlabels]{enumitem}
\setlist[enumerate,1]{label=(\roman*)}
\usepackage[colorlinks,citecolor={blue}]{hyperref}
\usepackage[authoryear,round]{natbib}
\usepackage{doi}
\usepackage{xurl} 

\usepackage[margin = 1.5in]{geometry}
\usepackage{setspace}

\usepackage{graphicx}
\graphicspath{{figures/}} 

\newcommand{\cV}{\mathcal{V}}
\newcommand{\cR}{\mathcal{R}}
\newcommand{\cW}{\mathcal{W}}
\newcommand{\Tvac}{1}
\newcommand{\IFR}{0.0027}
\newcommand{\CFR}{0.0133}

\newcommand{\uD}{-12.22}
\newcommand{\ImaxMy}{23.9}
\newcommand{\ImaxEq}{6.63}

\title{Optimal Epidemic Control in Equilibrium with Imperfect Testing and Enforcement\thanks{The views stated herein are those of the authors and are not necessarily those of the Federal Reserve Bank of Cleveland or the Board of Governors of the Federal Reserve System. Our work was motivated by the COVID-19 pandemic and the first draft \citep{PhelanToda2021} was posted at \url{https://arxiv.org/abs/2104.04455v1} on April 9, 2021. Our analysis is based on the best scientific evidence available at that time and is applicable to (possibly future) pandemics for which the SIR framework is appropriate. All replications files are available at \url{https://github.com/tphelanECON/Epidemic_Equilibrium}.}}
\author{Thomas Phelan\thanks{Federal Reserve Bank of Cleveland. Email: \href{mailto:tom.phelan@clev.frb.org}{tom.phelan@clev.frb.org}.} \and Alexis Akira Toda\thanks{Department of Economics, University of California San Diego. Email: \href{mailto:atoda@ucsd.edu}{atoda@ucsd.edu}.}}

\numberwithin{equation}{section}
\numberwithin{thm}{section}

\begin{document}
\maketitle

\begin{abstract}
We analyze equilibrium behavior and optimal policy within a Susceptible-Infected-Recovered epidemic model augmented with potentially undiagnosed agents who infer their health status and a social planner with imperfect enforcement of social distancing. We define and prove the existence of a perfect Bayesian Markov competitive equilibrium and contrast it with the efficient allocation subject to the same informational constraints. We identify two externalities, static (individual actions affect current risk of infection) and dynamic (individual actions affect future disease prevalence), and study how they are affected by limitations on testing and enforcement. We prove that a planner with imperfect enforcement will always wish to curtail activity, but that its incentives to do so vanish as testing becomes perfect. When a vaccine arrives far into the future, the planner with perfect enforcement may encourage activity before herd immunity. We find that lockdown policies have modest welfare gains, whereas quarantine policies are effective even with imperfect testing.


\medskip


\medskip

\end{abstract}

\onehalfspacing

\section{Introduction}

Soon after the evidence of the first community spread of the Coronavirus Disease 2019 (COVID-19) outside China was reported in Italy in late February 2020, European countries promptly introduced drastic mitigation measures (``lockdown'') such as closure of schools, restaurants, and other businesses. Many states and provinces in the United States and Canada as well as countries around the world followed suit by mid-March. While implementing policies to slow the spread of an infectious disease appears to be an obvious course of action for a prudent government, such mitigation policies have evidently not been without costs. Curtailing economic activities can (and did) cause unemployment, bankruptcy, and reduced access to education. Further, engaging in non-pharmaceutical interventions (better known as ``social distancing'') reduces new infections but also delays achieving herd immunity, and so may prolong the epidemic in the absence of a vaccine. It is therefore possible that lockdown policies can slow the progress of the epidemic but do little to alter its ultimate toll. 

Although it appears we are now in the later stages of the pandemic, with several vaccines developed and administered, the immense disruption to economic and social activity wrought by the virus and the possibility of future pandemics (either due to variants or new viruses) motivates a theoretical analysis of an epidemic model suitably augmented with realistic features to capture policy-relevant tradeoffs. In this paper we build upon the standard \cite{KermackMcKendrick1927} Susceptible-Infected-Recovered (SIR) model and add two important features that seem to be overlooked in the current literature. First, the agents in our model are forward-looking and endogenously respond to the epidemic, but must continuously update their beliefs about their own health status because they may lack symptoms or testing may not be available. Second, we study the extent to which prescriptions for policy depend upon the ability of the government to enforce their recommended actions over the long-term. The first feature is motivated by the fact that the infection fatality rate (IFR, fraction of deaths among \emph{all} cases) of COVID-19 inferred from seroprevalence studies is an order of magnitude smaller than the case fatality rate (CFR, fraction of deaths among \emph{confirmed} cases), which suggests substantial underreporting.\footnote{According to the meta-analysis of seroprevalence studies by \cite{Ioannidis2021}, the median IFR of COVID-19 is 0.27\%. On the other hand, as we document in Section \ref{sec:num}, the median CFR across more than 200 countries and regions is 1.33\%. Thus the diagnosis rate is of the order $0.27/1.33=20\%$.} The second is motivated by the fact that some governments were both slow to impose social distancing measures and may lack the ability to enforce such measures over the long-run, possibly due to opposition from constituents. To the best of our knowledge, no existing work considers these two features and studies the role they play in shaping optimal policy responses.

Our model works as follows. The society consists of four behavioral types of agents: unknown, infected, recovered, and dead. The unknown type consists of agents who lack immunity (are susceptible) against the infectious disease as well as those who are infected or recovered but unconfirmed due to imperfect testing or asymptomatic infections. Each period, alive agents take an action $a\in [0,1]$, which we interpret as their overall level of economic activity. We assume that in the absence of an epidemic, agents prefer taking the highest action $a=1$. During an epidemic, taking higher actions exposes oneself to the risk of infection. As a result, rational agents without confirmed immunity (the unknown type) optimally choose lower actions, \ie, they voluntarily practice social distancing. In contrast, known infected and recovered agents have no incentive to social distance and choose the highest action available to them. We define a perfect Bayesian Markov competitive equilibrium to be an allocation in which
\begin{enumerate*}
\item agents form beliefs about their health status and optimize given the state variables (population shares of each type) and
\item the evolution of beliefs and state variables are consistent with Bayes' rule and the collective behavior of agents.
\end{enumerate*}
We obtain two main theoretical results. First, we prove the existence of a (pure strategy) perfect Bayesian Markov competitive equilibrium (Markov equilibrium for short). This result is important because to achieve the equilibrium, individuals and policy makers only need to form expectations about the future given a few state variables and do not require implausibly sophisticated coordination among them.

The equilibrium allocation is in general inefficient due to the externality caused by the actions of infected agents. Existing analyses of the pandemic have either focused upon allocations in which a planner dictates all activity in perpetuity,\footnote{Examples include \cite{AlvarezArgenteLippi2021}, \cite{FarboodiJaroschShimer2021}, and \cite{KruseStrackControl}, among others.} or laissez-faire allocations in which no social distancing is imposed.\footnote{An example is \cite{ToxvaerdEquilibrium}.} Although obviously informative, neither of these cases addresses the problem of a planner who was previously slow to act, or who believes their capacity to enforce restrictions may dissipate over time. The recursive nature of our solution methods allows us to address this situation, as we compute equilibrium and efficient activity levels at every point in the state space. In order to highlight the role that imperfect enforcement plays in optimal policy responses, we distinguish between two types of externalities: static and dynamic. The static externality arises because the activity of infected agents affects the probability that susceptible agents become infected in the current period. The dynamic externality arises because the collective behavior of agents affects the evolution of the prevalence of the virus throughout the population. The interplay between the static and dynamic externalities is subtle as they can move in different directions: when an individual chooses higher activity, they increase the risk of infecting their fellow citizens today, but if they become infected they reduce the risk they pose to others in the future.

This brings us to the second theoretical result. We prove that the difference between the static efficient (the socially optimal choice taking future prevalence as given) and individually optimal actions is bounded above by a number proportional to the fraction of unknown infected agents, a quantity that vanishes as the probability of diagnosis converges to 1. This shows that a government that can only enforce short-term lockdown policies will always wish to curb activity, but its incentive to do so vanishes with the fraction of unconfirmed cases. This observation is particularly noteworthy because as we show in our numerical results, the activity levels recommended by a government with unlimited enforcement power are not always lower than those which occur in equilibrium. It is in this sense that we show how policy prescriptions can depend crucially on enforcement capabilities.


The presence of 
confirmed and unconfirmed infected agents implies that there are two tools for intervention available to the planner: activity recommendations for unknown agents (referred to as \emph{lockdown} policies) and recommendations for known infected agents (referred to as \emph{quarantine} policies). The combination of undiagnosed agents and the possibility of imperfect enforcement capabilities is precisely what makes the problem of the planner difficult. Indeed, if all infected agents could be immediately and costlessly quarantined, then the pandemic would likely not have had the immense economic and social impact that we have observed over the past years. In the majority of this paper, we therefore focus primarily on the recommendations to unknown agents, taking as given a fixed level of activity by the known infected agents.

To illustrate our theory as well as to study the optimal interventions, we calibrate our model to the COVID-19 epidemic and conduct a number of robustness exercises. Due to endogenous social distancing, in equilibrium the infection curve flattens relative to the case of myopic behavior (the standard SIR model). We find that the planner's optimal interventions are significantly affected by the diagnosis rate of infections, the vaccine arrival rate, and the planner's ability to enforce policies over the long term. The welfare gains from lockdown policies tend to be small, reducing the welfare loss from the pandemic relative to the equilibrium by less than 10\%. When the vaccine is expected to arrive within a year or two (as could be expected at the beginning of the COVID pandemic), the optimal reduction in activity begins earlier, is more gradual, and extends well beyond the date at which activity has returned to normal in the equilibrium allocation. In contrast, when the vaccine is expected to arrive decades into the future, the optimal policy is to \emph{encourage} (discourage) agents to take high actions \emph{before} (after) achieving herd immunity so that herd immunity is achieved quickly but unnecessary deaths are avoided. In contrast to lockdown policies, we find that quarantine policies are effective even with imperfect testing.


\subsection{Related literature}

\cite{KermackMcKendrick1927} present the basic mathematical framework for studying the evolution of infectious diseases. Models that build upon this framework assume that agents can be placed into categories based on their health status, that there is a fixed probability that an infected agent passes the infection to a susceptible agent when they meet, and that there is a fixed rate at which infected agents recover or die. In the simplest formulation, there are only susceptible and infected agents. These SI models are appropriate for infectious diseases that are incurable but not deadly.\footnote{An example is the Epstein-Barr virus (EBV) infection, which causes latent lifelong infections. Most people are infected during childhood and experience only mild symptoms such as fever. When infection occurs among adolescents and adults, EBV causes infectious mononucleosis (kissing disease) in 30--50\% of the cases.} SIS models are ones in which infected agents can recover, but when they do, they become susceptible to reinfection. SIR models are ones in which infected agents can recover (or die) and acquire lifelong immunity.

Although mathematical epidemic models provide insights into the spread of infectious diseases, they often ignore individual choice or public policy.\footnote{\cite{EksinPaarpornWeitz2019} document that mathematical epidemic models that ignore behavioral changes have large forecast errors.} There are several papers that modify the basic SIR model to allow for either government policies or individual decisions to influence the course of the epidemic. We describe the models as \emph{non-strategic} if the government has the ability to mandate changes through quarantines, lockdowns, or other non-pharmaceutical interventions. We call the models \emph{strategic} if individual agents independently decide levels of care that influence exposure to the disease.

\cite{Sethi1978} examines the problem of a planner who can choose to quarantine a fraction of infected agents in an SIS model. With linear payoffs and costs, he identifies a bang-bang solution in which the planner either quarantines all infected agents or none of them. \cite{GersovitzHammer2004} study the externalities in a model in which the disease transmission and recovery can be controlled at a cost through preventive and therapeutic efforts. More recently, \cite{KruseStrackControl} incorporate social distancing in a non-strategic model of infections and study a version of the SIR model in which a social planner can, at a cost, influence the transmission rate. With linear cost, they show that socially optimal policies are bang-bang: the social planner reduces the transmission rate as much as possible or does not reduce the transmission rate at all. It is typically not optimal to reduce the transmission rate when the fraction of infected agents is small. For some of their analysis, \cite{KruseStrackControl} assume that the planner can only impose social distancing for (no more than) a fixed length. With this restriction, lockdowns should start only after the number of infected agents reaches a threshold. 

Turning to strategic models, \cite{GeoffardPhilipson1996}, \cite{Kremer1996}, and \cite{Auld2003} study the extent to which strategic choices may undermine the effects of public policy regarding the spread of an infectious disease. These papers focus on HIV (human immunodeficiency virus) infections, where the heterogeneity of the population and the ability to select who to interact with are of first-order importance. \cite{Reluga2010}, \cite{ChenJiangRabidouxRobinson2011}, \cite{Fenichel2011}, \cite{Chen2012}, \cite{Fenichel2013}, and \cite{Toxvaerd2019} present strategic models of social distancing that predates the current COVID-19 epidemic. \cite{Fenichel2013}, which is an extended analysis of \cite{Fenichel2011}, assumes that agents can select the intensity of their interaction with others and assumes that flow utility is a single-peaked concave function of this intensity. He contrasts socially optimal choices of these contact levels with privately selected values and points out that if the social planner cannot distinguish between groups (and therefore any restriction on interactions must apply to susceptible, infected, and recovered agents alike), then social welfare may be higher in the laissez-faire equilibrium than in the constrained social planner's problem. This possibility arises because the planner's intervention constrains the participation of recovered agents (who generate positive externalities) in addition to the participation of infected agents.

\cite{ChenJiangRabidouxRobinson2011} and \cite{Chen2012} study a static game in which susceptible agents decide on their level of activities. This game may exhibit multiple equilibria, which are typically inefficient. Whether the susceptible agents are more or less active than the social optimum depends on the nature of the matching technology. More recently, \cite{Toxvaerd2019} points out that in the presence of strategic agents, public policy interventions that lower the infectiousness of a disease may lower social welfare because agents respond to the change by increasing their own exposure.

Following the onset of the COVID-19 epidemic, a large number of papers have been written by economists. Since this literature is too large to review, we only discuss the subset of papers that focus on the theory and applications. \cite{AbelPanageas2021} study an optimal control problem as well as the laissez-faire equilibrium in an SIR model with population growth and show that a steady state exists and the disease becomes endemic regardless of the cost of excess deaths. \cite{Budish2020} conceptualizes $\mathcal{R}\le 1$ (effective reproduction number less than 1) as a constraint, discusses the optimal policy in a static setting with heterogeneous economic activities, and illustrates that cheap policies such as mask-wearing go a long way in containing the virus spread with minimal welfare costs. \cite{ToxvaerdEquilibrium} studies an SIR model with endogenous social distancing, which is similar to ours. Assuming linear utility and costs, he shows that susceptible agents either engage in no social distancing at all or social distance to maintain a target peak prevalence, which endogenously flattens the infection curve.

Relative to this small literature of theoretical strategic epidemic models, our main contribution is that we explicitly model imperfect testing and enforcement and systematically study the welfare implications and optimal policies.

\section{Behavioral SIR model with imperfect testing}\label{sec:model}

We introduce rational and potentially undiagnosed agents into the basic \cite{KermackMcKendrick1927} Susceptible-Infected-Recovered (SIR) epidemic model.

\subsection{Model}

We consider an infectious disease that can be transmitted between agents in a society, which consists of a large (but finite) number of agents indexed by $n\in N=\set{1,\dots,N}$. Time is discrete, runs forever, and is denoted by $t=0,\Delta, 2\Delta,\dotsc$, where $\Delta>0$ is the length of time in one period.

\paragraph{Agent types and information}

At each point in time, agents are categorized into several types based on their health status and information. Agents who do not have immunity against the infectious disease are called \emph{susceptible} and denoted by $S$. Infected agents who are known (unknown) to be infected are denoted by $I_k$ ($I_u$). Agents with immunity who are known (unknown) to be immune are called \emph{removed} (or \emph{recovered}) and denoted by $R_k$ ($R_u$). Dead agents are denoted by $D$. The set of all types (health status) is denoted by
\begin{equation*}
H\coloneqq \set{S,I_k,I_u,R_k,R_u,D}.
\end{equation*}
For instance, agents could be $I_k$ if they test positive for antigen or they show specific symptoms (are \emph{symptomatic}), and $I_u$ if no tests are available and they show no specific symptoms (are \emph{asymptomatic}). Similarly, agents could be $R_k$ if they test positive for antibody, or they recovered from past symptomatic infection and immunity is lifelong, or they are vaccinated. Thus the set of information types is $\set{U,I_k,R_k,D}$, where $U\coloneqq S\cup I_u\cup R_u$ denotes the \emph{unknown} type. Importantly, we suppose that when agents get infected, with probability $\sigma\in (0,1]$ they receive a signal that reveals the true health status (known infected, $I_k$). Otherwise, they become unknown infected ($I_u$). Although we refer to the signal as a ``test'', the signal could be literally a laboratory test as well as other information such as the presence of specific symptoms, knowledge of close contacts with confirmed cases, etc. We refer to the probability $\sigma$ as the \emph{diagnosis rate}.\footnote{In Appendix \ref{sec:signal} we show that a model with a single signal is observationally equivalent to a model with multiple signals with potentially heterogeneous fatality rates.}

Let $N_h\subset N$ be the set of agents with health status $h\in H$, simply referred to as ``$h$ agents''. With a slight abuse of notation, we also use the same symbol $h$ to denote the fraction of $h$ agents in the population, so $h=\abs{N_h}/N$. The space of the aggregate state (type distribution) is denoted by
\begin{equation}
Z\coloneqq \set{z=(S,I_k,I_u,R_k,R_u,D)|S+I_k+I_u+R_k+R_u+D=1, Nz\in \Z_+^6}.\label{eq:Z}
\end{equation}
We suppose that the aggregate state is observable. ($I_k,R_k,D$ are observable, and $I_u$ and $R_u$ can be inferred from a small scale random antigen and antibody testing.) The economy starts at $t=0$ with some initial condition $z_0\in Z$.

\paragraph{Actions and preferences}
At each point in time, each alive agent takes action $a\in A=[\ubar{a},1]$, where the minimum action $\ubar{a}$ satisfies $0\le \ubar{a}<1$. We interpret $a$ as the economic activity level: loosely, $a=1$ corresponds to following a normal life and $a=\ubar{a}$ to completely being locked down (minimum activity level for subsistence). The utility function of $U$ (unknown) and $R_k$ (known recovered) agents is denoted by $u:A\to \R$. The utility function of an $I_k$ (known infected) agent is denoted by $u_I:A\to \R$. A $D$ (dead) agent receives the flow utility $u_D\in \R$.\footnote{More precisely, $u_D$ is the flow utility of being dead anticipated by alive agents.} Agents discount future payoffs with discount factor $\e^{-r\Delta}$, where $r>0$ is the discount rate.

\paragraph{Disease transmission}
Agents meet each other randomly over time and transmit the infectious disease. If agents $n,n'$ take actions $a_n,a_{n'}$ respectively (we interpret a dead agent's action as $a=0$), then agent $n$ bumps into agent $n'$ during a period with probability $\lambda \Delta a_na_{n'}/N$, where $\lambda\in (0,1/\Delta)$ is a parameter (meeting rate) that governs the level of social interaction with full activity ($a_n=a_{n'}=1$). We take $\lambda$ as given, which depends on how the society is organized (\eg, population density, whether workers commute by cars or public transportation, whether consumers shop online or at physical stores, whether classes are taught remotely or in-person, etc.).


If agent $n$ is susceptible ($n\in N_S$) and agent $n'$ is infected ($n'\in N_{I_k}\cup N_{I_u}$), the infectious disease is transmitted from $n'$ to $n$ with probability $\tau\in (0,1]$ conditional on $n$ bumping into $n'$ at time $t$,\footnote{Thus the act of ``bumping into'' is asymmetric between the members in a meeting. Figuratively, when agent $n$ bumps into $n'$: 
\begin{enumerate*}
\item agent $n$ meets agent $n'$,
\item agent $n'$ sneezes into agent $n$'s face (and transmits the disease to $n$ with probability $\tau$ if $n'$ is infected), and
\item they part with each other.
\end{enumerate*}} which we take as given.\footnote{The transmission probability $\tau$ depends on how contagious the disease is as well as how the society is organized (\eg, how often people wash their hands, whether they wear masks, whether they greet others by bowing, shaking hands, hugging, or kissing, etc.). In addition to the activity choice, it is straightforward to extend the model to allow for the choice of preventive efforts such as hand-washing and mask-wearing. 
} We assume that $\Delta$ is small enough such that in any period an agent bumps into at most one other agent. Therefore if $a_h\coloneqq (\sum_{n'\in N_h} a_{n'})/\abs{N_h}$ denotes the average action of $h$ agents, a particular susceptible agent $n$ who takes action $a$ gets infected with probability
\begin{align}
\sum_{n'\in N_{I_k}\cup N_{I_u}}\tau\frac{\lambda \Delta aa_{n'}}{N}&=\tau \lambda \Delta a\left(\frac{\sum_{n'\in N_{I_k}} a_{n'}}{\abs{N_{I_k}}}\frac{\abs{N_{I_k}}}{N}+\frac{\sum_{n'\in N_{I_u}} a_{n'}}{\abs{N_{I_u}}}\frac{\abs{N_{I_u}}}{N}\right) \notag \\
&=\beta \Delta (a_{I_k}I_k+a_{I_u}I_u)a,\label{eq:infect_prob}
\end{align}
where $\beta\coloneqq \tau\lambda$ is the baseline transmission rate and we have used the notation $h=\abs{N_h}/N$ for $h=I_k,I_u$. The timing convention is that if an infection occurs at time $t$, the (previously susceptible, now infected) agent changes status to $I_k$ or $I_u$ at time $t+\Delta$. Since only susceptible agents are prone to infection, the expected fraction of the population that gets newly infected between time $t$ and $t+\Delta$ is
\begin{equation}
\beta \Delta a_SS(a_{I_k}I_k+a_{I_u}I_u),\label{eq:incidence}
\end{equation}
which is called \emph{incidence} in epidemiology. The fraction of infected agents
\begin{equation}
\abs{N_{I_k}\cup N_{I_u}}/N=I_k+I_u \label{eq:prevalence}
\end{equation}
is called \emph{prevalence}.

\paragraph{Recovery and death}
An $I$ agent is removed (becomes no longer infected by either recovering or dying) with probability $\gamma \Delta$ each period, where $\gamma\in (0,1/\Delta]$ is the removal rate. Conditional on being removed, an $I_k$ (known infected) agent dies with probability $\delta\in (0,1]$ and an $I_u$ (unknown infected) agent always recovers. The rationale for this assumption is that infected agents with more severe symptoms are more likely to get tested as well as to die. Letting $\delta_0$ be the fatality rate among all (known and unknown) infected agents, we have 
\begin{equation}
\delta=\delta_0/\sigma.\label{eq:deltasigma}
\end{equation}
Finally, $R$ (recovered) or $D$ (dead) agents remain in their corresponding states forever, which embodies the assumption that recovered agents acquire lifelong immunity. (We briefly discuss in the conclusion how this assumption can be relaxed.) 
Again the timing convention is that if an $I$ agent is removed at time $t$, the agent changes status to $R_k$, $R_u$, or $D$ at time $t+\Delta$.

In epidemiology, there are several notions of fatality rate, and it is important to understand the distinction. The fatality rate among all (known and unknown) infected cases (which corresponds to $\delta_0$) is called the \emph{infection fatality rate (IFR)}. The fatality rate among known (confirmed) infected cases (which corresponds to $\delta$ if the signal is a laboratory test) is called the \emph{case fatality rate (CFR)}. The fatality rate among the entire population is called \emph{mortality}. Clearly, by definition we have $\text{Mortality}\le \text{IFR}\le \text{CFR}$.

\paragraph{Vaccine arrival}

We assume that agents expect a vaccine to arrive at a Poisson rate $\nu\ge 0$, independent of everything else. Thus in our discrete-time setting, the probability that a vaccine arrives between time $t$ and $t+\Delta$ is $1-\e^{-\nu\Delta}$. We assume that the vaccine is perfectly effective, perfectly safe, and has no cost. Thus once a vaccine arrives, all non-infected agents will be vaccinated and become immune ($R_k$). The vaccine is not a cure and hence has no effect on infected agents.

\subsection{Assumptions}

Throughout the rest of the paper, we maintain the following assumptions.

\begin{asmp}[Utility function]\label{asmp:u}
The utility functions satisfy the following conditions:
\begin{enumerate*}
\item $u:A=[\ubar{a},1]\to \R$ is twice continuously differentiable and satisfies $u(1)=0$, $u'>0$, and $u''<0$,
\item $u_I:A=[\ubar{a},1]\to \R$ is continuous, strictly concave, and achieves a unique maximum at $a_I\in A$, and
\item $u_D<u_I(a_I)\le u(1)$.
\end{enumerate*}
\end{asmp}

The assumption $u(1)=0$ simplifies the algebra and is without loss of generality because we can shift the utility functions by a constant without affecting behavior. The assumptions $u_D<u_I(a_I)\le u(1)$ simply imply that being asymptomatic is preferable to being symptomatic, which is in turn preferable to being dead. The condition that $u_I$ is single-peaked at $a_I\in A$ implies that a potentially intermediate value of activity level (rest) is myopically optimal for symptomatic agents. This assumption can also be interpreted as altruism, sense of duty, or an enforcement of a quarantine policy.

\begin{asmp}[Perfect competition]\label{asmp:comp}
Agents view the evolution of the aggregate state $z$ as exogenous and ignore the impact of their behavior on the aggregate state.
\end{asmp}

Assumption \ref{asmp:comp} is made for analytical tractability and is reasonable when the number of agents $N$ is large.

\begin{asmp}[Consistency]\label{asmp:consistent}

On equilibrium paths, agents update their beliefs using Bayes' rule. Off equilibrium paths, $U$ (unknown) agents believe they are susceptible with probability
\begin{equation}
\mu(z)\coloneqq 
\begin{cases*}
\frac{S}{S+I_u+R_u} & if $S>0$, \\
0 & otherwise.
\end{cases*}
\label{eq:muz}
\end{equation}
\end{asmp}

The assumption that agents apply Bayes' rule may not be realistic because a pandemic such as COVID-19 is rare and agents may have difficulty forming beliefs when faced with an unprecedented situation. However, we focus on Bayes' rule because it provides a useful benchmark. The assumption that we specify the off-equilibrium beliefs as in \eqref{eq:muz} anticipates our choice of perfect Bayesian equilibrium (PBE) as our solution concept. We discuss this choice (and possible alternatives) in more detail in Section \ref{subsec:eq} and Appendix \ref{CE2new}.

\section{Equilibrium analysis}\label{sec:eq}

This section defines and establishes the existence of equilibrium and characterizes individual behavior. We first characterize the individual best response in Section \ref{subsec:individual}, and then impose equilibrium conditions in Section \ref{subsec:eq}.

\subsection{Individual best response}\label{subsec:individual}

We first analyze the individual optimization problems recursively. Let $z\in Z$ be the aggregate state and $V_h$ be the value function of type $h\in \set{U,I_k,R_k,D}$.

\paragraph{Dead agents}

Because $D$ agents remain dead and their flow utility is $u_D$, their value function is constant and satisfies
\begin{equation*}
V_D=(1-\e^{-r\Delta})u_D+\e^{-r\Delta}V_D\iff V_D=u_D.
\end{equation*}

\paragraph{Known recovered agents}

Because $R_k$ agents have lifelong immunity, their value function is constant and the associated Bellman equation is
\begin{equation*}
V_{R_k}=\max_{a\in A}\set{(1-\e^{-r\Delta})u(a)+\e^{-r\Delta}V_{R_k}}.
\end{equation*}
The optimal action is clearly $a_{R_k}=1$ (full activity) and the value function is $V_{R_k}=u(1)=0$ by Assumption \ref{asmp:u}.

\paragraph{Known infected agents}

By assumption, $I_k$ agents are removed with probability $\gamma\Delta$, and conditional on removal, die with probability $\delta=\delta_0/\sigma$. Since the health status transitions are independent of the aggregate state and action, their value function is constant and the associated Bellman equation is
\begin{equation}
V_{I_k} = \max_{a\in A} \set{(1-\e^{-r\Delta})u_I(a) +\e^{-r\Delta}(\underbrace{(1-\gamma\Delta)V_{I_k}}_\text{stay infected}+\underbrace{\gamma\Delta[(1-\delta) V_{R_k}+\delta V_D])}_\text{removal}}.\label{eq:Bellman.i}
\end{equation}
By Assumption \ref{asmp:u} the function $u_I$ is single-peaked at $a_I$, and so the optimal action is $a_{I_k}=a_I$. Since $V_{R_k}=u(1)=0$ and $V_D=u_D$, \eqref{eq:Bellman.i} simplifies to
\begin{equation}
V_{I_k}=\frac{(1-\e^{-r\Delta})u_I+\e^{-r\Delta}\gamma\Delta\delta u_D}{1-\e^{-r\Delta}(1-\gamma\Delta)},\label{eq:VIk}
\end{equation}
where $u_I\coloneqq u_I(a_I)$. Note by Assumption \ref{asmp:u} that we have $u_I=u_I(a_I)\le u(1)=0$, so $V_{I_k}<0=V_{R_k}$.

\paragraph{Unknown agents}

Because $U$ agents need to infer their health status, the analysis of their decision problem is more complicated. In principle, the state variables of the individual optimization problem are the aggregate state $z$ and the belief $\mu$ of being susceptible. However, because making the belief part of the state variable makes the model analytically intractable, for our benchmark analysis we suppose that $U$ agents' belief is simply inferred from the aggregate state and given by $\mu(z)$ in \eqref{eq:muz}. (As we shall see in Theorem \ref{thm:exist} below, this assumption is consistent with the Bayes rule on equilibrium path.) Under this assumption, we now consider an individual $U$ agent's best response when all other $U$ agents adhere to a policy function $a_U:Z\to A$.

The policy $a_U(z)$ together with the mechanisms of disease transmission, symptom development, recovery, and death generate transition probabilities $\set{q(z,z')}_{(z,z')\in Z^2}$ for the aggregate state conditional on no vaccine arrival. (Note that $Z$ in \eqref{eq:Z} is a finite set.) By Assumption \ref{asmp:comp}, agents view this law of motion as exogenous. Let $V_U(z)$ be the value function of a $U$ agent who chooses the action optimally in this environment. By \eqref{eq:infect_prob} and the analysis of $I_k$ agents, an agent taking full action ($a=1$) gets infected with probability
\begin{equation}
p(z) \coloneqq \beta\Delta (a_II_k+a_U(z)I_u)
\label{eq:pz}
\end{equation}
conditional on being susceptible. Noting that infection is revealed with probability $\sigma$ and the vaccine arrives with probability $1-\e^{-\nu\Delta}$, the Bellman equation for unknown agents is
\begin{multline}
V_U(z)=\max_{a\in A}\Biggl\{(1-\e^{-r\Delta})u(a)+\underbrace{\e^{-r\Delta}\sigma\mu pa V_{I_k}}_\text{known infection}\\
+\underbrace{\e^{-r\Delta}(1-\sigma\mu pa)}_\text{stay unknown}\E_z(\underbrace{\e^{-\nu\Delta}V_U(z')}_\text{no vaccine}+\underbrace{(1-\e^{-\nu\Delta})V_{R_k}}_\text{vaccine})\Biggr\},\label{eq:Bellman.u}
\end{multline}
where $\E_z$ denotes the expectation with respect to $\set{q(z,z')}$, $\mu=\mu(z)$ is given by \eqref{eq:muz}, $p=p(z)$ is given by \eqref{eq:pz}, and $V_{I_k}$ is given by \eqref{eq:VIk}. Noting that $V_{R_k}=u(1)=0$, \eqref{eq:Bellman.u} simplifies to
\begin{equation}
V_U(z) = \max_{a\in A}\set{(1-\e^{-r\Delta})u(a) + \e^{-r\Delta}\E_z((1-\sigma\mu pa)\e^{-\nu\Delta}V_U(z')+\sigma\mu pa V_{I_k})}.
\label{eq:Bellman.u2}
\end{equation}

The following proposition establishes the existence and uniqueness of $V_U$ and provides some bounds on value functions.

\begin{prop}[Value functions]\label{prop:V}
Fix a policy function $a_U:Z\to A$ of unknown agents. Then there exists a unique value function $V_U:Z\to \R$ satisfying the Bellman equation \eqref{eq:Bellman.u2}. Furthermore, the value functions satisfy the following inequalities:
\begin{align}
V_D=u_D&<\frac{(\e^{r\Delta}-1)u_I+\gamma\Delta\delta u_D}{\e^{r\Delta}-1+\gamma\Delta}=V_{I_k}\notag \\
&<\frac{\e^{\nu\Delta}\sigma\beta\Delta}{\e^{(r+\nu)\Delta}-1+\sigma\beta\Delta}V_{I_k}\le V_U(z)\le V_{R_k}=0.\label{eq:Vineq}
\end{align}
\end{prop}

The proof of Proposition \ref{prop:V} as well as other longer proofs are deferred to Appendix \ref{sec:proof}. The inequality \eqref{eq:Vineq} is quite intuitive. In terms of flow utility, having no symptoms is better than having symptoms, which is better than death. Because the states $R_{k},D$ are absorbing and an $I_k$ agent may recover or die, the inequalities $V_D<V_{I_k}<V_{R_k}=0$ are immediate. The inequality $V_U\le V_{R_k}$ is also immediate because a $U$ agent could get infected and generally chooses a lower action. The inequality $V_{I_k}<V_U$ follows from the fact that a $U$ agent can always choose the myopic optimal action ($a=1$, which generates flow utility $0=u(1)\ge u_I(a_I)=u_I$) but gets infected only in the future.

The following proposition characterizes the best response of a $U$ agent. To state the result, we define the inverse marginal utility function $\phi:\R\to [\ubar{a},1]$ by
\begin{equation}
\phi(x)\coloneqq\begin{cases*}
1 & if $x\le u'(1)$,\\
(u')^{-1}(x) & if $u'(1)<x<u'(\ubar{a})$,\\
\ubar{a} & if $x\ge u'(\ubar{a})$.
\end{cases*}
\label{eq:phi}
\end{equation}

\begin{prop}[Best response of $U$ agents]\label{prop:aU}
Fix a policy function $a_U:Z\to A$ and 
let $V_U:Z\to \R$ be the corresponding value function established in Proposition \ref{prop:V}. Then the best response of a $U$ agent is
\begin{equation}
a^*=\phi\left(\frac{\sigma\mu(z)p(z)}{\e^{r\Delta}-1}(\E_z \e^{-\nu\Delta}V_U(z')-V_{I_k})\right),\label{eq:aU}
\end{equation}
where $\mu(z)$ and $p(z)$ are given by \eqref{eq:muz} and \eqref{eq:pz}, respectively.
\end{prop}

The next corollary shows that when prevalence is sufficiently low, agents take no precautions.

\begin{cor}[Full activity with sufficiently low prevalence]\label{cor:aU1}
Let $I\coloneqq I_k+I_u$ be the prevalence defined in \eqref{eq:prevalence}. There exists $\bar{I}>0$ such that for all policy function $a_U(z)$, we have $a^*=1$ whenever $I<\bar{I}$. In particular, we can take
\begin{equation}
\bar{I}=-\frac{\e^{r\Delta}-1}{\Delta}\frac{u'(1)}{\sigma\beta V_{I_k}},\label{eq:Ibar}
\end{equation}
where $V_{I_k}<0$ is given by \eqref{eq:VIk}.
\end{cor}

\subsection{Definition and existence of equilibrium}\label{subsec:eq}
Our equilibrium concept is the (pure strategy) perfect Bayesian Markov competitive equilibrium defined below. Here ``perfect Bayesian'' means that agents update beliefs on equilibrium paths using Bayes' rule as in Assumption \ref{asmp:consistent}; ``Markov'' means that the optimal actions agents choose are functions of state variables; and ``competitive'' means that agents view the evolution of aggregate state variables as exogenous as in Assumption \ref{asmp:comp}.

\begin{defn}[Markov equilibrium]\label{defn:eq}
A (pure strategy) \emph{perfect Bayesian Markov competitive equilibrium}, or \emph{Markov equilibrium} for short, consists of $U$ agents' belief $\mu(z)$ of being susceptible, transition probabilities $\set{q(z,z')}_{z,z'\in Z}$ for the aggregate state, value functions $\set{V_h(z)}_{h=U,I_k,R_k,D}$, and policy functions $\set{a_h(z)}_{h=U,I_k,R_k}$ such that
\begin{enumerate}
\item (Consistency) The belief $\mu(z)$ satisfies Bayes' rule on equilibrium paths and is given by \eqref{eq:muz} off equilibrium paths; the transition probabilities $\set{q(z,z')}$ are consistent with individual actions and the mechanisms of disease transmission, symptom development, recovery, and death,
\item (Sequential rationality)
\begin{enumerate}
\item[($U$)] $V_U(z)$ satisfies the Bellman equation \eqref{eq:Bellman.u2} and $a=a_U(z)$ achieves the maximum, where $\E_z$ denotes the conditional expectation using $\set{q(z,z')}$ and $p=p(z)$ is as in \eqref{eq:pz},
\item[($I_k$)] $V_{I_k}(z)$ is as in \eqref{eq:VIk} and $a_{I_k}(z)=a_I$,
\item[($R_k$)] $V_{R_k}(z)=0$ and $a_{R_k}(z)=1$,
\item[($D$)] $V_D(z)=u_D$.
\end{enumerate}
\end{enumerate}
\end{defn}

Note that Definition \ref{defn:eq} only describes the society before vaccine arrival. Once the vaccine arrives, because no new infection occur by assumption, it is optimal for all agents to take their myopic optimal action ($a=1$ for $h=U,R_k$ and $a=a_I$ for $h=I_k$) and the problem becomes trivial.

The astute reader may wonder why we adopt the notion of perfect Bayesian equilibrium and specify that beliefs are given by \eqref{eq:muz} even off the equilibrium path. The belief $\mu(z)$ equals the posterior belief if $U$ agents have a common prior, take identical actions, and learn that the aggregate state is $z$, but a Bayesian $U$ agent who is contemplating deviating from the equilibrium action $a_U(z)$ generally has a different posterior. The primary reason for our choice of using PBE as the solution concept is that of tractability, as we wish to study the role of forward-looking agents that are uncertain about their health status without obscuring the analysis with technicalities.\footnote{To be more specific, in our discrete-time setting, insisting on applying Bayes' rule off the equilibrium path leads to non-concave maximization problems, and hence the existence of equilibrium is not assured.} The above notion of perfect Bayesian equilibria allows agents in our model to be forward-looking and rational, if somewhat ``forgetful''. In Appendix \ref{CE2new} we extend the model to the case with perfect recall and show that the qualitative features of the equilibrium remain unchanged, with the main difference that activity is higher with perfect recall.

The following theorem establishes the existence of a perfect Bayesian equilibrium. 

\begin{thm}[Existence of equilibrium]\label{thm:exist}
Suppose Assumptions \ref{asmp:u}--\ref{asmp:consistent} hold. Then there exists a pure strategy perfect Bayesian Markov competitive equilibrium, where the belief $\mu(z)$ always satisfies \eqref{eq:muz}.
\end{thm}

As we noted in the introduction, one benefit of the recursive nature of our analysis is that it allows us to study the role of equilibrium and efficient actions at every point in the state space. This is crucial for our analysis of the role of enforcement and for 
the distinction we later draw between static efficient and efficient activity levels.

\subsection{Externalities and efficiency}
\label{subsec:external}


In this section we study the source of externalities in the model and the efficiency properties of the equilibrium. Let $a_U(z)$ and $a_{I_k}(z)$ be any policy functions for $U$ and $I_k$ agents chosen by the social planner. Since the behavior of $R_k$ agents does not affect the aggregate state dynamics, without loss of generality we set $a_{R_k}(z)=1$, as this is both individually and socially optimal. Suppose the planner wishes to choose $(a_U(z),a_{I_k}(z))$ to maximize the social welfare. In general, individual agents have an incentive to deviate from such recommendations and choose the individually optimal actions characterized by Proposition \ref{prop:aU}. In order to model imperfect enforcement and therefore study the distinction between static and dynamic externalities, we suppose that at some exogenous hazard rate $\eta\ge 0$ society reverts back to the Markov equilibrium characterized in Theorem \ref{thm:exist}. Letting $V_h^\eta(z;a_U,a_I)$ be the value function of $h$ agents in this environment, since the probability of reverting to equilibrium is $1-\e^{-\eta\Delta}$, we have
\begin{subequations}\label{eq:Vlambda}
\begin{align}
V_U^\eta(z) & =(1-\e^{-\eta\Delta})V_U(z)+\e^{-\eta\Delta}\bigl[(1-\e^{-r\Delta})u(a_U)\notag \\
& +\e^{-r\Delta}\E_z((1-\sigma\mu pa_U)\e^{-\nu\Delta}V_U^\eta(z')+\sigma\mu pa_U V_{I_k}^\eta(z'))\bigr],\label{eq:Vlambda.u}\\
V_{I_k}^\eta(z) & = (1-\e^{-\eta\Delta})V_{I_k}+\e^{-\eta\Delta}\bigl[(1-\e^{-r\Delta})u_I(a_{I_k})\notag \\
& +\e^{-r\Delta}\E_z((1-\gamma\Delta)V_{I_k}^\eta(z')+\gamma\Delta((1-\delta)V_{R_k}+\delta V_D))\bigr],\label{eq:Vlambda.i}
\end{align}
\end{subequations}
where $V_U(z)$ and $V_{I_k}$ are the equilibrium value functions in Theorem \ref{thm:exist}, $\mu=\mu(z)$ and $p=p(z)$ are given by \eqref{eq:muz} and \eqref{eq:pz}, and we write $a_h=a_h(z)$ and $V_h^\eta(z)=V_h^\eta(z;a_U,a_{I_k})$ for brevity. By the standard contraction mapping argument, $V_h^\eta$ is well-defined. The utilitarian social welfare associated with the policies $(a_U(z),a_{I_k}(z))$ is then
\begin{align}
\cW^\eta(z)&\coloneqq (S+I_u+R_u)V_U^\eta(z)+I_kV_{I_k}^\eta(z)+R_kV_{R_k}(z)+DV_D(z)\notag \\
&=(S+I_u+R_u)V_U^\eta(z)+I_kV_{I_k}^\eta(z)+Du_D, \label{eq:wz}
\end{align}
where we have used $V_{R_k}=0$ and $V_D=u_D$. Noting that $u_D$ is constant and $D$ depends only on $z$ (and not on the policies $(a_U(z),a_{I_k}(z))$), conditional on the aggregate state $z$, we can rank social welfare associated with particular policy functions by using
\begin{equation}
W^\eta(z)\coloneqq UV_U^\eta(z)+I_kV_{I_k}^\eta(z) \label{eq:welfare}
\end{equation}
instead of $\cW^\eta$, where $U\coloneqq S+I_u+R_u$ denotes the fraction of $U$ agents.

Given the hazard rate $\eta$, the optimal policies $(a_U(z),a_{I_k}(z))$ maximize \eqref{eq:welfare}. The Markov equilibrium in Definition \ref{defn:eq} is generally inefficient because the equilibrium policies $(a_U^*(z),a_{I_k}^*(z))$ do not maximize the welfare criterion \eqref{eq:welfare} due to externalities. We distinguish between two types of externalities, static and dynamic. The static externality refers to the fact that when $I$ (infected) agents take higher actions, they infect $S$ (susceptible) agents with some probability and affect the \emph{current} value function of $U$ agents, even if the future value functions and the aggregate state transitions remain the same. The dynamic externality refers to the fact that the collective behavior of agents affect the aggregate state transitions and hence future value functions. The difference between the equilibrium and efficient activity level may be interpreted as a measure of the externality in this environment. In this section, we study the static externality analytically by comparing the efficient activity level with the level chosen by a planner who can intervene for only a vanishingly small time interval. The difference between these two activity levels may be interpreted as a measure of the static externality, while the remaining difference between the static efficient and equilibrium activity may be interpreted as a measure of the dynamic externality (which we explore in our numerical exercises in Section \ref{sec:num}).

To this end, consider a social planner who can intervene to alter agents' current actions $(a_U,a_{I_k})$ but who takes as given the transition probabilities $\set{q(z,z')}$ of the aggregate state as well as next period's value functions $\set{V_h(z')}_{h=U,I_k,R_k,D}$. This is  equivalent to solving the above problem with $\eta=\infty$ and $\Delta\to 0$. Using \eqref{eq:Bellman.i}, \eqref{eq:Bellman.u2}, \eqref{eq:welfare}, and noting that $V_{R_k}=0$, we can define the objective function of the planner that seeks to eliminate the static externality by
\begin{multline}
W(a_U,a_{I_k},z) \coloneqq \\ U\left[(1-\e^{-r\Delta})u(a_U) + \e^{-r\Delta}\E_z((1-\sigma\mu pa_U)\e^{-\nu\Delta}V_U(z')+\sigma\mu pa_U V_{I_k}(z'))\right] \\
+I_k\left[(1-\e^{-r\Delta})u_I(a_{I_k}) + \e^{-r\Delta}\E_z((1-\gamma\Delta)V_{I_k}(z')+\gamma\Delta\delta V_D)\right], \label{eq:Waz}
\end{multline}
where $p$ is given by \eqref{eq:pz}. Note that while the continuation value in the individual optimization problem in \eqref{eq:Bellman.u} is linear in the agent's action $a$, the continuation value in the planner's problem \eqref{eq:Waz} is quadratic in the $U$ agent action $a_U$ because $p$ is affine in $a_U$. This interaction is precisely the source of static externality.

Restricting the action of $I_k$ agents can be interpreted as \emph{quarantine}. Restricting the action of $U$ agents can be interpreted as \emph{lockdown}. Hence we introduce the following definition.

\begin{defn}[Static efficient actions]\label{defn:se}
Fix transition probabilities $\set{q(z,z')}$ and value functions $V_U,V_{I_k}$. Given a policy function $a_U(z)$ of $U$ agents, we say that the \emph{quarantine policy} $a_{I_k}^\dagger(z)$ is \emph{static efficient} if $a_{I_k}=a_{I_k}^\dagger(z)$ achieves the maximum of \eqref{eq:Waz}. Similarly, given a policy function $a_{I_k}(z)$ of $I_k$ agents, we say that the \emph{lockdown policy} $a_U^\dagger(z)$ is \emph{static efficient} if $a_U=a_U^\dagger(z)$ achieves the maximum of \eqref{eq:Waz}.
\end{defn}


The reason we use the qualifier ``static'' can be understood as follows. Imagine that a benevolent government wishes to implement some lockdown policy (the choice of a specific $a_U(z)$) when faced with an epidemic. The optimal policy will then obviously depend on the duration of the intervention. If the intervention is implemented for an extended period of time (\ie, $\eta$ is small), the future aggregate state $z$ will change and so too will the continuation values $\E_z V_h(z')$. The policies in Definition \ref{defn:se} are static efficient in the sense that they are what the government would choose if they could commit only for a short period of time and thus cannot influence the future evolution of the aggregate state.

The following proposition characterizes the static efficient quarantine policy.

\begin{prop}[Static efficient quarantine policy]\label{prop:aI_se}
Suppose the utility function of $I_k$ agents $u_I$ is continuously differentiable, strictly concave, with inverse marginal utility function $\phi_I$ analogous to \eqref{eq:phi}. Then the static efficient quarantine policy is given by
\begin{equation}
a_{I_k}^\dagger(z)=\phi_I\left(\frac{\sigma\beta\Delta}{\e^{r\Delta}-1}\mu(z) a_U(z)U\E_z(\e^{-\nu\Delta}V_U(z')-V_{I_k}(z'))\right).\label{eq:aI_se}
\end{equation}
\end{prop}

\begin{proof}
Let $W$ be as in \eqref{eq:Waz}. Using $p=\beta\Delta(a_{I_k}I_k+a_UI_u)$, we obtain
\begin{align*}
\frac{\partial W}{\partial a_{I_k}}&=-U\e^{-r\Delta}\sigma\mu\beta\Delta I_k a_U\E_z(\e^{-\nu\Delta}V_U(z')-V_{I_k}(z'))+I_k(1-\e^{-r\Delta})u_I'(a_{I_k})\\
&=I_k\e^{-r\Delta}\Delta\left(\frac{\e^{r\Delta}-1}{\Delta}u_I'(a_{I_k})-\sigma\mu\beta a_UU\E_z(\e^{-\nu\Delta}V_U(z')-V_{I_k}(z'))\right).
\end{align*}
The rest of the proof is the same as Proposition \ref{prop:aU}.
\end{proof}

Proposition \ref{prop:aI_se} has several implications. First, \eqref{eq:aI_se} does not depend on the fraction of known infected agents $I_k$ except through the continuation values. This is because the welfare gain from reducing new infections and the welfare loss from restricting infected agents' actions are both proportional to $I_k$, as we can see from the proof of Proposition \ref{prop:aI_se}. Second, under normal circumstances the planner seeks to quarantine $I_k$ agents intensively. To see this, using \eqref{eq:pz}, note that the individually optimal action of $U$ agents \eqref{eq:aU} can be rewritten as
\begin{equation}
a_U^*=\phi\left(\frac{\sigma\beta\Delta}{\e^{r\Delta}-1}\mu(z) (a_{I_k}I_k+a_UI_u)\E_z(\e^{-\nu\Delta}V_U(z')-V_{I_k}(z'))\right).\label{eq:aU_alt}
\end{equation}
Assuming that $I_k$ and $U$ agents have identical preferences (so $u_I=u$), the only difference between \eqref{eq:aI_se} and \eqref{eq:aU_alt} is that the argument of the former is proportional to $a_UU$, whereas the argument of the latter is proportional to $a_{I_k}I_k+a_UI_u$. The fraction of actively infected agents in the society is typically small, so $I_k,I_u\ll U$. In this case, the argument of \eqref{eq:aI_se} is much larger than the argument of \eqref{eq:aU_alt}, and because $\phi$ is decreasing, we would have $a_{I_k}^\dagger\ll a_U^*$.

We next study the optimal intervention to $U$ agents. Suppose all $I_k$ agents take action $a_I\in A$, which can be interpreted as an exogenous quarantine policy by the discussion after Assumption \ref{asmp:u}. Let $a_U^*(z)$ be the equilibrium policy of $U$ agents and $V_U,V_{I_k}$ the corresponding value functions established in Theorem \ref{thm:exist}. The following theorem provides a bound for the static efficient lockdown policy, which is our main theoretical result.

\begin{thm}[Static efficient lockdown policy]\label{thm:aU_se}
Let $a_U^*(z)$ be the equilibrium policy of $U$ agents and
\begin{equation*}
0<m\coloneqq \min_{a\in A}\abs{u''(a)}\le \max_{a\in A}\abs{u''(a)}\eqqcolon M.
\end{equation*}
There exists a unique static efficient lockdown policy $a_U^\dagger(z)$, which satisfies
\begin{subequations}\label{eq:aU_bound}
\begin{align}
&\frac{a_U^*(z)}{2}\le a_U^\dagger(z)\le a_U^*(z),\label{eq:aU_bound1}\\
&a_U^*(z)-a_U^\dagger(z)\le \frac{\sigma\beta\Delta\mu(z)}{m(\e^{r\Delta}-1)}(2a_U^\dagger(z)-a_U^*(z))I_u\E_z(\e^{-\nu\Delta}V_U(z')-V_{I_k}).\label{eq:aU_bound2}
\end{align}
\end{subequations}
In particular, if $\sigma=1$ (so $I_u=0$), then $a_U^\dagger(z)=a_U^*(z)$.

Furthermore, whenever $a_U^*,a_U^\dagger$ are interior, we have
\begin{equation}
a_U^*(z)-a_U^\dagger(z)\ge \frac{\sigma\beta\Delta\mu(z)}{M(\e^{r\Delta}-1)}(2a_U^\dagger(z)-a_U^*(z))I_u\E_z(\e^{-\nu\Delta}V_U(z')-V_{I_k}).\label{eq:aU_bound3}
\end{equation}
\end{thm}

Although the mathematical statement of Theorem \ref{thm:aU_se} is relatively complicated, its main message is clear:
\begin{enumerate*}
\item the equilibrium action $a_U^*$ is too high (but at most twice as high) relative to the recommended action $a_U^\dagger$ of a planner that can intervene only for a short period, but
\item the difference in recommended actions $a_U^*-a_U^\dagger$ has the same order of magnitude as the fraction of unknown infected agents $I_u$.
\end{enumerate*}
Therefore a government that can only enforce short-term lockdown policies will always (weakly) wish to curb activity, but its incentive to do so vanishes with the fraction of unconfirmed cases. This is particularly noteworthy because we shall later provide examples in which a government with unlimited enforcement power would wish to do the \emph{reverse}, and recommend higher activity than that which occurs in equilibrium. It is in this sense that policy prescriptions can depend crucially on the enforcement capabilities of the government.

The proof of Theorem \ref{thm:aU_se} is long and technical but we briefly discuss its intuition. In the model, both $I_k$ and $I_u$ agents cause negative externality by infecting $S$ agents. The inequality $a_U^*\ge a_U^\dagger$ comes from the fact that $U$ agents are not internalizing the negative externality caused by $I_u$ agents: intuitively, individual agents view the utility loss from infection as a linear function of their own activity, while the planner views it as a quadratic function due to the interaction between $S$ and $I_u$ agents. The result that $a_U^*-a_U^\dagger$ has the same order of magnitude as $I_u$ comes from the fact that the objective functions of individuals and the planner differ only by the quadratic term due to the interaction between $S$ and $I_u$ agents, which is proportional to the fraction of unknown infected agents $I_u$. Because the first-order conditions differ only by a term proportional to $I_u$, we can show that the corresponding maximizers differ by a term proportional to $I_u$.

\subsection{Equilibrium dynamics}\label{subsec:dynamics}

This section studies the equilibrium dynamics. Given arbitrary policy functions $\set{a_h(z)}_{h=U,I_k,R_k}$, we define the transition probabilities $\set{q(z,z')}$ induced by these policies and the mechanisms of disease transmission, symptom development, recovery, and death. To simplify the notation, let $z_t=z$, $S_t=S$, $S_{t+\Delta}=S_{+\Delta}$, etc. Using \eqref{eq:infect_prob}, it is straightforward to show that
\begin{subequations}\label{eq:SIRD}
\begin{align}
\E_z(S_{+\Delta}-S)&=-\beta \Delta a_U(z)S(a_{I_k}(z)I_k+a_U(z)I_u), \label{eq:SIRD.s}\\
\E_z(I_{k,+\Delta}-I_k)&=\sigma \beta \Delta a_U(z)S(a_{I_k}(z)I_k+a_U(z)I_u)-\gamma\Delta I_k, \label{eq:SIRD.ik}\\
\E_z(I_{u,+\Delta}-I_u)&=(1-\sigma)\beta \Delta a_U(z)S(a_{I_k}(z)I_k+a_U(z)I_u)-\gamma\Delta I_u. \label{eq:SIRD.iu}
\end{align}
\end{subequations}
(We omit the dynamics for $R_k,R_u,D$ agents because they do not depend on policy functions.) We can simplify these equations further if we consider the limit $N\to\infty$ and apply the strong law of large numbers. In the large population limit, letting $I=I_k+I_u$ be the fraction of infected agents, we have $I_k=\sigma I$ and $I_u=(1-\sigma)I$. Therefore, by adding \eqref{eq:SIRD.ik} and \eqref{eq:SIRD.iu} we obtain the system of deterministic difference equations
\begin{subequations}\label{eq:SI}
\begin{align}
\frac{S_{+\Delta}-S}{\Delta}&=-\beta a_U(z)S(\sigma a_{I_k}(z)+(1-\sigma)a_U(z))I, \label{eq:SI.s}\\
\frac{I_{+\Delta}-I}{\Delta}&=(\beta a_U(z)S(\sigma a_{I_k}(z)+(1-\sigma)a_U(z))-\gamma)I =\gamma(\cR(z)-1)I, \label{eq:SI.i}
\end{align}
\end{subequations}
where
\begin{equation}
\cR(z)\coloneqq \frac{\beta}{\gamma}a_U(z)S(\sigma a_{I_k}(z)+(1-\sigma)a_U(z)) \label{eq:Rz}
\end{equation}
is known as the \emph{effective reproduction number} in epidemiology. At the early stage of the epidemic, by definition we have $S\approx 1$ and $I\approx 0$. Therefore by Corollary \ref{cor:aU1} (and assuming $u_I(a)=u(a)$), the equilibrium policies are $a_U(z)=a_{I_k}(z)=1$, and \eqref{eq:Rz} reduces to
\begin{equation}
\cR(z)=\beta/\gamma\eqqcolon \cR_0, \label{eq:R0}
\end{equation}
which is known as the \emph{basic reproduction number}.\footnote{\cite{Delamater2019} argue that ``$\cR_0$ is a function of human social behavior and organization, as well as the innate biological characteristics of particular pathogens.'' In our context $\cR_0$ is a constant because we do not explicitly model precautionary measures such as mask-wearing and hand-washing.} When $U$ agents are myopic (so $a_U(z)=a_{I_k}(z)=1$), \eqref{eq:SI} reduces to
\begin{subequations}\label{eq:KM}
\begin{align}
\frac{S_{+\Delta}-S}{\Delta}&=-\beta SI,\label{eq:KM.s}\\
\frac{I_{+\Delta}-I}{\Delta}&=(\beta S-\gamma)I,\label{eq:KM.i}
\end{align}
\end{subequations}
which is the usual \cite{KermackMcKendrick1927} Susceptible-Infected-Recovered (SIR) model in discrete time. Since $a_U(z),a_{I_k}(z)\le 1$, by \eqref{eq:SI.i}, \eqref{eq:Rz}, and \eqref{eq:R0}, we always have
\begin{equation*}
\frac{I_{+\Delta}-I}{\Delta}\le \gamma(\cR_0S-1)I.
\end{equation*}
Since by \eqref{eq:SI.s} the fraction of susceptible agents $S$ always decreases over time, the fraction of infected agents $I$ decreases over time once
\begin{equation}
\cR_0S\le 1\iff S\le 1/\cR_0 \label{eq:herd_immunity}
\end{equation}
holds, where $\cR_0$ is the basic reproduction number in \eqref{eq:R0}. We say that the society has achieved \emph{herd immunity} if \eqref{eq:herd_immunity} holds.

\section{Numerical analysis}\label{sec:num}

In this section we use a numerical example calibrated to the COVID-19 epidemic to study how the agents' optimizing behavior, diagnosis rate, and lockdown policies affect the epidemic dynamics.

\subsection{Model specification}
\label{subsec:num.spec}


One period corresponds to a day and we assume 5\% annual discounting, so $r=0.05/365.25$ and $\Delta=1$. Because 
it was reasonable to expect COVID-19 vaccines to be developed within a few years, we set the vaccine arrival rate to $\nu=1/365.25$. Unless otherwise stated, we suppose that the social planner has unlimited enforcement power ($\eta=0$ in \eqref{eq:Vlambda}--\eqref{eq:welfare}).

We set the epidemic parameters from the medical literature. The daily transmission rate is $\beta=1/5.4$ based on the median serial interval (the number of days it takes for an infected person to transmit the disease to another person) in the meta-analysis of \cite{RaiShuklaDwivedi2021}. The daily recovery rate is $\gamma=1/13.5$ based on the median infectious period estimated in \cite{You2020}. These numbers imply that the basic reproduction number in \eqref{eq:R0} is $\cR_0=\beta/\gamma=13.5/5.4=2.5$, which equals the current best estimate used by Centers for Disease Control and Prevention (CDC).\footnote{\label{fn:CDC}\url{https://www.cdc.gov/coronavirus/2019-ncov/hcp/planning-scenarios.html}} The infection fatality rate (IFR) is $\delta_0=\IFR$, which is the median value in the meta-analysis of \cite{Ioannidis2021} on studies that use seroprevalence data.\footnote{Unlike the case fatality rate (CFR), which is defined by the reported number of deaths divided by the reported number of cases, the estimation of IFR is complicated by the fact that cases and deaths may be underreported. The true number of cases can be estimated from a random antibody testing, which is called a seroprevalence survey. If we assume underreporting in deaths is not severe, then IFR can be estimated by dividing the reported number of deaths by the estimated number of cases.}

The choice of the diagnosis rate $\sigma$ is more controversial. One possibility is to estimate the case fatality rate (CFR) $\delta$ and set $\sigma=\delta_0/\delta$ based on \eqref{eq:deltasigma}. Using the data on cumulative number of reported cases and deaths provided by Johns Hopkins University Center for Systems Science and Engineering,\footnote{\label{fn:CSSE}\url{https://github.com/CSSEGISandData/COVID-19}} we find that the median CFR across all (more than 200) regions is $\delta=\CFR$. This would imply $\sigma=\delta_0/\delta=0.2$ (20\%). However, this calculation ignores other information such as the presence of symptoms. Another possibility is to set $\sigma$ as the fraction of symptomatic agents. Based on the case study of the cruise ship Diamond Princess, which experienced a COVID-19 outbreak in February 2020 and whose passengers were all tested, \cite{Mizumoto2020} document that about 50\% of confirmed cases were asymptomatic. Noting that the symptoms of COVID-19 are often similar to other upper respiratory infections and not specific enough to confirm the diagnosis, as a compromise, we choose an intermediate value $\sigma=0.4$ for the diagnosis rate in our baseline analysis. Subsequent to our benchmark calculations, we study the consequences of varying $\sigma$, and find that such variation has only a small impact on aggregate welfare. 


We set the minimum action to $\ubar{a}=0.01$ (which is somewhat arbitrary but never binds in simulations). The utility function of a symptom-free agent exhibits constant relative risk aversion (CRRA) $\alpha>0$, so
\begin{equation}
u(a)=\begin{cases}
\frac{a^{1-\alpha}-1}{1-\alpha}, & (0<\alpha\neq 1)\\
\log a, & (\alpha=1)
\end{cases}\label{eq:CRRA}
\end{equation}
which satisfies Assumption \ref{asmp:u}. We suppose $u_I=u$, so the optimal action for known infected agents is $a_I=1$. Although this assumption is unrealistic because infected agents may be incapacitated, altruistic, or quarantined, it provides the most conservative (worst case) analysis. For numerical illustrations, we set $\alpha=1$ (log utility). We calibrate the flow utility from death to $u_D=\uD$ based on the case study from Sweden, which did not introduce mandatory lockdowns (see Appendix \ref{sec:uD} for details). Finally, we set the initial condition to $I_0=10^{-6}$, $S_0=1-I_0$, $R_0=0$, and $D_0=0$.

\subsection{Equilibrium, static efficient, and efficient actions}\label{subsec:num.eq}

We solve for the perfect Bayesian Markov competitive equilibrium using the algorithm discussed in Appendix \ref{sec:solve}. As a point of comparison, we also solve for the myopic allocation in which all agents choose $a=1$, as well as the static efficient and efficient actions, which correspond to setting $\eta=\infty$ and $\eta=0$ in \eqref{eq:Vlambda}--\eqref{eq:welfare}, respectively. Figure \ref{fig:SIRD} shows the epidemic dynamics studied in Section \ref{subsec:dynamics} for the myopic equilibrium, which is the standard \cite{KermackMcKendrick1927} SIR model. Here and elsewhere, the vertical dashed line indicates the first date of achieving herd immunity as defined by \eqref{eq:herd_immunity}.

\begin{figure}[!htb]
\centering
\includegraphics[width=0.7\linewidth]{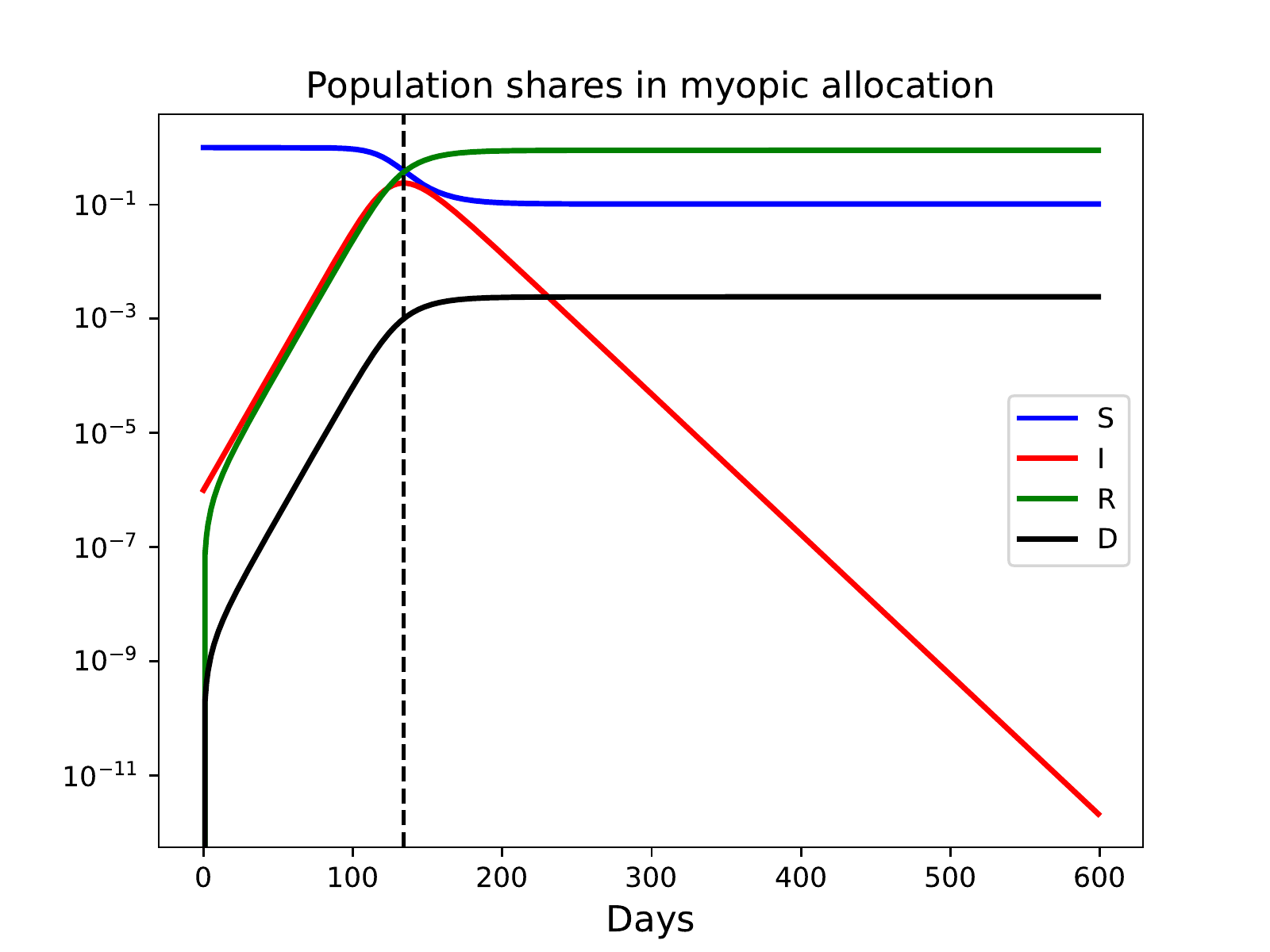}
\caption{Epidemic dynamics in myopic equilibrium.}\label{fig:SIRD}
\end{figure}

When agents are myopic and choose $a=1$, as is well known, the epidemic has two phases: in the first phase, the fraction of infected agents initially grows exponentially (at daily rate approximately $\beta-\gamma=\gamma(\cR_0-1)$ by setting $S=1$ in \eqref{eq:KM.i}) until the society achieves herd immunity (peak prevalence is \ImaxMy \%); in the second phase, the fraction of infected decays exponentially at daily rate $\gamma$. The epidemic dynamics significantly changes once we introduce optimizing behavior. Figure \ref{fig:eq} shows the epidemic dynamics (left panels) and contour plots of recommended actions over the state space (right panels). The top and bottom panels are for the Markov equilibrium and efficient allocations, respectively.

\begin{figure}[!htb]
\centering
\includegraphics[width=0.48\linewidth]{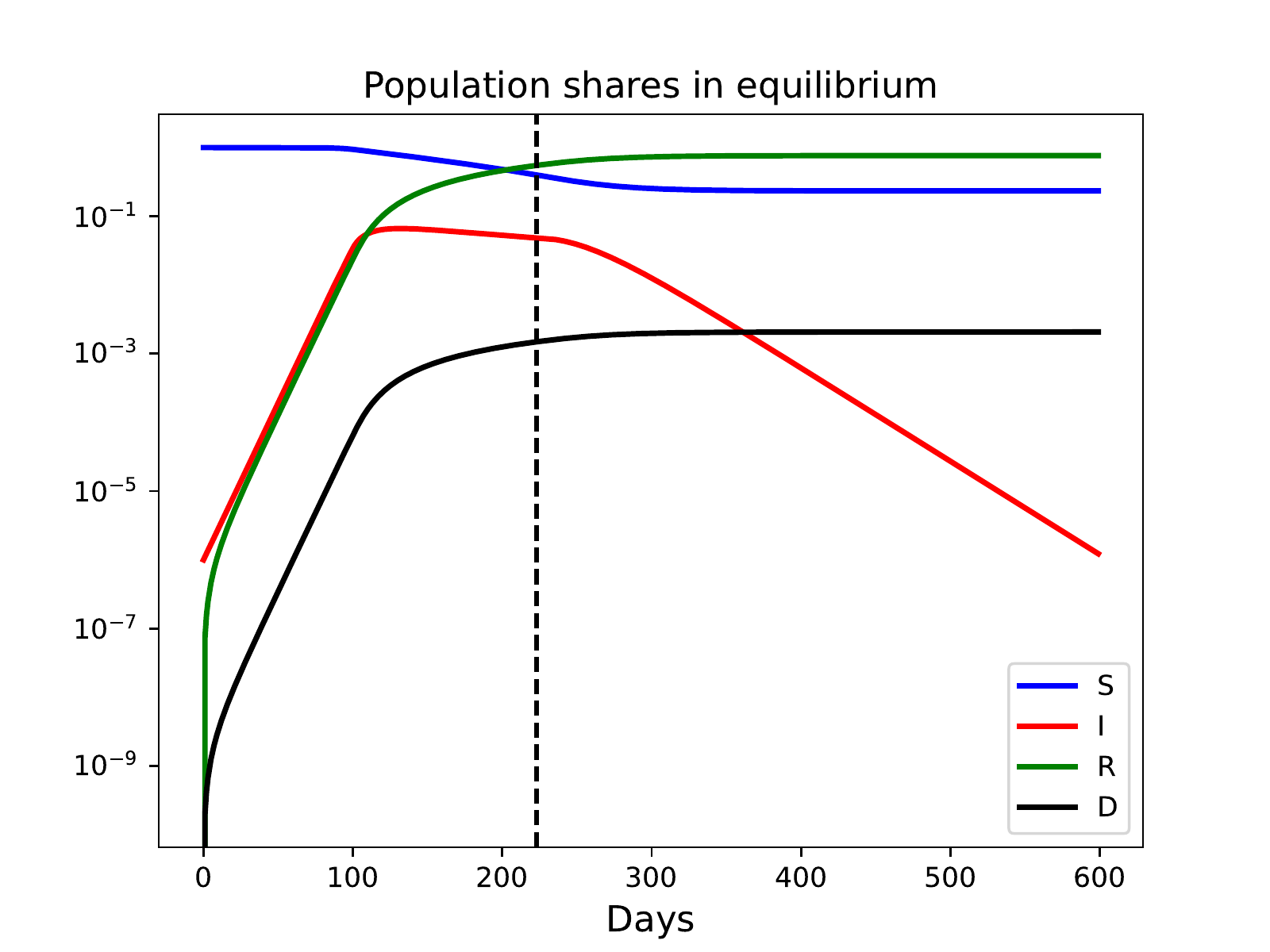}
\includegraphics[width=0.48\linewidth]{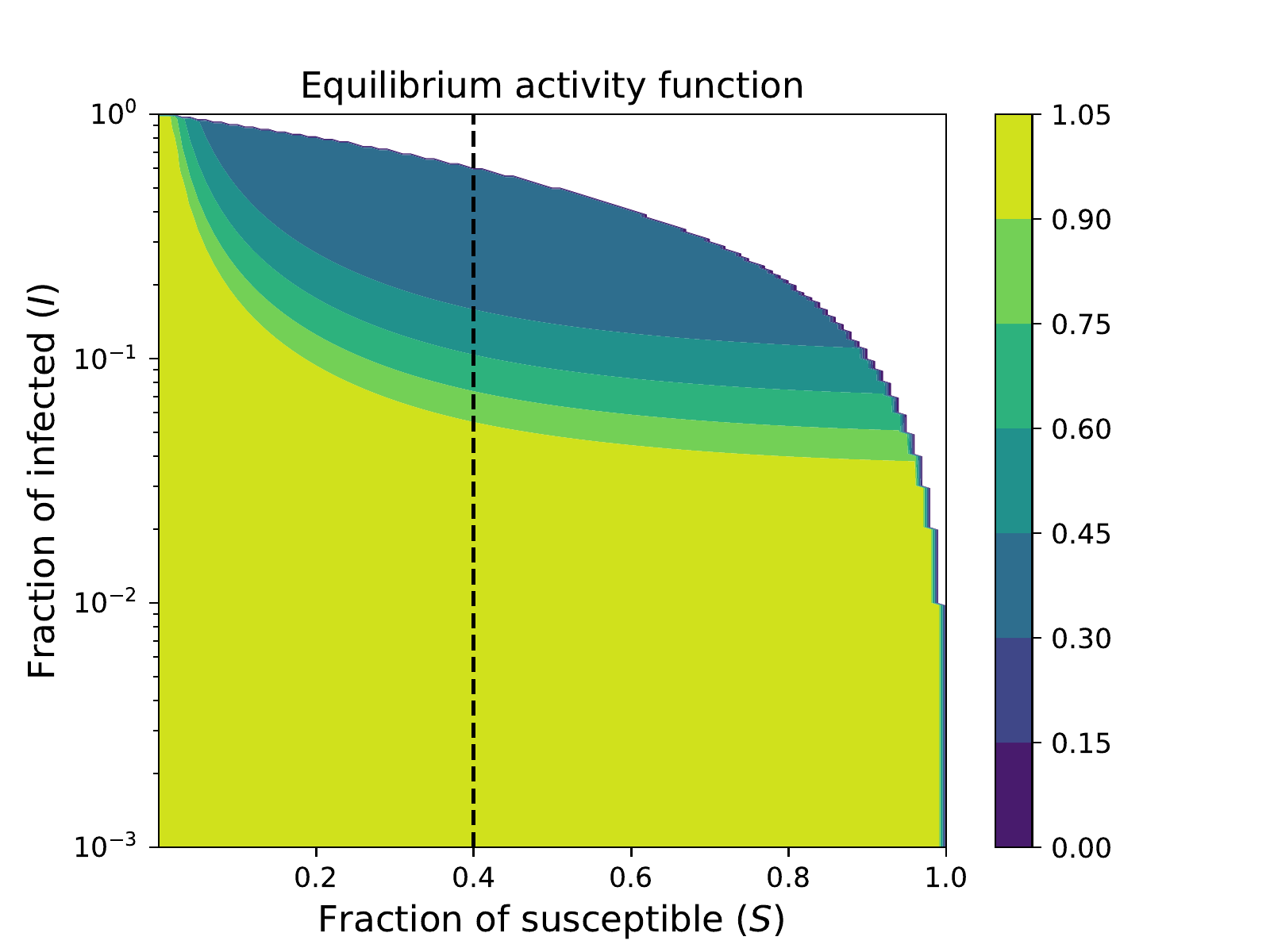}
\includegraphics[width=0.48\linewidth]{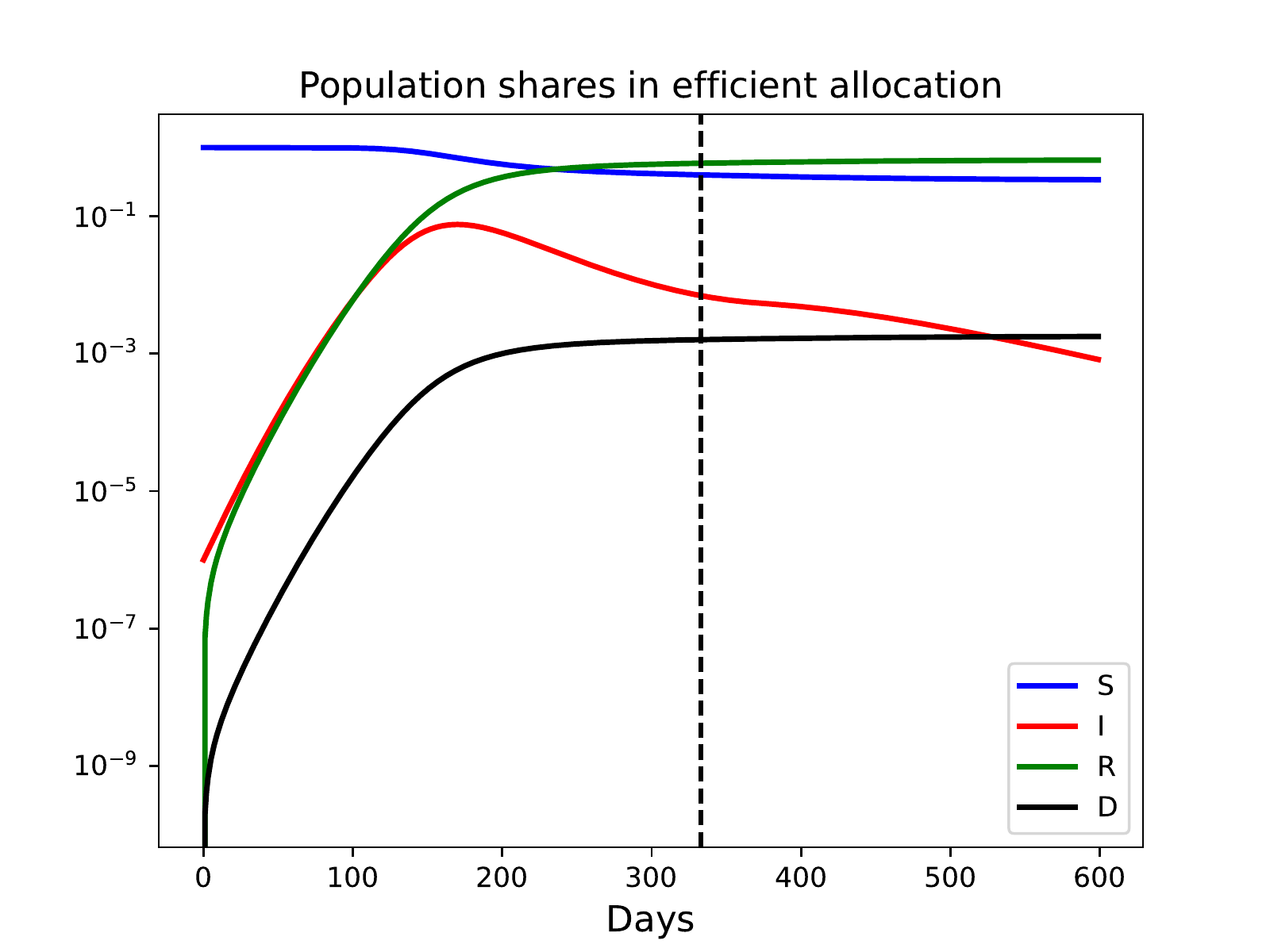}
\includegraphics[width=0.48\linewidth]{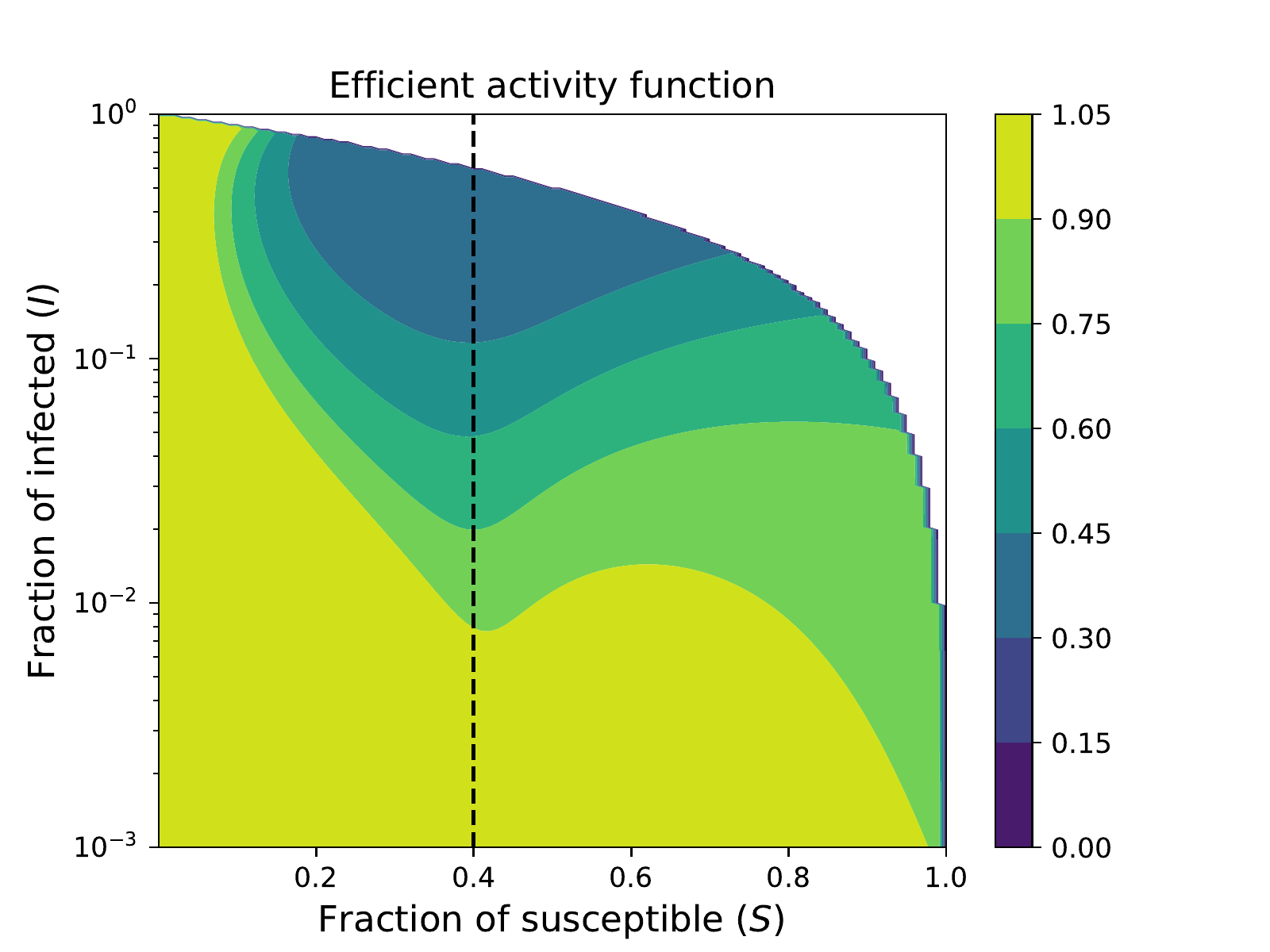}
\caption{Epidemic dynamics and recommended actions in equilibrium (top) and efficient (bottom) allocations.}\label{fig:eq}
\end{figure}

We now make a few observations regarding efficient and equilibrium allocations. First, when agents are forward-looking, the epidemic has three phases: during the first phase, the disease spreads freely and we see exponential growth (peak prevalence is \ImaxEq\%); during the second phase, $U$ agents voluntarily practice social distancing (set $a_U^*<1$) and the disease spread is endogenously mitigated; during the third phase, the society achieves herd immunity and the fraction of infected decays exponentially. Second, the epidemic dynamics and recommended action for the solution to the planner's problem (bottom panels) are qualitatively different from the Markov equilibrium (top panels). The planner who can enforce and commit to a lockdown policy finds it optimal to substantially reduce the action of unknown agents (bottom right), especially so when the population share of susceptible agents is high or the society is about to achieve herd immunity. In the resulting dynamics (bottom left), the fraction of infected agents both grows and declines more slowly, and it takes almost 50\% longer to reach herd immunity. Figure \ref{fig:aUT}  (left) plots the time paths of recommended actions (equilibrium and efficient) corresponding to the dynamics in the left panels of Figure \ref{fig:eq}. As expected, we see that the efficient action is almost everywhere below that of the equilibrium action, implying that the planner wishes to curb activity more quickly and for longer than do the individual agents. Further, compared with the equilibrium allocation, the reduction in activity in the efficient allocation begins earlier, is more gradual, and extends well beyond the date at which activity has returned to normal in the equilibrium allocation. 

\begin{figure}[!htb]
\centering
\includegraphics[width=0.48\linewidth]{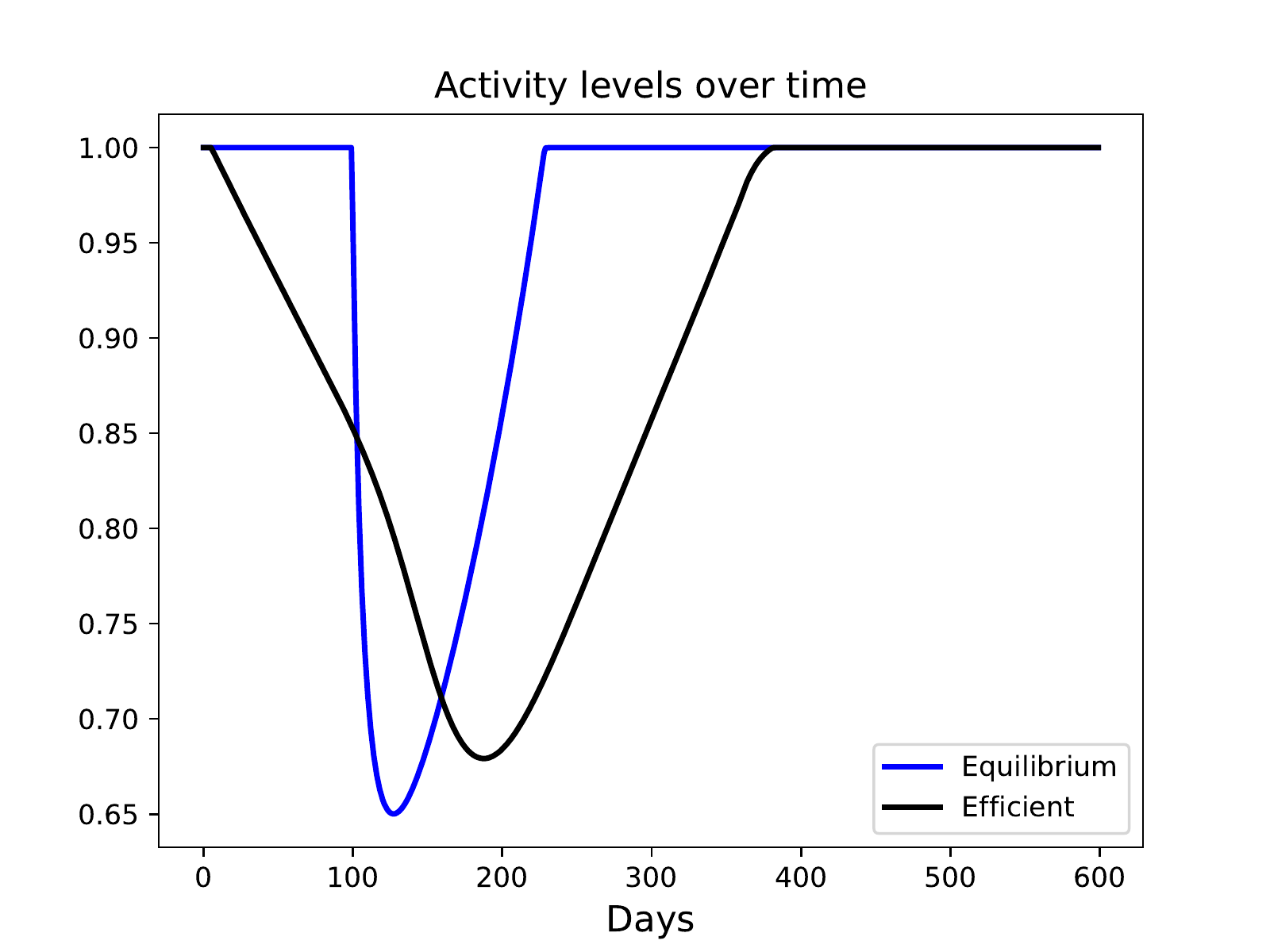}
\includegraphics[width=0.48\linewidth]
{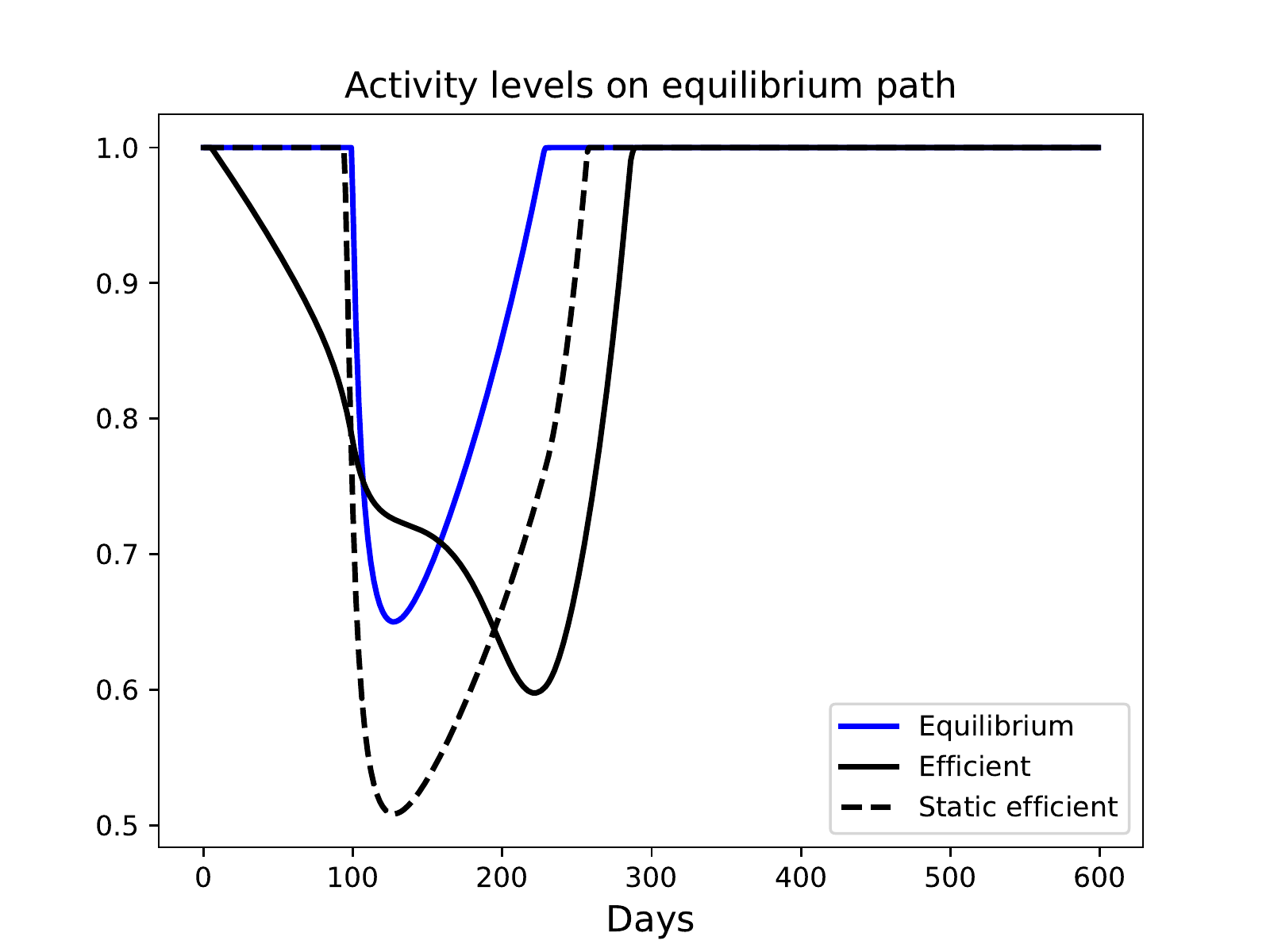}
\caption{Recommended activity over time and along equilibrium paths.}\label{fig:aUT}
\end{figure}

However, as we noted in the introduction, the efficient and equilibrium time paths plotted in Figure \ref{fig:aUT} (left) do not necessarily provide adequate guidance to a government who was either slow to react to the initial infection, or who has limited ability to carry out their desired intervention. Figure \ref{fig:aUT} (left) ought to be interpreted as plotting the activity levels across two fictitious countries, one that pursued no non-pharmaceutical interventions at all, and one that followed the optimal path from the outset of the virus. Therefore it is instructive to complement the foregoing and examine how the recommendations of governments with varying degrees of enforcement ability vary along the \emph{same} path for the state variables.

To analyze this point and to illustrate the theoretical analysis in Section \ref{subsec:external}, Figure \ref{fig:aUT} (right) plots the equilibrium level of activity together with the static efficient action ($\eta=\infty$) and the action the planner would recommend ($\eta=0$) along the equilibrium path. As Theorem \ref{thm:aU_se} suggests, the static efficient path is everywhere below that of the equilibrium path, and so a government who can only enforce lockdowns for a short period of time would unambiguously wish to do so. However, it does not follow that a government with the ability to perfectly enforce activity levels in perpetuity would wish to reduce activity. Indeed, in this example the difference between the efficient and equilibrium activity levels cannot be unambiguously signed, for at some points in the middle of the pandemic the activity choice of the planner exceeds that of the agents in equilibrium. We interpret this last observation to illustrate that relative to the efficient allocation, the competitive equilibrium exhibits inefficiently volatile consumption, 
with abrupt changes over the time that the planner may wish to avoid with more gradual increases and decreases in activity. 

\subsection{The importance of unknown infected agents}\label{subsec:num.sigma}

We have so far assumed that unknown infected agents account for $1-\sigma=60\%$ of all the infected agents. To illustrate the role of imperfect testing and reporting (and to therefore relate our results to the existing literature) we compute the welfare cost and death toll across different specifications for the diagnosis rate $\sigma$. For welfare, we use the utilitarian social welfare $\cW$ in \eqref{eq:wz} and apply the inverse utility function $u^{-1}$ to convert into units of activity. Since the welfare without epidemic equals the highest action 1, we can compute the welfare cost of the epidemic given the current state $z\in Z$ as
\begin{equation}
C(z)\coloneqq 1-u^{-1}(\cW(z)).\label{eq:wcost}
\end{equation}
If we identify activity with current output, we can interpret $C(z)$ as the fraction of \emph{permanent} consumption the society is willing to give up to avoid the epidemic. For the death toll, we depict the expected cumulative death (per 100,000 population), taking into consideration the stochastic nature of vaccine arrival.

Figure \ref{fig:welfare_death} shows the results, where the series labeled ``Myopic'', ``PBE'', ``SPP'' correspond to the myopic equilibrium (standard SIR model), perfect Bayesian Markov competitive equilibrium, and the social planner's problem (efficient action), respectively. In the equilibrium allocation, the epidemic has a large welfare cost of about 1.8\% reduction in permanent consumption. Interestingly, for all diagnosis rates examined, the difference in welfare is larger between the equilibrium and myopic allocation than between the equilibrium and efficient allocation, implying that the welfare gains from the optimal lockdown policy are modest for the benchmark parameters. Furthermore, both the welfare loss and expected death toll are increasing in the diagnosis rate $\sigma$. This is because when testing improves, part of the infected agents switch from unknown to known, take full action $a_I=1$, and infect susceptible agents. This example shows that improving testing without enforcing quarantine could (in principle) worsen welfare.

\begin{figure}[!htb]
\centering
\includegraphics[width=0.48\linewidth]{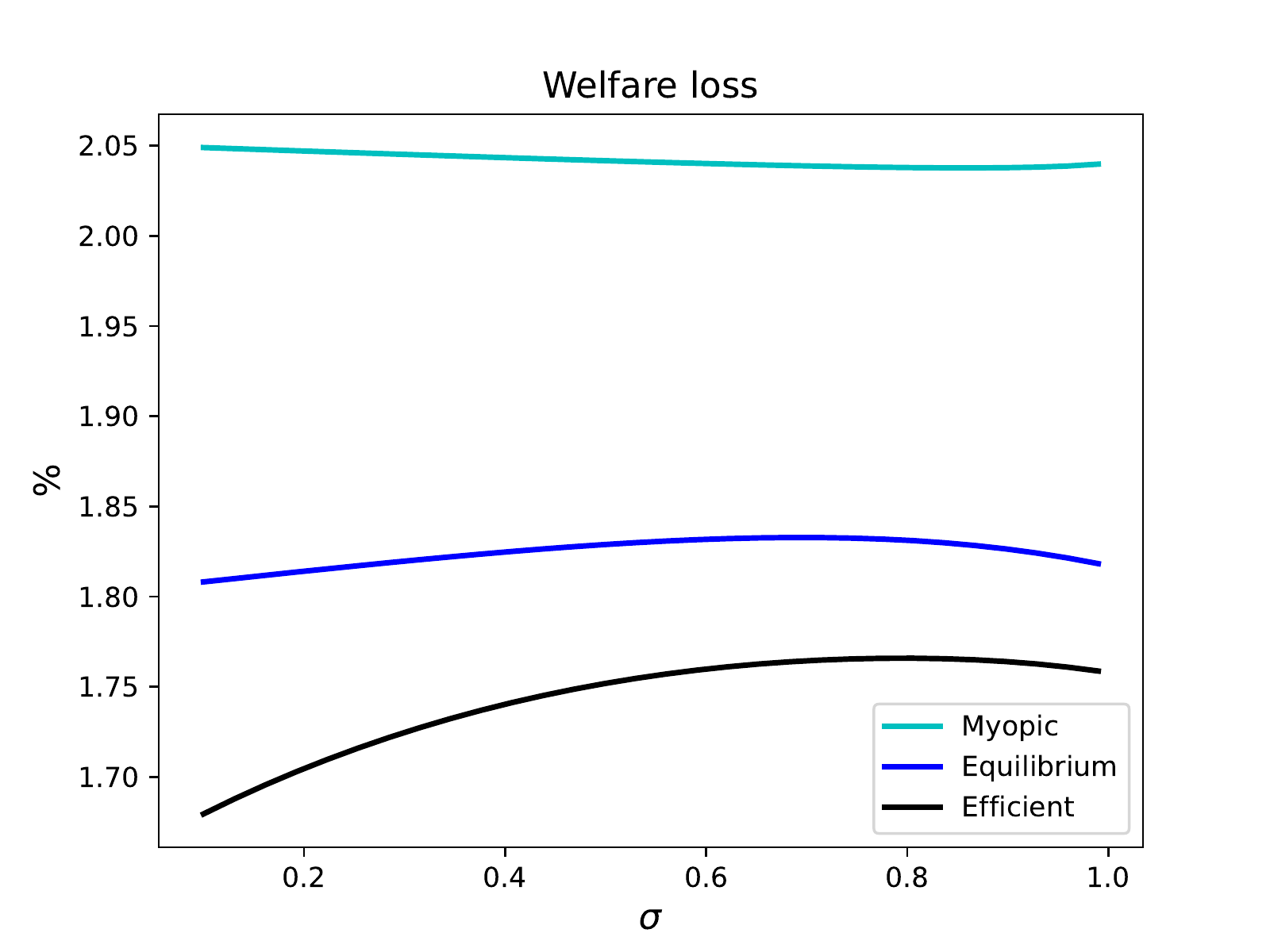}
\includegraphics[width=0.48\linewidth]{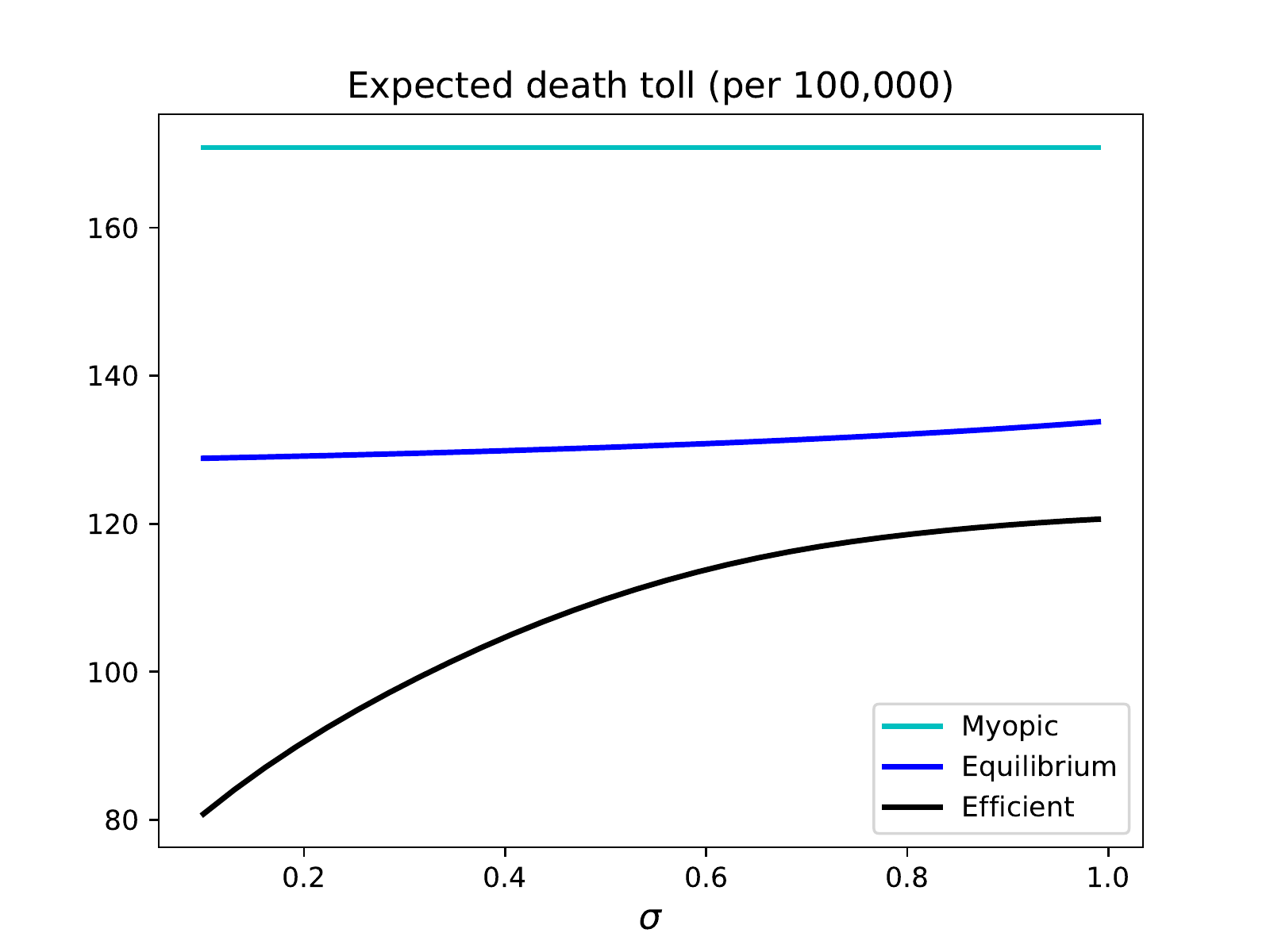}
\caption{Welfare loss and death toll.}\label{fig:welfare_death}
\end{figure}

\subsection{The importance of quarantine policies} \label{subsec:num.quar}

The analysis so far has assumed that the $I_k$ (known infected) agents choose the highest action $a_I = 1$. Although this is unrealistic because $I_k$ agents may self-isolate due to incapacitation or altruism, it provides the most conservative analysis (worst-case scenario). As a complementary analysis, we now solve the model assuming that $I_k$ agents choose $a_I = 0.4$ yet enjoy the highest possible utility ($u_I = 0$). The choice of $a_I = 0.4$ is motivated by the empirical study of \cite{He2020}, who document that 44\% of secondary cases were infected during the presymptomatic stage (prior to diagnosis) of the primary cases. This assumption corresponds to the immediate isolation of the infected cases upon diagnosis, which provides the best-case analysis.

Figure \ref{fig:welfare_death_q} shows the welfare cost and death toll with maximal quarantine. In each case, the welfare gains from quarantine are substantial even with relatively low diagnosis rate $\sigma$. Relative to the equilibrium outcome, the additional welfare gains from the optimal lockdown policy are modest, but now exceed the welfare gains from self-interested behavior for many diagnostic rates. As is intuitive, the difference between the high and low quarantine allocations is particularly stark when the diagnosis rate is high (which increases the effectiveness of quarantine). In the extreme case of $\sigma=1$, individuals take no precautions in either the efficient or equilibrium allocations, since the effective reproduction number in \eqref{eq:Rz} satisfies $\cR<1$ and the epidemic is essentially avoided.

\begin{figure}[!htb]
\centering
\includegraphics[width=0.48\linewidth]{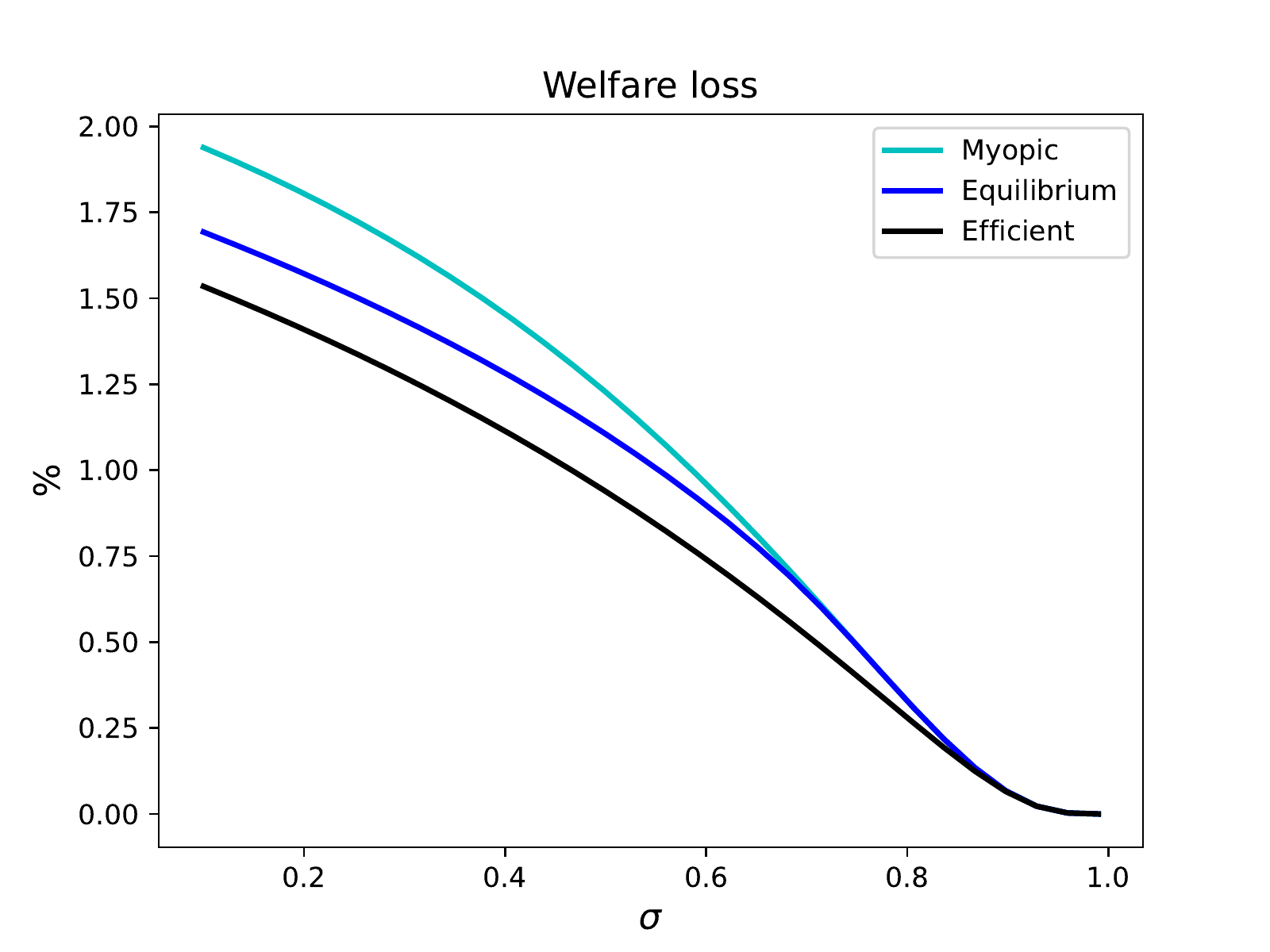}
\includegraphics[width=0.48\linewidth]{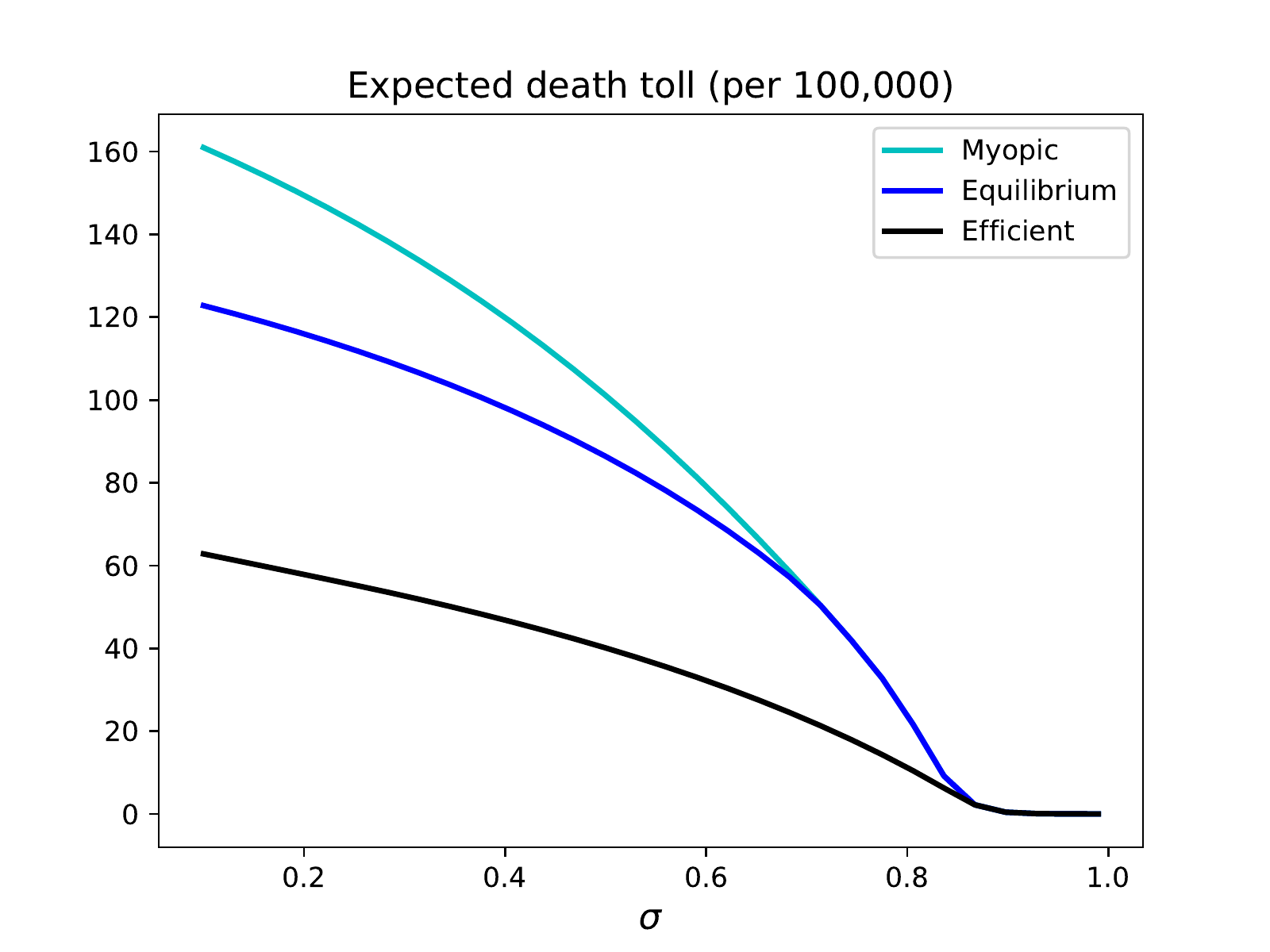}
\caption{Welfare loss and death toll with quarantine.}\label{fig:welfare_death_q}
\end{figure}

\subsection{The importance of vaccine arrival}\label{subsec:num.nu}

Our model has so far assumed that the vaccine is expected to arrive in roughly {\Tvac} year. Consequently, it is likely that the vaccine will have arrived before the attainment of herd immunity. How much does this assumption affect the optimal policies? Should we expect qualitatively similar features to emerge when a vaccine is unlikely to be forthcoming? To explore this point, we solve for the equilibrium and efficient actions when the vaccine arrives very far into the future ($T=100$ years). Although the assumption that a vaccine is (effectively) not forthcoming is quite extreme and unrealistic for COVID-19, we include it here because it illustrates the important distinction between static and dynamic externalities and could be relevant for future pandemics. Figure \ref{fig:eq_nu0} shows the epidemic dynamics (left panels), contour plots of recommended actions over the state space (right panels), and the recommended activity levels, both over time and along the equilibrium path (bottom panel).

\begin{figure}[!htb]
\centering
\includegraphics[width=0.48\linewidth]{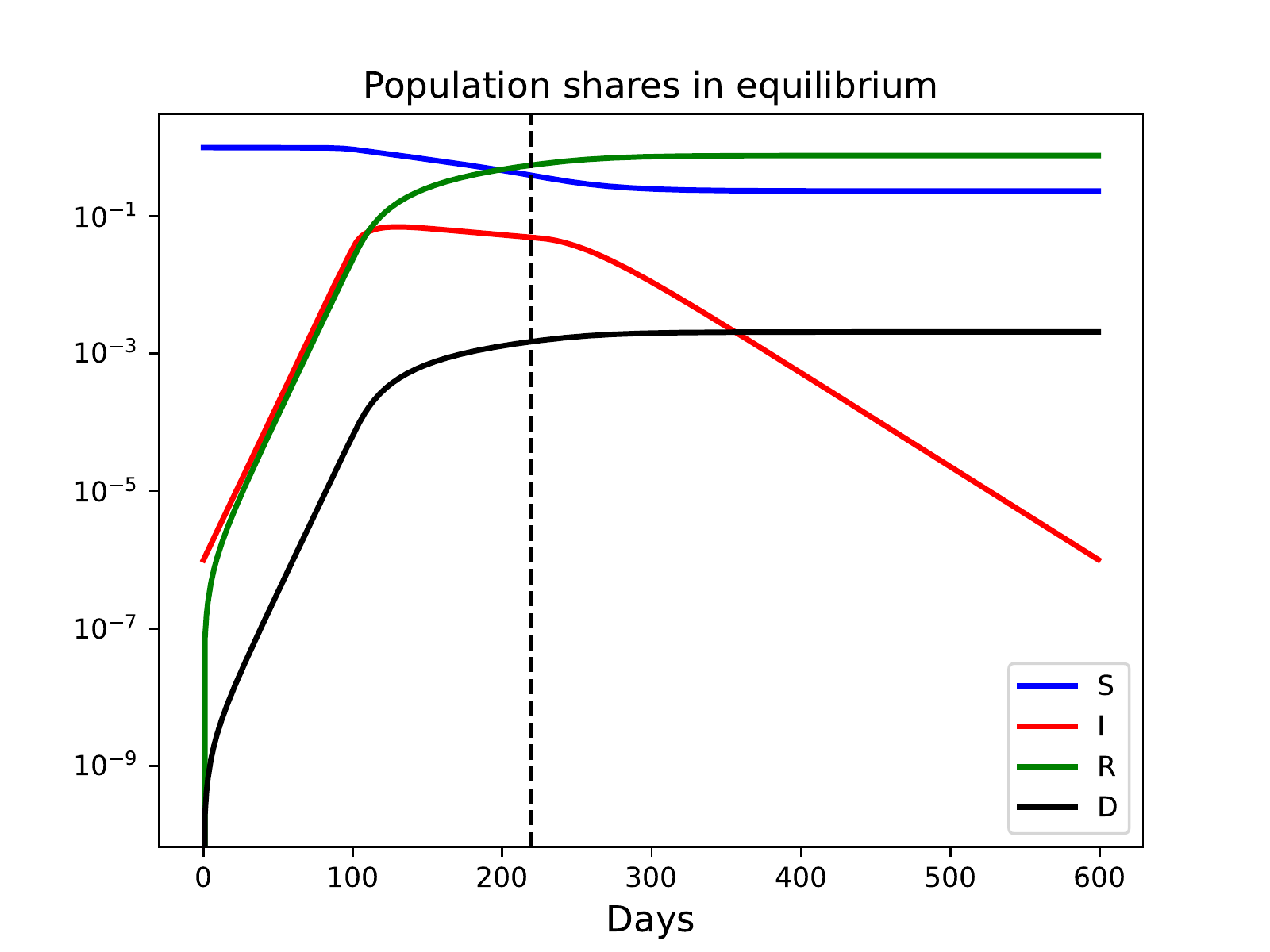}
\includegraphics[width=0.48\linewidth]{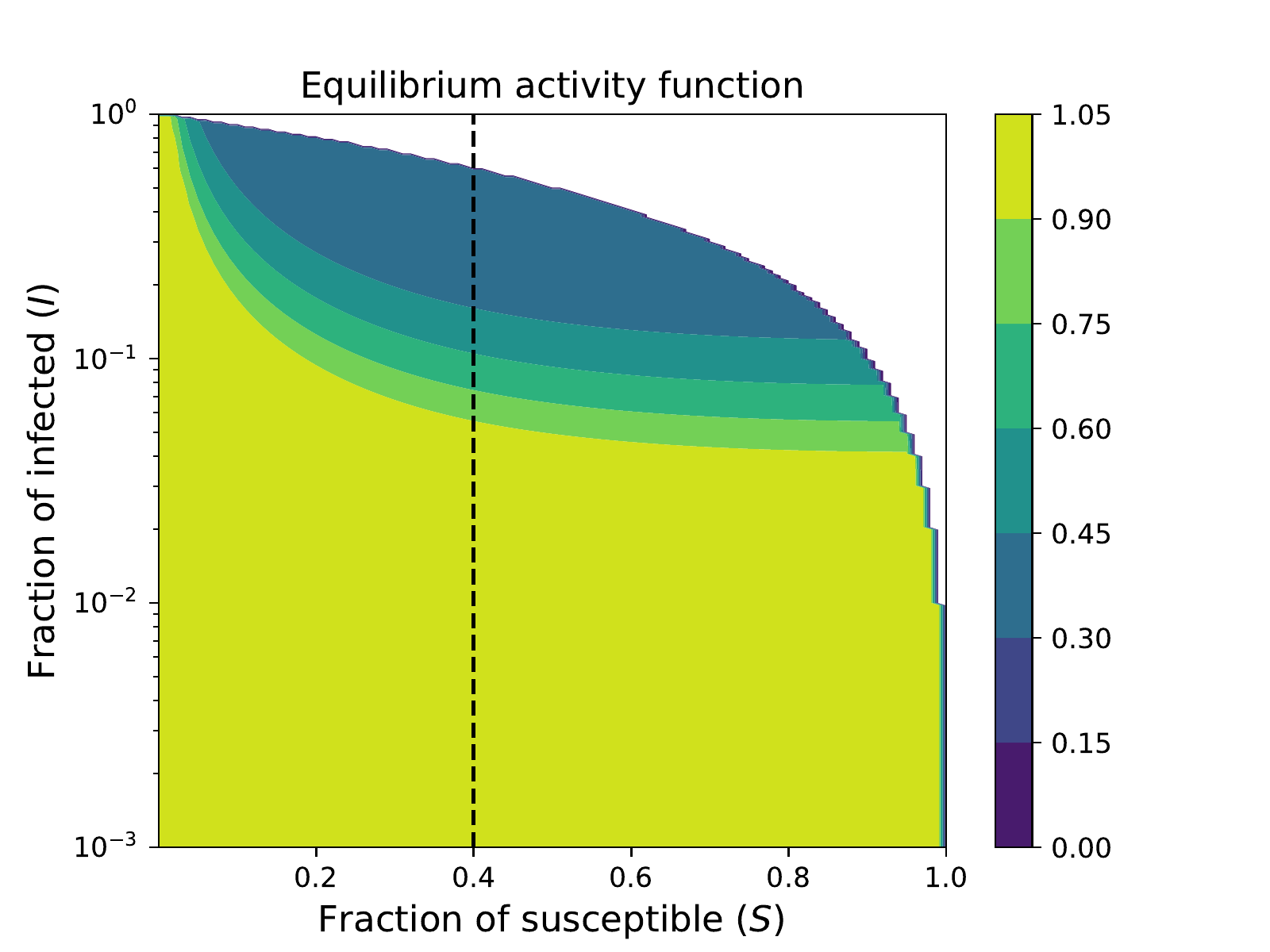}
\includegraphics[width=0.48\linewidth]{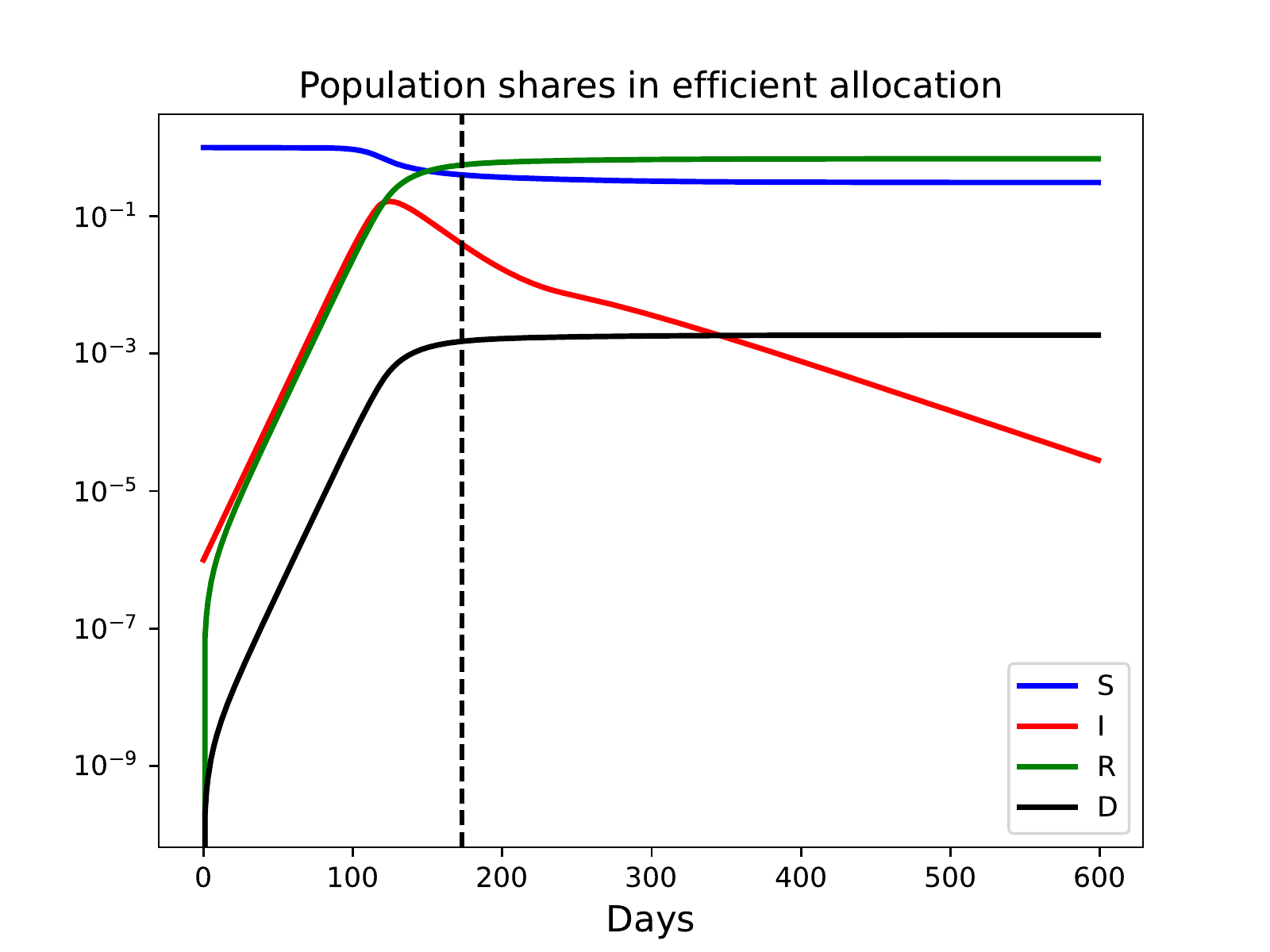}
\includegraphics[width=0.48\linewidth]{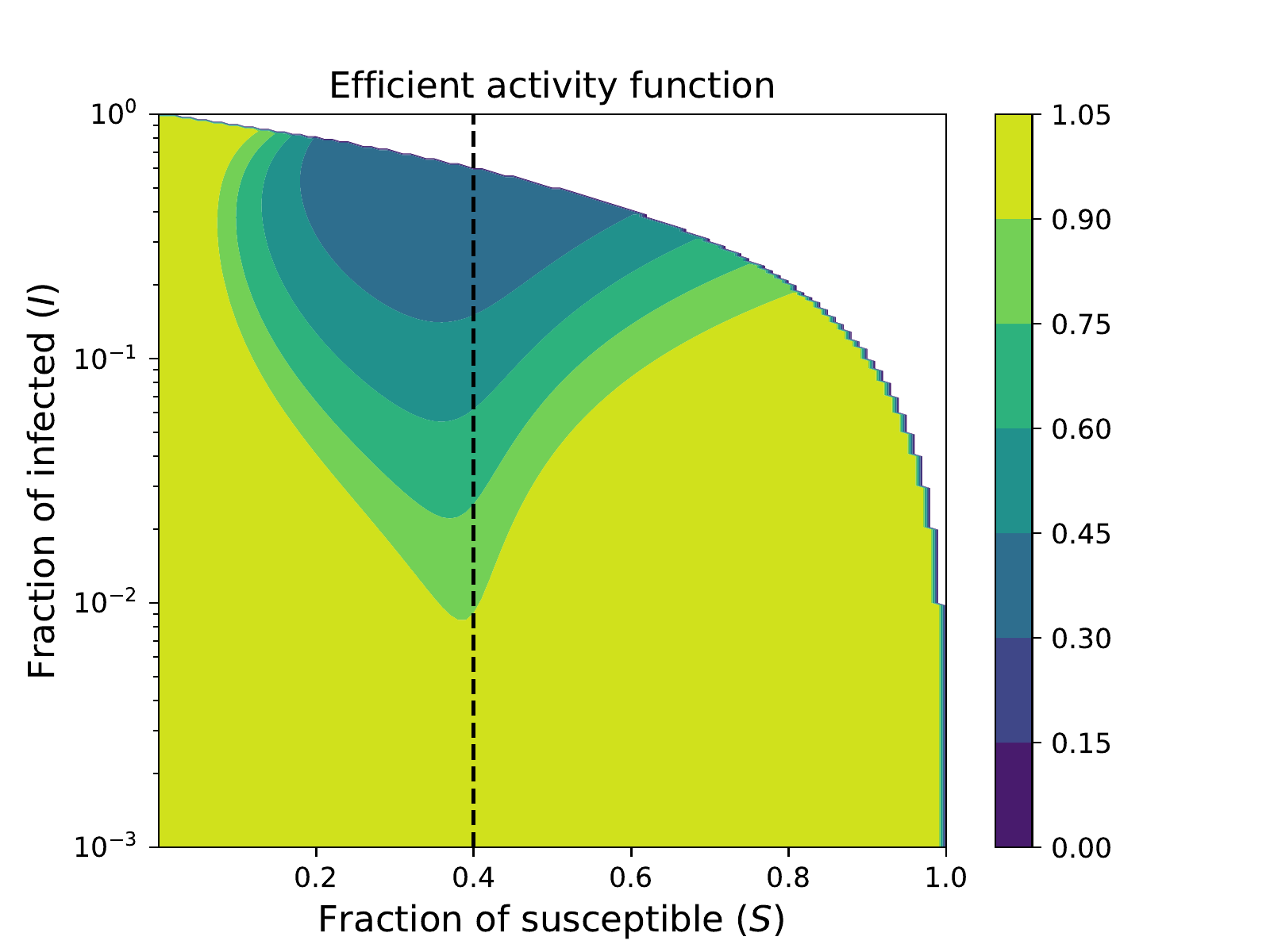}
\includegraphics[width=0.48\linewidth]{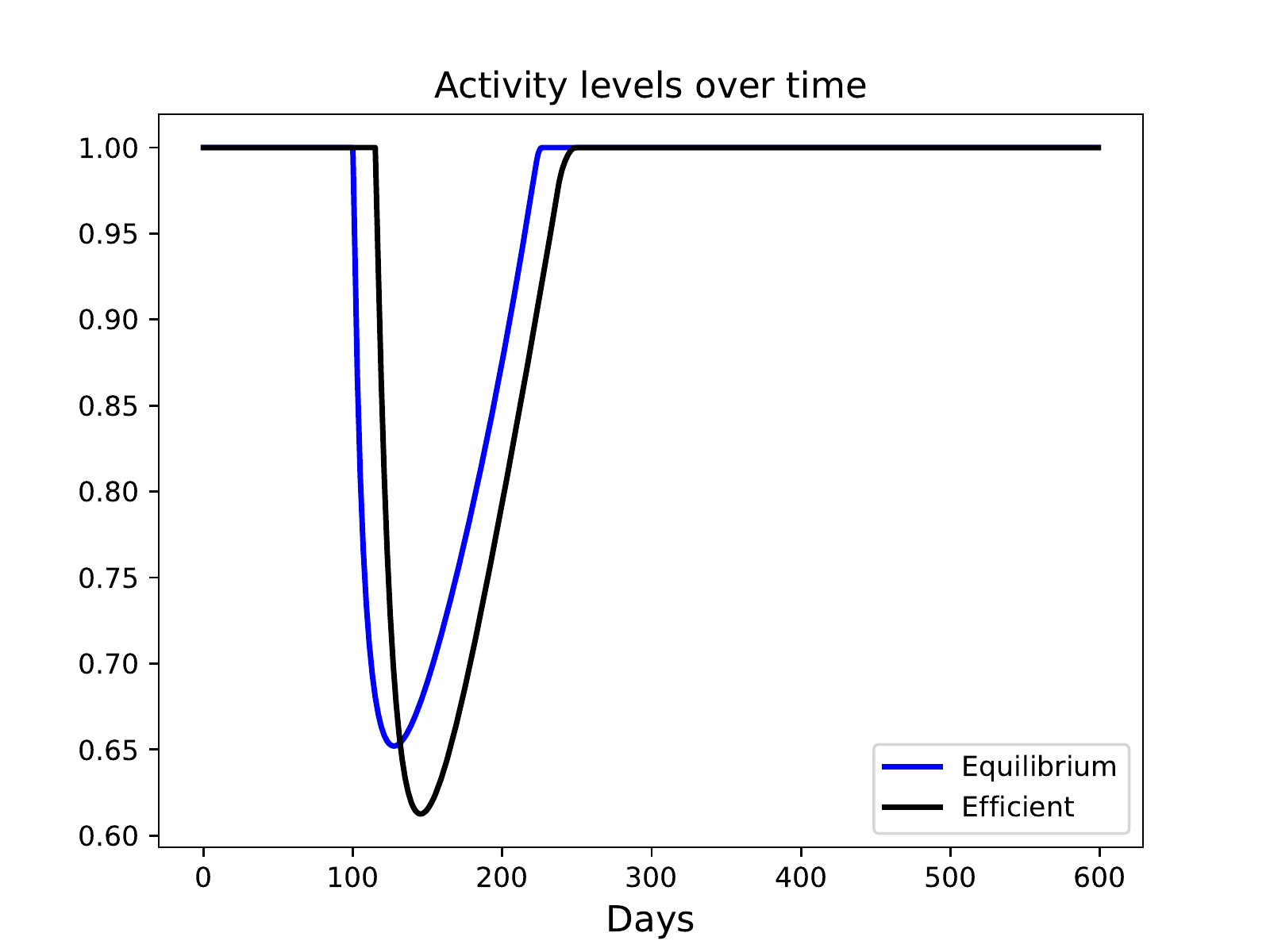}
\includegraphics[width=0.48\linewidth]{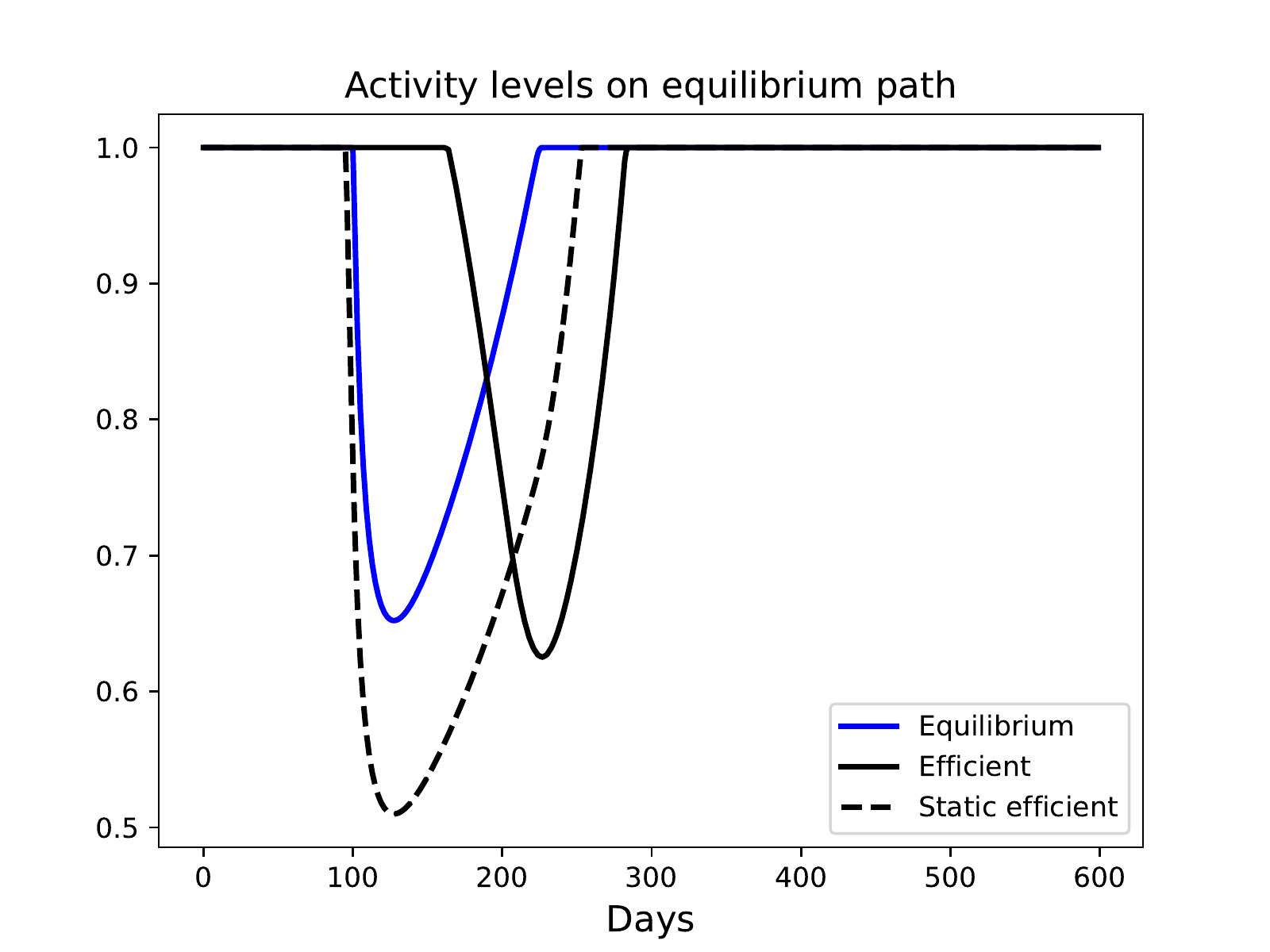}
\caption{Epidemic dynamics and recommended actions with $T=100$.}\label{fig:eq_nu0}
\end{figure}

Compared with the benchmark case (Figure \ref{fig:eq}, where $T=\Tvac$), the equilibrium dynamics and activity levels (top panels) in Figure \ref{fig:eq_nu0} are essentially identical, with herd immunity achieved after roughly 200 days. However, the optimal lockdown policy is qualitatively different. In contrast with Figure \ref{fig:aUT}, in which the reduction in activity is both immediate and gradual, the bottom left panel of Figure \ref{fig:eq_nu0} shows that early in the pandemic agents take high actions in the efficient allocation for longer than in the equilibrium allocation, so that the planner is effectively shifting the ``lockdown'' toward the later stages of the pandemic.


Our intuition for this result is that when the vaccine is unlikely to arrive until very far into the future, the only way in which the pandemic ends is through the attainment of herd immunity. In this case, the planner reduces excess deaths by ensuring that herd immunity is only barely attained.  Consequently, when compared with the equilibrium allocation, the planner first accelerates the reduction in susceptible shares before sharply reducing activity towards the middle and end of the pandemic in order to ensure that herd immunity is not excessively ``overshot''. As shown in the bottom right panel of Figure \ref{fig:eq_nu0}, this difference between equilibrium and efficient activity is starker when evaluated along equilibrium paths. The optimal activity recommended by the planner amounts to encouraging (discouraging) agents to take high action before (after) achieving herd immunity. This is despite the fact that the static efficient activity is always below the equilibrium level (as is assured by Theorem \ref{thm:aU_se} regardless of the rate of vaccine arrival). 

Figure \ref{fig:eq_nu0SD} reinforces this intuition, and shows that the share of susceptible agents is reduced more rapidly and plateaus more smoothly in the efficient allocation when compared with the equilibrium allocation, leading to both a shorter pandemic and an overall reduction in deaths. Compared with the efficient and equilibrium allocations, the pandemic in the myopic allocation is shorter (no precautions are taken) but the loss of life is much higher. 

\begin{figure}[!htb]
\centering
\includegraphics[width=0.48\linewidth]{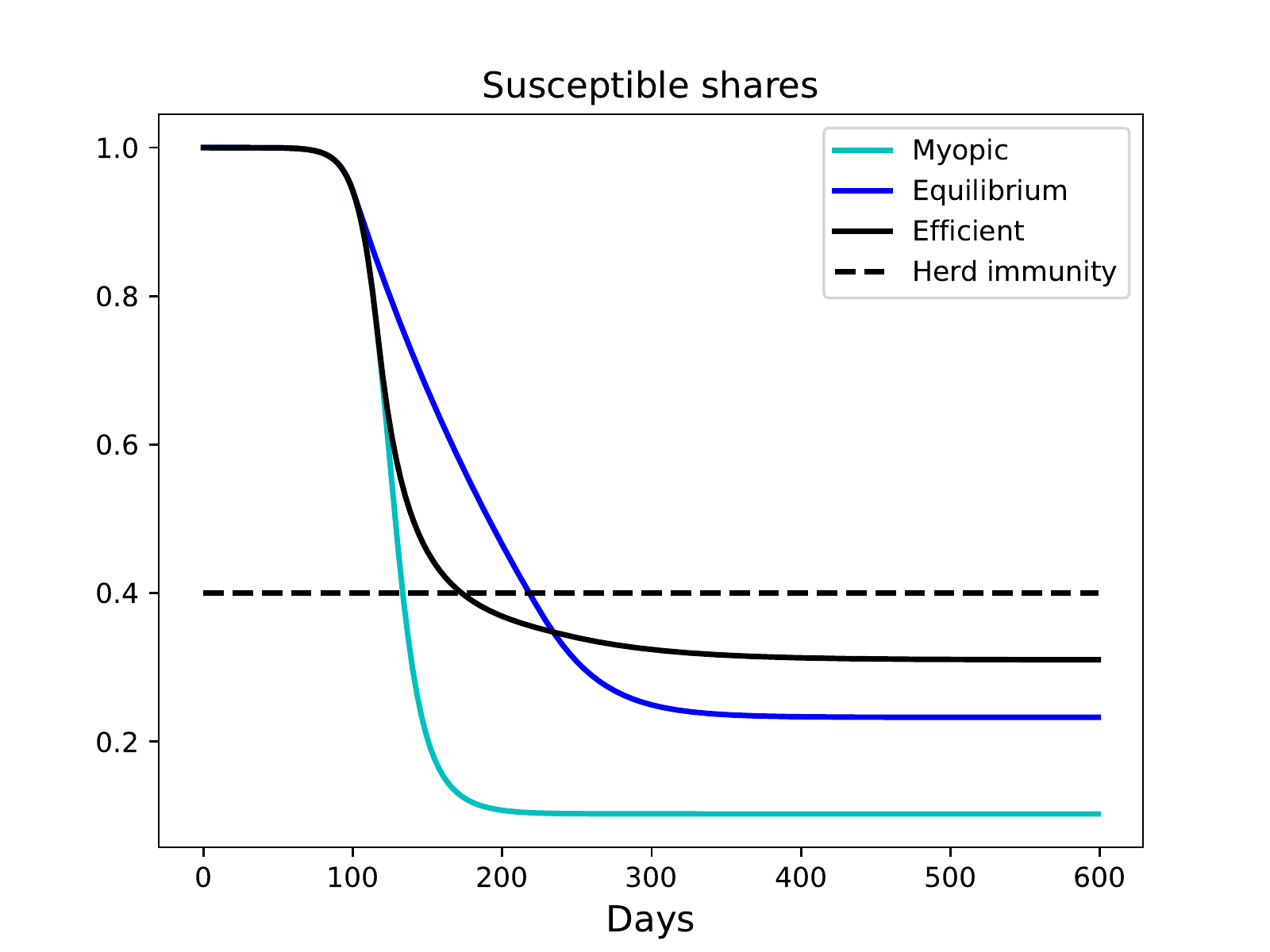}
\includegraphics[width=0.48\linewidth]{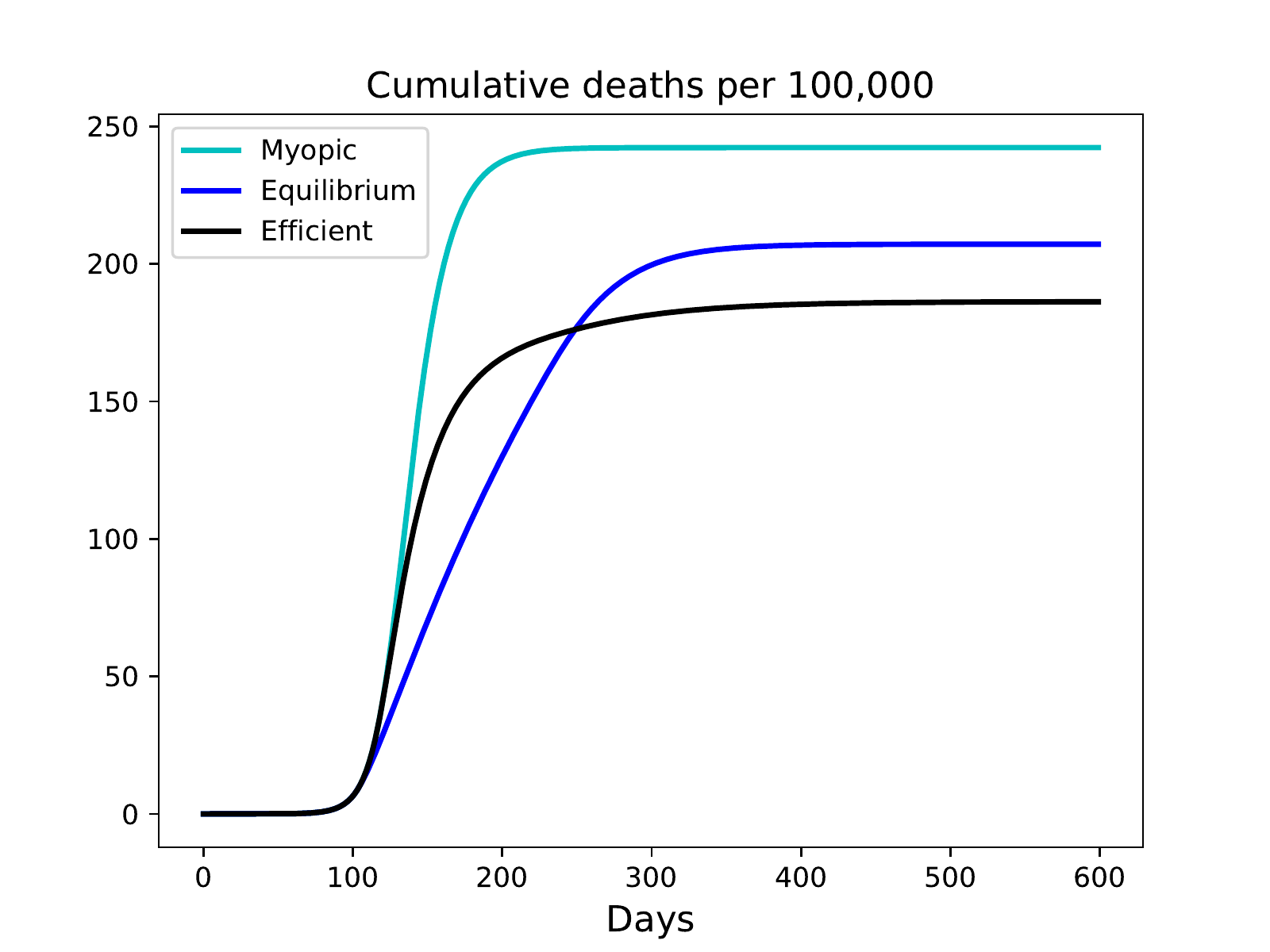}
\caption{Myopic, equilibrium and efficient susceptible shares and deaths with late vaccine arrival.}\label{fig:eq_nu0SD}
\end{figure}

Figure \ref{fig:welfare_death_T100} complements the above discussion by fixing $\sigma=0.4$ and depicting the welfare cost and death toll as a function of the expected arrival time for the vaccine, ranging from $T = 1/4$ (three months) to $T=100$ (note the logarithmic scale). As expected, in all three allocations (myopic, equilibrium and efficient), the welfare loss and expected death toll rise monotonically with the expected arrival times. 

\begin{figure}[!htb]
\centering
\includegraphics[width=0.48\linewidth]{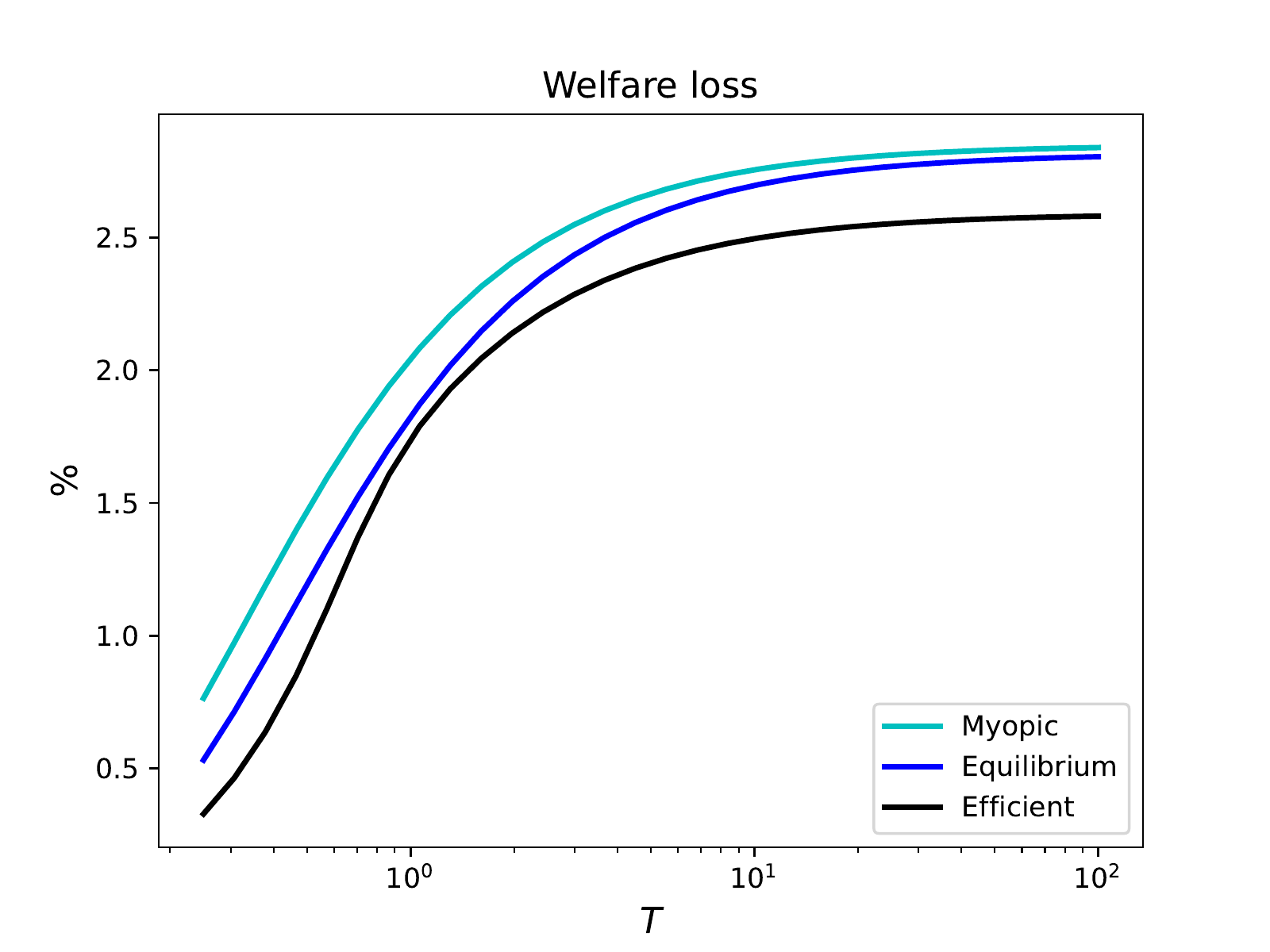}
\includegraphics[width=0.48\linewidth]{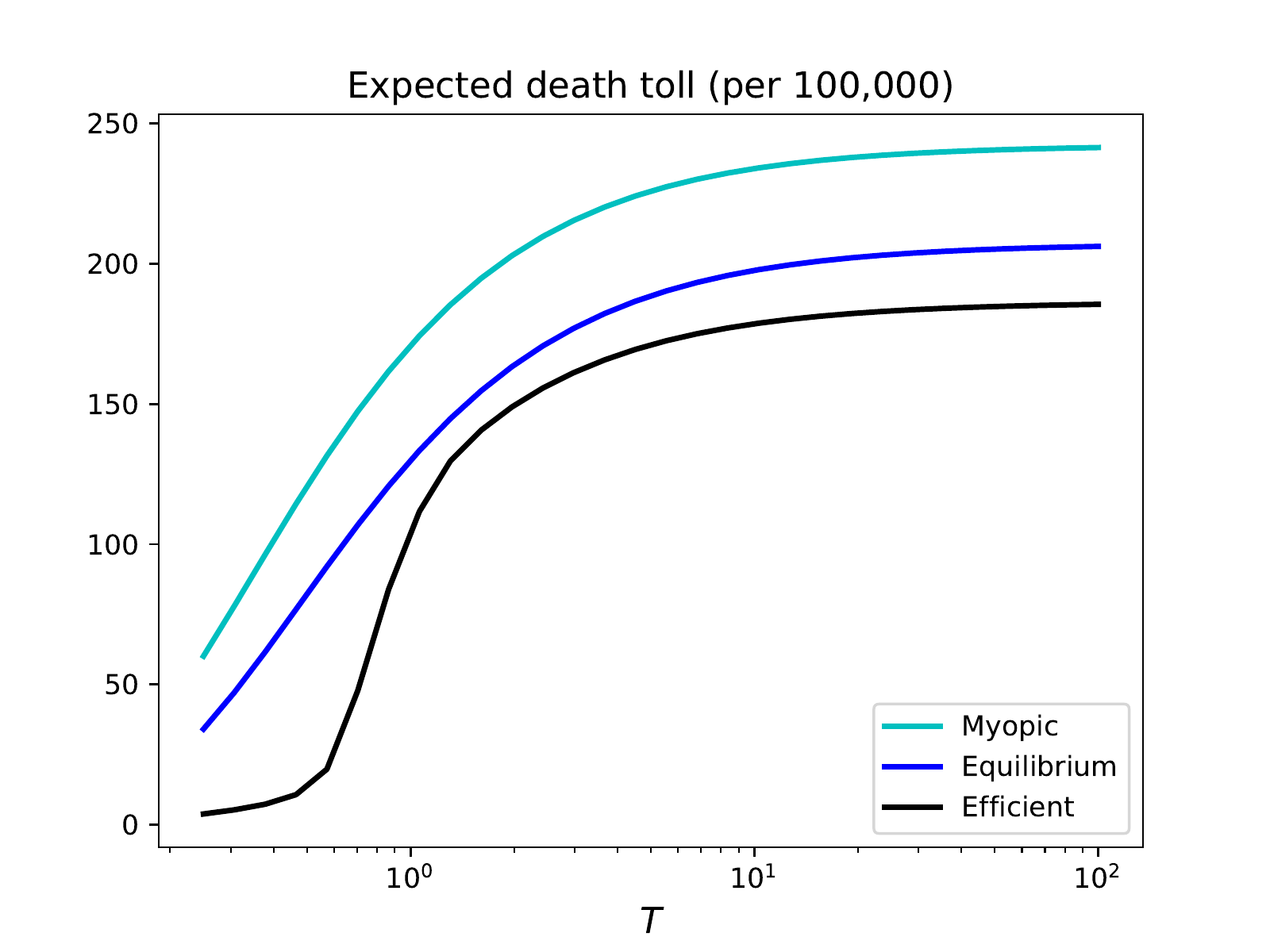}
\caption{Welfare loss and death toll with varying arrival.}\label{fig:welfare_death_T100}
\end{figure}

Interestingly, the welfare benefits from the optimal lockdown policy (relative to the equilibrium outcome) appear to be nonlinear in the arrival length, first falling until roughly $T=1$ and rising thereafter. In contrast, the welfare benefits from individually optimal behavior (relative to the myopic outcome) appear to decrease monotonically with the expected arrival date.

In particular, note that in contrast with the benchmark quantities depicted in Figure \ref{fig:welfare_death}, when the vaccine does not arrive until far into the future there is little difference in welfare between the myopic and equilibrium allocations. Our intuition for this result is that when the vaccine is not forthcoming, the precautionary behavior of agents in equilibrium leads to a relatively small reduction in deaths, and so the benefits from lower mortality are almost outweighed by the reduced length of the pandemic.

\section{Concluding remarks}

In this paper we have theoretically studied optimal epidemic control in an equilibrium model with imperfect testing and enforcement. We proved the existence of a perfect Bayesian Markov competitive equilibrium, and showed that a government that can only enforce lockdowns for a short period of time will always wish to (weakly) reduce activity, but that their incentive to do vanishes as testing becomes perfect. In contrast, a government with the ability to enforce lockdowns for an indefinitely long period of time may wish to increase or decrease activity relative to the equilibrium outcome.

We then calibrated the model to the COVID-19 epidemic and found the following in our numerical experiments:
\begin{enumerate*}
\item relative to the benchmark SIR model (with no precautionary behavior), the welfare gains from lockdown (isolation of the general public) are smaller than the gains from rational, self-interested behavior;
\item in the absence of altruistic behavior, improving testing (which increases $\sigma$) provides limited welfare gains (and may even reduce welfare) if it is not accompanied by quarantine (isolation of infected agents); and
\item the welfare gains from quarantine are large even with imperfect testing.
\end{enumerate*} 


Since the first draft of this paper was circulated (April 2021), considerable scientific knowledge has been accumulated about COVID-19. In particular, we have learned that immunity wanes over time and that reinfections are possible, which illustrates some shortcomings of the SIR framework adopted in this paper. The primary difficulty with incorporating imperfect testing and reinfections is that agents will have heterogeneous beliefs (depending on how many times they have contracted COVID) and so the state space is now infinite-dimensional. We have largely retained our original model in part due to our view that the insights of this paper are not restricted to the specifics of the COVID-19 pandemic, and apply to any (possibly future) pandemic for which the SIR structure holds. Further, if reinfections are not deadly (and only a minor inconvenience), for which there is some evidence \citep{Wolter_2022}, then the resulting SI(D)S model could be analytically and numerically studied using similar techniques to those in this paper because agents that once recovered will always take the highest action. However, we leave this and other extensions (such as population growth or mutations) to future research.



\bibliographystyle{plainnat}
\bibliography{localbib}

\appendix

\section{Proofs}\label{sec:proof}

To prove Proposition \ref{prop:V}, we note the following simple lemma.

\begin{lem}\label{lem:fp_subset}
Let $(X,d)$ be a complete metric space and $T:X\to X$ be a contraction with fixed point $x^*\in X$. If $\emptyset \neq X_1\subset X$ is closed and $TX_1\subset X_1$, then $x^*\in X_1$.
\end{lem}

\begin{proof}
Since $X_1\subset X$ is closed, $(X_1,d)$ is a complete metric space. Since $T:X_1\to X_1$ is a contraction, it has a unique fixed point $x_1^*\in X_1$, which is also a fixed point of $T:X\to X$. Therefore $x^*=x_1^*\in X_1$ by uniqueness.
\end{proof}

\begin{proof}[Proof of Proposition \ref{prop:V}]
Let $X=\R^Z$. For $V\in X$, define $TV(z)$ by the right-hand side of \eqref{eq:Bellman.u2}, where $V_U=V$. Since $A=[\ubar{a},1]$ is nonempty and compact, $u:A\to \R$ is continuous, and $Z$ is a finite set, the maximum is achieved and $T:X\to X$ is well-defined. Let us show that $T$ is a contraction by verifying \cite{blackwell1965}'s sufficient conditions. If $V_1,V_2\in X$ and $V_1\le V_2$ pointwise, since $0\le \sigma\mu pa\le 1$, it follows from the definition of $T$ that
\begin{align*}
(TV_1)(z)& =\max_{a\in A}\bigl\{(1-\e^{-r\Delta})u(a) +\e^{-r\Delta}\E_z((1-\sigma\mu pa)\e^{-\nu\Delta}V_1(z')+\sigma\mu pa V_{I_k})\bigr\}\\
\le &\max_{a\in A}\bigl\{(1-\e^{-r\Delta})u(a)
 +\e^{-r\Delta}\E_z((1-\sigma\mu pa)\e^{-\nu\Delta}V_2(z')+\sigma\mu pa V_{I_k})\bigr\}\\
& =(TV_2)(z),
\end{align*}
so $T$ is monotone. If $V\in X$ and $c\ge 0$ is any constant, then
\begin{align*}
&(T(V+c))(z)\\
&=\max_{a\in A}\bigl\{(1-\e^{-r\Delta})u(a) + \e^{-r\Delta}\E_z((1-\sigma\mu pa)\e^{-\nu\Delta}(V(z')+c)+\sigma\mu pa V_{I_k})\bigr\}\\
&\le \max_{a\in A}\bigl\{(1-\e^{-r\Delta})u(a) + \e^{-r\Delta}\E_z((1-\sigma\mu pa)\e^{-\nu\Delta}V(z')+\sigma\mu pa V_{I_k})\bigr\}+\e^{-(r+\nu)\Delta}c\\
&=(TV)(z)+\e^{-(r+\nu)\Delta}c,
\end{align*}
so the discounting property holds with modulus $\e^{-(r+\nu)\Delta}<1$. Therefore $T$ is a contraction mapping and has a unique fixed point $V_U\in X$. Because $u:A=[\ubar{a},1]\to \R$ is continuous (hence bounded), it is straightforward to verify the transversality condition. Therefore $V_U$ is the value function.

Next, let us show the inequalities \eqref{eq:Vineq}. Noting that $u_D<u_I\le 0$, $r>0$, $\gamma>0$, $\delta\in (0,1]$, and $V_{I_k}$ is given by \eqref{eq:VIk}, we have
\begin{equation*}
V_D=u_D<\frac{(\e^{r\Delta}-1)u_I+\gamma\Delta\delta u_D}{\e^{r\Delta}-1+\gamma\Delta}=\frac{(1-\e^{-r\Delta})u_I+\e^{-r\Delta}\gamma\Delta\delta u_D}{1-\e^{-r\Delta}(1-\gamma\Delta)}=V_{I_k}<0.
\end{equation*} 
To obtain the lower bound for $V_U$, Define the constants
\begin{equation}
c_1\coloneqq \frac{\sigma\beta\Delta}{\e^{(r+\nu)\Delta}-1+\sigma\beta\Delta}\in (0,1)\label{eq:c1}
\end{equation}
and $V_1\coloneqq c_1\e^{\nu\Delta}V_{I_k}$. Note that
\begin{align*}
c_1\e^{\nu\Delta}<1&\iff \frac{\e^{\nu\Delta}\sigma\beta\Delta}{\e^{(r+\nu)\Delta}-1+\sigma\beta\Delta}<1\\
&\iff \e^{\nu\Delta}\sigma\beta\Delta<\e^{(r+\nu)\Delta}-1+\sigma\beta\Delta\\
&\iff (\e^{\nu\Delta}-1)\sigma\beta\Delta<\e^{(r+\nu)\Delta}-1,
\end{align*}
which is true because $\sigma\in (0,1]$, $\beta\Delta\in [0,1]$, $r>0$, and $\nu\ge 0$. Since $c_1\e^{\nu\Delta}<1$ and $V_{I_k}<0$, it follows that
\begin{equation*}
V_{I_k}<c_1\e^{\nu\Delta}V_{I_k}=V_1=\frac{\e^{\nu\Delta}\sigma\beta\Delta}{\e^{(r+\nu)\Delta}-1+\sigma\beta\Delta}V_{I_k},
\end{equation*}
which is part of \eqref{eq:Vineq}. Let us next show $V_U\ge V_1$. To this end, define the (nonempty closed) subset of $X$ by $X_1=\set{V\in X|V\ge V_1}$. By Lemma \ref{lem:fp_subset}, to show $V_U\ge V_1$, it suffices to show $TX_1\subset X_1$. If $V\in X_1$, then by the definition of $T$, we have
\begin{align}
(TV)(z)& =\max_{a\in A}\bigl\{(1-\e^{-r\Delta})u(a) +\e^{-r\Delta}\E_z((1-\sigma\mu pa)\e^{-\nu\Delta}V(z')+\sigma\mu pa V_{I_k})\bigr\} \notag \\
\ge &\max_{a\in A}\bigl\{(1-\e^{-r\Delta})u(a) +\e^{-r\Delta}\E_z((1-\sigma\mu pa)c_1V_{I_k}+\sigma\mu pa V_{I_k})\bigr\}. \label{eq:VUlb1}
\end{align}
Setting $a=1$ in the right-hand side of \eqref{eq:VUlb1} and noting that $u(1)=0$, we obtain
\begin{equation*}
(TV)(z)\ge \e^{-r\Delta}((1-\sigma\mu p)c_1+\sigma\mu p)V_{I_k}.
\end{equation*}
Therefore to show $TX_1\subset X_1$, noting that $V_{I_k}<0$, it suffices to show that
\begin{equation}
\e^{-r\Delta}((1-\sigma\mu p)c_1+\sigma\mu p)\le c_1\e^{\nu\Delta}.\label{eq:c1ineq}
\end{equation}
By \eqref{eq:pz}, $a_U(z)\le 1$, and the definition of $Z$ in \eqref{eq:Z}, we have
\begin{equation*}
p=\beta\Delta(a_II_k+a_U(z)I_u)\le \beta\Delta.
\end{equation*}
Since $\mu\in [0,1]$ and $c_1\in (0,1)$, it follows that
\begin{equation*}
\e^{-r\Delta}((1-\sigma\mu p)c_1+\sigma\mu p)\le \e^{-r\Delta}((1-\sigma\beta\Delta)c_1+\sigma\beta\Delta)=c_1\e^{\nu\Delta}
\end{equation*}
by \eqref{eq:c1}, which implies \eqref{eq:c1ineq}. Therefore $TX_1\subset X_1$, as desired.

To show $V_U(z)\le V_{R_k}=0$, define the (nonempty closed) subset of $X$ by $X_2\coloneqq \set{V\in X|V\le 0}$. If $V\in X_2$, then
\begin{align*}
(TV)(z)&=\max_{a\in A}\set{(1-\e^{-r\Delta})u(a)+\e^{-r\Delta}\E_z((1-\sigma\mu pa)\e^{-\nu\Delta}V(z')+\sigma\mu pa V_{I_k})}\\
&\le 0
\end{align*}
because $u(a)\le u(1)=0$, $V\le 0$, and $V_{I_k}<0$. Therefore $TX_2\subset X_2$, which implies $V_U\le 0$ by Lemma \ref{lem:fp_subset}.
\end{proof}

\begin{proof}[Proof of Proposition \ref{prop:aU}]
Let
\begin{equation*}
f(a)\coloneqq (1-\e^{-r\Delta})u(a)+\e^{-(r+\nu)\Delta}\E_z((1-\sigma\mu pa)V_U(z')+\sigma\mu pa V_{I_k})
\end{equation*}
be the expression inside the braces in \eqref{eq:Bellman.u2}. Since by Assumption \ref{asmp:comp} individual agents view the next period's state $z'$ as exogenous, $V_U(z')$ does not depend on $a$ and $f$ is strictly concave. Furthermore,
\begin{equation*}
f'(a)=(1-\e^{-r\Delta})u'(a)-\e^{-r\Delta}\sigma\mu p(\E_z \e^{-\nu\Delta}V_U(z')-V_{I_k}).
\end{equation*}
Hence considering the cases $f'(\ubar{a})\gtrless 0$ and $f'(1)\gtrless 0$, the optimal action is given by \eqref{eq:aU}.
\end{proof}

\begin{proof}[Proof of Corollary \ref{cor:aU1}]
Since $0\ge V_U>V_{I_k}$ by \eqref{eq:Vineq} and $\e^{-\nu\Delta}\le 1$, we always have $\E_z \e^{-\nu\Delta}V_U(z')-V_{I_k}>0$. By \eqref{eq:Vineq} and \eqref{eq:aU}, the best response can be rewritten as $a^*=\phi(x)$, where
\begin{equation*}
x\coloneqq \frac{\sigma\mu(z)p(z)}{\e^{r\Delta}-1}\left(\E_z \e^{-\nu\Delta}V_U(z')-V_{I_k}\right)>0.\label{eq:x}
\end{equation*}
Using the bound $V_U\le 0$ in \eqref{eq:Vineq}, we can bound $x$ from above as
\begin{equation}
x\le \frac{\sigma\mu(z)p(z)}{\e^{r\Delta}-1}(-V_{I_k}).\label{eq:xub1}
\end{equation}
Using $\mu(z)\le 1$, $a_U(z)\le 1$, and
\begin{equation*}
p(z)=\beta\Delta (a_II_k+a_U(z)I_u)\le \beta\Delta (I_k+I_u)=\beta \Delta I
\end{equation*}
by \eqref{eq:pz}, it follows from \eqref{eq:xub1} and $V_{I_k}<0$ that
\begin{equation}
x\le -\frac{\sigma\beta\Delta}{\e^{r\Delta}-1}I V_{I_k}.\label{eq:xub2}
\end{equation}
Therefore if $I\le \bar{I}$, where $\bar{I}$ is as in \eqref{eq:Ibar}, it follows from \eqref{eq:xub2} that $x\le u'(1)$. Therefore $a^*=\phi(x)=1$ by \eqref{eq:phi}.
\end{proof}

To prove Theorem \ref{thm:exist}, we consider a general stochastic game with complete information and finitely many competitive agents denoted by $n\in N=\set{1,\dots,N}$.\footnote{There is a large literature on existence theorems for stationary Markov perfect equilibria in stochastic games, see \cite{HeSun2017} and the references therein. We do not use these earlier results as we are interested in symmetric pure strategy equilibria, which requires a special structure as in our SIR model.} Suppose that there are finitely many agent types indexed by $h\in H=\set{1,\dots,H}$. At each point in time, the aggregate state of the economy is denoted by $z\in Z$, which is a finite set. A type $h$ agent takes action $x\in A_h$. Let $A=\bigotimes_{h=1}^H A_h$.

Let $u_h(x,a,z)$ be the flow utility of a type $h$ agent when the agent takes action $x\in A_h$, the average action of other agents is $a\in A$, and the aggregate state is $z\in Z$. Type $h$ agents discount future utility with discount factor $\rho_h\in [0,1)$. Time is denoted by $t=0,1,\dotsc$. The aggregate state $z$ evolves stochastically depending on agents' average actions and the current state. Let $q(a,z,z')$ be the probability of $z_{t+1}=z'$ conditional on $(a_t,z_t)=(a,z)\in A\times Z$. When computing $q(a,z,z')$, agents take the average action $a$ as given and ignore the impact of their action $x$ on $a$ (competitive behavior). An agent's type evolves stochastically depending on an agent's own action, the other agents' average actions, and the current state. Let $p_{hh'}(x,a,z)$ be the probability that a type $h$ agent switches to type $h'$ next period when the state is $z\in Z$, he takes action $x\in A_h$, the average action is $a\in A$, and the current state is $z\in Z$. Value functions are denoted by $V=(V_1,\dots,V_H)$, where $V_h:Z\to \R$. Reaction functions are denoted by $\theta=(\theta_1,\dots,\theta_H)$, where $\theta_h:Z\to A_h$.

\begin{defn}\label{defn:eq_gen}
A (pure strategy) \emph{Markov competitive equilibrium} is a pair $(V,\theta)$ of value and reaction functions such that
\begin{enumerate}
\item (Consistency) The transition probability of the aggregate state is consistent with individual behavior, so $\Pr(z'\mid z)=q(\theta(z),z,z')$.
\item (Sequential rationality) For each $h\in H$ and $z\in Z$, the Bellman equation
\begin{equation}
V_h(z)=\max_{x\in A_h}\set{(1-\rho_h)u_h(x,\theta(z),z)+\rho_h\sum_{h'=1}^Sp_{hh'}(x,\theta(z),z)\E_zV_{h'}(z')}\label{eq:Bellman_gen}
\end{equation}
holds, and $x=\theta_h(z)$ achieves the maximum.
\end{enumerate}
\end{defn}

\begin{thm}\label{thm:exist_gen}
Suppose the following assumptions hold:
\begin{enumerate*}
\item For each $h\in H$, the action set $A_h$ is a nonempty compact convex subset of some Euclidean space.
\item For each $h\in H$, the payoff function $u_h:A_h\times A\times Z \to \R$ is continuous in $(x,a,z)$ and strictly concave in $x$.
\item The transition probability $q:A\times Z\times Z\to [0,1]$ is continuous.
\item For each $h,h'\in H$, the transition probability $p_{hh'}:A_h\times A\times Z\to [0,1]$ is continuous in $(x,a,z)$ and affine in $x$.
\end{enumerate*}
Then there exists a pure strategy Markov competitive equilibrium.
\end{thm}

To prove Theorem \ref{thm:exist_gen}, we need the following lemma.

\begin{lem}\label{lem:fp_cont}
Let $(X,d)$ be a complete metric space and $\Theta$ be a topological space. Endow $X\times \Theta$ with the product topology. Suppose $T:X\times \Theta\to X$ is continuous and there exists $\beta\in [0,1)$ such that
\begin{equation}
d(T(x,\theta),T(y,\theta))\le \beta d(x,y) \label{eq:contraction}
\end{equation}
for all $x,y\in X$ and all $\theta\in\Theta$. Then
\begin{enumerate*}
\item for each $\theta\in \Theta$, there exists a unique $x^*(\theta)\in X$ such that $T(x^*(\theta),\theta)=x^*(\theta)$, and
\item $x^*:\Theta\to X$ is continuous.
\end{enumerate*}
\end{lem}

\begin{proof}
The first claim is immediate from the contraction mapping theorem.

To show the second claim, for simplicity write $T_\theta x\coloneqq T(x,\theta)$. Fix any $(x_0,\theta)\in X\times \Theta$ and define $x_n=T_\theta^n x_0$ for $n\in \N$. Then by the triangle inequality and the contraction property \eqref{eq:contraction}, we have
\begin{align*}
d(x_n,x_0)&\le d(x_n,x_{n-1})+\dots+d(x_1,x_0)\\
&\le (\beta^{n-1}+\dots+1)d(x_1,x_0)\\
&\le \frac{1}{1-\beta}d(x_1,x_0)=\frac{1}{1-\beta}d(T_\theta x_0,x_0).
\end{align*}
Letting $n\to\infty$ and noting that $x_n\to x^*(\theta)$ by the contraction mapping theorem, we have
\begin{equation}
d(x^*(\theta),x_0)\le \frac{1}{1-\beta}d(T_\theta x_0,x_0).\label{eq:dub}
\end{equation}
Since $x_0\in X$ and $\theta\in \Theta$ are arbitrary in \eqref{eq:dub}, set $\theta=\theta'$ and $x_0=x^*(\theta)$, where $\theta,\theta'\in \Theta$ are arbitrary. Since by definition $T_\theta x_0=x_0$, it follows that
\begin{align}
d(x^*(\theta'),x^*(\theta))&\le \frac{1}{1-\beta}d(T_{\theta'} x_0,x_0)=\frac{1}{1-\beta}d(T_{\theta'} x_0,T_\theta x_0)\notag \\
&=\frac{1}{1-\beta}d(T(x_0,\theta'),T(x_0,\theta)),\label{eq:dub1}
\end{align}
where $x_0=x^*(\theta)$. Since $T$ is continuous and $x_0=x^*(\theta)$ depends only on $\theta$, for any $\epsilon>0$ there exists an open neighborhood $U$ of $\theta$ such that
\begin{equation}
d(T(x_0,\theta'),T(x_0,\theta))<(1-\beta)\epsilon \label{eq:dub2}
\end{equation}
whenever $\theta'\in U$. Combining \eqref{eq:dub1} and \eqref{eq:dub2}, for all $\theta'\in U$ we have 
\begin{equation*}
d(x^*(\theta'),x^*(\theta))<\epsilon, 
\end{equation*}
so $x^*:\Theta\to X$ is continuous.
\end{proof}

\begin{proof}[Proof of Theorem \ref{thm:exist_gen}]
Let $\cV=bc(H\times Z)$ be the space of bounded continuous functions from $H\times Z$ to $\R$ equipped with the supremum norm. Then $(\cV,\norm{\cdot})$ is a Banach space. (Indeed, it is just the Euclidean space $\R^{HZ}$ because $H,Z$ are finite sets.) For $(V,\theta)\in \cV\times A^Z$, define
\begin{multline}
(T_\theta V)_h(z)\\
=
\max_{x\in A_h}\set{(1-\rho_h)u_h(x,\theta(z),z)+\rho_h\sum_{h',z'}p_{hh'}(x,\theta(z),z)q(\theta(z),z,z')V_{h'}(z')}.\label{eq:TVhz}
\end{multline}
A straightforward application of the maximum theorem implies that $(T_\theta V)_h(z)$ is continuous in $\theta$. Since $H$ is a finite set, we have $\rho\coloneqq \max_h \rho_h\in [0,1)$. It is then straightforward to verify \cite{blackwell1965}'s sufficient conditions, and for fixed $\theta \in A^Z$, the map $\cV \ni V \mapsto T_\theta V$ is a contraction mapping with modulus $\rho$. It follows from Lemma \ref{lem:fp_cont} that $T_\theta$ has a unique fixed point 
$V^*(\theta)=\set{V_h^*(z,\theta)}_{(h,z)\in H\times Z}\in \cV$, 
which is continuous in $\theta$.

Since by assumption $u_h(x,a,z)$ is continuous and strictly concave in $x$ and $p_{hh'}(x,a,z)$ is affine in $x$, the objective function inside the braces of \eqref{eq:TVhz} (where $V_h(z)=V_h^*(z,\theta)$) is continuous and strictly concave in $x$. Therefore there exists a unique maximizer, which we call $x=a_h^*(z,\theta)$. By the maximum theorem, $a_h^*(z,\theta)$ is continuous in $\theta$. Since $H,Z$ are finite sets, we may view $a^*:A^Z\to A^Z$ as a continuous map. Since $A$ is nonempty, compact, and convex, by Brouwer's fixed point theorem, $a^*$ has a fixed point. Letting $\theta=(\theta_h(z))$ be this fixed point, it is clear that all conditions in Definition \ref{defn:eq_gen} are satisfied.
\end{proof}

\begin{proof}[Proof of Theorem \ref{thm:exist}]
We apply Theorem \ref{thm:exist_gen}.

The agent type is denoted by $h\in H=\set{U,I_k,R_k,D}$, which is finite. The aggregate state is denoted by $z\in Z$ defined in \eqref{eq:Z}, which is finite. Define the action set of type $h$ agents by $A_h=[\ubar{a},1]$ if $h\neq D$ and $A_D=\set{0}$, which are nonempty compact convex subsets of $\R$. Define the utility function $u_h:A_h\times A\times Z\to \R$ by 
$u_U(x,a,z)=u_{R_k}(x,a,z)=u(x)$, $u_{I_k}(x,a,z)=u_I(x)$, and $u_D(x,a,z)=u_D$.
Note that each $u_h$ is continuous in $(x,a,z)$ and strictly concave in $x\in A_h$. (Although $u_D$ is a constant function, it is strictly concave in $x$ because its domain $A_D=\set{0}$ is a singleton.)

Define the transition probabilities of individual states $p_{hh'}(x,a,z)$ as follows. For $h=R_k,D$ agents, because they remain in their corresponding state forever, the transition probabilities are 0 or 1, which are clearly continuous in all variables and affine in $x$. For $I_k$ agents, by assumption they die with probability $\gamma\Delta \delta$. Hence $p_{I_kD}(x,a,z)=\gamma \Delta \delta$, which is continuous in all variables and constant (hence affine) in $x$. The same is true for $p_{I_kh'}$ for any $h'\in H$. By \eqref{eq:infect_prob} and \eqref{eq:Bellman.u2}, the transition probabilities of $U$ agents are
\begin{align*}
p_{UU}(x,a,z)&=\e^{-\nu\Delta}(1-\sigma\mu(z)\beta\Delta (a_{I_k}I_k+a_UI_u)x),\\
p_{UI_k}(x,a,z)&=\sigma\mu(z)\beta\Delta (a_{I_k}I_k+a_UI_u)x,\\
p_{UR_k}(x,a,z)&=(1-\e^{-\nu\Delta})(1-\sigma\mu(z)\beta\Delta (a_{I_k}I_k+a_UI_u)x),\\
p_{UD}(x,a,z)&=0,
\end{align*}
which are continuous in all variables and affine in $x$. Finally, the transition probability for the aggregate state $q:A\times Z\times Z\to [0,1]$ is clearly continuous because it is determined by the current state $z\in Z$ (which is finite) and individual's actions, whose transition probabilities are all continuous. The existence of equilibrium in the sense of Definition \ref{defn:eq_gen} then follows from Theorem \ref{thm:exist_gen}. The resulting value and policy functions clearly satisfy Definition \ref{defn:eq}. Because $U$ agents take identical actions, their beliefs always satisfy \eqref{eq:muz}.
\end{proof}

\begin{proof}[Proof of Theorem \ref{thm:aU_se}]
Fix $z\in Z$ and let $\mu=\mu(z)$ and $a_U^*=a_U^*(z)$. Define the functions $f,g:A\to \R$ by
\begin{align*}
f(x)&=(1-\e^{-r\Delta})u(x)+\e^{-r\Delta}\E_z((1-\sigma\mu p(a_U^*)x)\e^{-\nu\Delta}V_U(z')+\sigma\mu p(a_U^*)x V_{I_k}(z')),\\
g(x)&=(1-\e^{-r\Delta})u(x)+\e^{-r\Delta}\E_z((1-\sigma\mu p(x)x)\e^{-\nu\Delta}V_U(z')+\sigma\mu p(x)x V_{I_k}(z')),
\end{align*}
where (with a slight abuse of notation) $p(x)\coloneqq \beta\Delta(a_II_k+xI_u)$. Note that by Definition \ref{defn:eq}, the maximum of $f$ is precisely $a_U^*$, and since the last term of \eqref{eq:Waz} does not depend on $a_U^*$, by Definition \ref{defn:se} the maximum of $g$ is precisely $a_U^\dagger=a_U^\dagger(z)$. Our proof relies on careful comparisons of the derivatives of $f$ and $g$ together with some case-by-case analysis. To this end, first note that
\begin{align*}
f'(x)&=(1-\e^{-r\Delta})u'(x)-\e^{-r\Delta}\sigma\beta\Delta\mu (a_II_k+a_U^*I_u) \E_z(\e^{-\nu\Delta}V_U(z')-V_{I_k}),\\
g'(x)&=(1-\e^{-r\Delta})u'(x)-\e^{-r\Delta}\sigma\beta\Delta\mu (a_II_k+2xI_u) \E_z(\e^{-\nu\Delta}V_U(z')-V_{I_k}),\\
f''(x)&=(1-\e^{-r\Delta})u''(x)<0,\\
g''(x)&=(1-\e^{-r\Delta})u''(x)-2\e^{-r\Delta}\sigma\beta\Delta\mu I_u\E_z(\e^{-\nu\Delta}V_U(z')-V_{I_k})<0,
\end{align*}
where the last two inequalities use $u''<0$ and $\e^{-\nu\Delta}V_U(z')\ge V_U(z')>V_{I_k}$, which follows from Proposition \ref{prop:V} (recall that $V_U,V_{I_k} < 0$). Since $A=[\ubar{a},1]$ is nonempty, compact, convex, and $f,g$ are continuous and strictly concave, they achieve unique maxima. Below, we divide the proof into four steps.

\begin{step}
If $I_u=0$, then $a_U^\dagger=a_U^*$.
\end{step}

To simplify notation, let
\begin{equation*}
C\coloneqq \e^{-r\Delta}\sigma\beta\Delta\mu\E_z(\e^{-\nu\Delta}V_U(z')-V_{I_k})>0.
\end{equation*}
Then we can rewrite the derivatives of $f,g$ as
\begin{align*}
f'(x)&=(1-\e^{-r\Delta})u'(x)-C(a_II_k+a_U^*I_u),\\
g'(x)&=(1-\e^{-r\Delta})u'(x)-C(a_II_k+2xI_u).
\end{align*}
If $I_u=0$ (which holds if $\sigma=1$), then $f'=g'$, so the maxima of $f,g$ agree and $a_U^*=a_U^\dagger$. \qedsymbol

In what follows, without loss of generality we may assume $I_u>0$.

\begin{step}
$a_U^\dagger\le a_U^*$ holds.
\end{step}

Note that at $x=a_U^*$, we have
\begin{align*}
f'(a_U^*)&=(1-\e^{-r\Delta})u'(a_U^*)-C(a_II_k+a_U^*I_u)\\
&\ge (1-\e^{-r\Delta})u'(a_U^*)-C(a_II_k+2a_U^*I_u)=g'(a_U^*).
\end{align*}
If $f'(a_U^*)\le 0$, then $g'(a_U^*)\le f'(a_U^*)\le 0$. Since $g$ is strictly concave and $a_U^\dagger$ achieves its maximum, it follows that $a_U^\dagger\le a_U^*$. If $f'(a_U^*)>0$, then since $f$ is strictly concave, it follows that $a_U^*=1$ and hence $a_U^\dagger\le 1=a_U^*$. In either case, we have $a_U^\dagger\le a_U^*$. \qedsymbol

\begin{step}
$a_U^*/2\le a_U^\dagger$ and \eqref{eq:aU_bound2} hold.
\end{step}
If $f'(a_U^*)<0$, then since $f$ is strictly concave, we must have $a_U^*=\ubar{a}$. Then $\ubar{a}\le a_U^\dagger\le a_U^*=\ubar{a}$, so $a_U^*-a_U^\dagger=0$ and the claim is trivial. Similarly, if $g'(a_U^\dagger)>0$, then since $g$ is strictly concave, it must be $a_U^\dagger=1$. Then $1=a_U^\dagger\le a_U^*\le 1$, so $a_U^*-a_U^\dagger=0$ and the claim is trivial. Therefore without loss of generality we may assume $f'(a_U^*)\ge 0\ge g'(a_U^\dagger)$.

Under this condition, we obtain
\begin{align*}
0&\le f'(a_U^*)=(1-\e^{-r\Delta})u'(a_U^*)-C(a_II_k+a_U^*I_u),\\
0&\ge g'(a_U^\dagger)=(1-\e^{-r\Delta})u'(a_U^\dagger)-C(a_II_k+2a_U^\dagger I_u).
\end{align*}
Taking the difference, we obtain
\begin{equation*}
0\le (1-\e^{-r\Delta})(u'(a_U^*)-u'(a_U^\dagger))+C(2a_U^\dagger-a_U^*)I_u.
\end{equation*}
Applying the mean value theorem to $u'$, there exists $a\in [a_U^\dagger,a_U^*]$ such that
\begin{equation}
-(1-\e^{-r\Delta})u''(a)(a_U^*-a_U^\dagger)\le C(2a_U^\dagger-a_U^*)I_u.\label{eq:u''bound}
\end{equation}
Since $u''<0$, $a_U^*-a_U^\dagger\ge 0$, and $I_u>0$, it follows from \eqref{eq:u''bound} that $2a_U^\dagger-a_U^*\ge 0$, which implies $a_U^*/2\le a_U^\dagger$. Finally, noting that $m=\min_{a\in A}\abs{u''(a)}$, it follows from \eqref{eq:u''bound} that
\begin{equation*}
(1-\e^{-r\Delta})m(a_U^*-a_U^\dagger)\le C(2a_U^\dagger-a_U^*)I_u,
\end{equation*}
which is equivalent to \eqref{eq:aU_bound2}. \qedsymbol

\begin{step}
If $a_U^*,a_U^\dagger$ are interior, then \eqref{eq:aU_bound3} holds.
\end{step}

If $a_U^*,a_U^\dagger$ are interior, then $f'(a_U^*)=g'(a_U^\dagger)=0$. Then \eqref{eq:u''bound} holds with equality. Noting that $M=\max_{a\in A}\abs{u''(a)}$, we obtain \eqref{eq:aU_bound3} by a similar argument.
\end{proof}

\newpage

\begin{center}
\textbf{\Huge Online Appendix}
\end{center}

\section{Observational equivalence between single and multiple signals}\label{sec:signal}


In Section \ref{sec:model} we supposed that an infected agent is informed of the infection with probability $\sigma$ and that a known (unknown) infected agent dies with probability $\delta$ ($0$). In this Appendix we elaborate upon this assumption and show that it can model the case in which there are multiple signals.

Suppose now that there are multiple signals that may be received upon infection, indexed $j=1,\dots,J$. Let $\sigma_j>0$ be the probability of receiving signal $j$ upon infection, with $\sigma\coloneqq \sum_{j=1}^J\sigma_j\in (0,1]$. Let $\delta_j\in [0,1]$ be the fatality rate conditional on receiving signal $j$ and $\delta\coloneqq (\sum_{j=1}^J\sigma_j\delta_j)/\sigma$ be the expected fatality conditional on receiving any signal. For instance, the signal could encode the result of a laboratory test as well as the type and severity of symptoms. Suppose a type $j$ known infected agent takes action $a_j$ with associated flow utility $u_j$. Then by \eqref{eq:VIk}, the value function of a type $j$ agent is
\begin{equation*}
V_j\coloneqq \frac{(1-\e^{-r\Delta})u_j+\e^{-r\Delta}\gamma\Delta\delta_j u_D}{1-\e^{-r\Delta}(1-\gamma\Delta)}.
\end{equation*}
Therefore the expected continuation value of being known infected is
\begin{equation}
V_{I_k}\coloneqq \frac{1}{\sigma}\sum_{j=1}^J\sigma_j V_j=\frac{(1-\e^{-r\Delta})u_I+\e^{-r\Delta}\gamma\Delta\delta u_D}{1-\e^{-r\Delta}(1-\gamma\Delta)},\label{eq:VIkj}
\end{equation}
where $u_I\coloneqq(\sum_{j=1}^J\sigma_ju_j)/\sigma$ is the expected flow utility of being known infected. Because \eqref{eq:VIkj} is identical to \eqref{eq:VIk} and the continuation values $\set{V_j}$ affect the behavior of unknown agents only through $V_{I_k}$ due to expected utility maximization, a model with multiple signals and heterogeneous fatality is observationally equivalent to a model with a single signal and uniform fatality. The flow utility $u_I$ and action $a_I$ in Section \ref{sec:model} can be interpreted as the average across those of known infected agents.

\section{Continuous-time recursive formulations} \label{sec:recAPP}

In the main text we considered a model in which the number of agents was finite and time was discrete. However, to make the model computationally tractable, when producing the figures in the main text we assume a continuous-time model with a continuum of agents. In this appendix we consider three allocations: 
\begin{enumerate}
\item the efficient allocation (denoted SPP for ``social planner's problem'');
\item Perfect Bayesian equilibria (denoted PBE);
\item Perfect recall Markov equilibria (denoted ME). 
\end{enumerate}
We first define the above allocations and formulate the individual decisions problems recursively. The numerical algorithm used to solve each of the above will be outlined in Appendix \ref{sec:solve}. 

\subsection{Simplifications of the state space} \label{sec:state_simp}

Prior to the separate analysis of the above problems, we outline some simplifications common to all three. Note that in the model of the main text, the state variable for the economy is the sextuple $(S,I_u,I_k,R_u,R_k,D)$. However, by using the fact that an exogenous fraction $\sigma$ of the total number of infected agents $I \coloneqq I_u + I_k$ develop symptoms, we have the simplifications
\begin{subequations}\label{eq:SIsimp}
\begin{align}
(I_u, I_k) & = ((1-\sigma) I, \sigma I) \label{eq:SIsimp.I}
\\ (R_u, R_k+D) & = ((1-\sigma) (R+D), \sigma (R+D)) \label{eq:SIsimp.RD}
\end{align}
\end{subequations}
and hence
\begin{equation}
S + I_u + R_u  = S + (1-\sigma) I + (1-\sigma) (1 - S - I) = \sigma S + 1 - \sigma.
\label{SIuRu}
\end{equation}
The fraction of susceptible agents who do not exhibit symptoms (which is also the belief of an unknown agent that they are susceptible) is given by
\begin{equation}
\mu \coloneqq \frac{S}{\sigma S + 1 - \sigma}.
\label{eq:muAPP}
\end{equation}
Section \ref{subsec:dynamics} shows that when we assume a continuum of agents, the law of motion of $S$ and $I$ becomes deterministic and is given by equations \eqref{eq:SI}. Because the behavior of $R_k,R_u,D$ agents does not affect state transitions, with a slight abuse of notation we can define the minimal state space by
\begin{equation}
Z=\set{(S,I)|S\ge 0, I\ge 0, S+I\le 1} \label{eq:Z_min}
\end{equation}
and so we henceforth write the action $\tilde{a}(S,I)$ of unknown agents as a function of $(S,I)$ rather than the entire state $z$. When unknown agents all adhere to the policy function $\tilde{a}(S,I)$, the continuous-time law of motion of $S$ and $I$ is given by
\begin{subequations}\label{eq:SIcont}
\begin{align}
\dot{S} & = -\beta \tilde{a}(S,I)S(\sigma a_{I_k}^*+(1-\sigma)\tilde{a}(S,I))I, \label{eq:SIcont.s}
\\ \dot{I} & = (\beta \tilde{a}(S,I)S(\sigma a_{I_k}^*+(1-\sigma)\tilde{a}(S,I)-\gamma)I  \label{eq:SIcont.i}
\end{align}
\end{subequations}
where the dot notation indicates derivative with respect to time. When known infected agents  adhere to $a_{I_k}^*$ and unknown infected agents adhere to the policy function $\tilde{a}$, a susceptible agent taking action $a$ becomes infected at rate
\begin{equation}
\beta a(\sigma a_{I_k}^*+(1-\sigma)\tilde{a}(S,I))I.
\label{ind_hazard}
\end{equation}
Consequently, an unknown agent taking action $a_U$ develops symptoms at rate 
\begin{equation}
\sigma \mu \beta a_U(\sigma a_{I_k}^*+(1-\sigma)\tilde{a}(S,I))I
\label{ind_hazard_sympt}
\end{equation}
where $\mu$ is given by \eqref{eq:muAPP}. 

\subsection{Social planner's problem}

In this section we derive the form of the social planner's problem in preparation for the numerical analysis in Appendix \ref{sec:solve}. The flow payoffs experienced by the planner when the state is $(S,I_u,I_k,R_u,R_k,D)$ and the action taken by the unknown agents is $a_U$ are given by
$$
F(S,I_u,I_k,R_u,R_k,D,a_U) = r{\left[Du_D + (S + I_u + R_u)u(a_U) + I_ku(a_{I_k}^*)\right]} 
$$
because $a=1$ for recovered agents and $u(1) = 0$. Using the simplifications in Appendix \ref{sec:state_simp}, the payoff to the planner may be written
$$
r(Du_D + (\sigma S + 1-\sigma)u(a_U) + \sigma Iu(a_{I_k}^*)).
$$ 
The value function of the planner conditional on no vaccine arrival solves the Hamilton-Jacobi-Bellman equation
\begin{equation}
\begin{aligned} 
rW(S,I,D) & = rDu_D + \nu[\overline{W}(I,D) - W(S,I,D)] + \gamma\delta\sigma I \partial_D W(S,I,D)
\\ & + \max_{\substack{a_U\in [\underline{a}, 1]} } r{\left((\sigma S + 1-\sigma )u(a_U) + \sigma Iu(a_{I_k}^*)\right)}
\\ & - \beta Sa_U{\left((1-\sigma) a_U + \sigma a_{I_k}^*\right)}I\partial_SW(S,I,D) 
\\ & + [\beta Sa_U{\left((1-\sigma) a_U + \sigma a_{I_k}^* \right)} - \gamma ]I \partial_I W(S,I,D)
\end{aligned}
\label{Wnovac}
\end{equation}
where $\overline{W}(I,D)$ is the value function of the planner after the arrival of the vaccine, which solves the Hamilton-Jacobi-Bellman equation
\begin{equation} 
r\overline{W}(I,D) = rDu_D + r \sigma Iu(a_{I_k}^*) - \gamma I\partial_I\overline{W}(I,D) + \gamma\delta\sigma I \partial_D \overline{W}(I,D).
\label{Wbar}
\end{equation} 
\noindent The following shows that the above equations admit further simplifications, because the value function after the arrival of the vaccine is linear in the population shares of infected and dead agents.  

\begin{lem}\label{LemmaFORM}
The value function of the planner after the arrival of the vaccine is $\overline{W}(I,D) = Du_D - \overline{C}_{\textnormal{vac}}I$, where the function $C$ solves the Hamilton-Jacobi-Bellman equation
\begin{equation}
\overline{C}_{\textnormal{vac}} = \frac{\sigma}{r+\gamma}(-\gamma\delta u_D - r u(a_{I_k}^*)).
\label{Vvac}
\end{equation}
The value function prior to the arrival of the vaccine may be written as $W(S,I,D) = Du_D - C(S,I)$, where 
\begin{align*}
(r+\nu)C(S,I) = & \min_{\substack{a_U\in [\underline{a}, 1]} } r{\left((\sigma S + 1-\sigma )[-u(a_U)] + \sigma I[-u(a_{I_k}^*)]\right)} + \nu \overline{C}_{\textnormal{vac}}I
\\ & + \gamma\delta\sigma I[-u_D] -\beta SIa_U{\left((1-\sigma) a_U + \sigma a_{I_k}^*\right)}\partial_SC(S,I)
\\ & + [\beta Sa_U{\left((1-\sigma) a_U + \sigma a_{I_k}^*\right)} - \gamma ]I \partial_IC(S,I).
\end{align*}
\end{lem}

\begin{proof}
Substituting the assumed form $\overline{W}(I,D) = Du_D - \overline{C}_{\textnormal{vac}}I$ into \eqref{Wbar} gives 
\begin{align*}
rDu_D - r\overline{C}_{\textnormal{vac}}I & = rDu_D + r\sigma Iu(a_{I_k}^*) + \gamma \overline{C}_{\textnormal{vac}} I + \gamma\delta\sigma I u_D
\end{align*} 
which rearranges to \eqref{Vvac}. Substituting $W(S,I,D) = Du_D - C(S,I)$ into \eqref{Wnovac} gives 
\begin{align*}
r(Du_D - C(S,I)) = & rDu_D + \nu[Du_D - \overline{C}_{\textnormal{vac}}I - Du_D + C(S,I)] + \gamma\delta\sigma I u_D
\\ & + \max_{\substack{a_U\in [\underline{a}, 1]} } r{\left((\sigma S + 1-\sigma )u(a_U) + \sigma Iu(a_{I_k}^*)\right)}
\\ & - \beta Sa_U{\left((1-\sigma) a_U + \sigma a_{I_k}^*\right)}I[-\partial_SC(S,I)]
\\ & + [\beta Sa_U{\left((1-\sigma) a_U + \sigma a_{I_k}^* \right)} - \gamma ]I[-\partial_IC(S,I)]
\end{align*}
from which the claim follows upon simplification.
\end{proof}

The function $C$ defined in Lemma \ref{LemmaFORM} has a natural interpretation as the ``cost of the pandemic'' in terms of utility (omitting the role of deaths) and is the solution to a cost minimization problem with discount rate $r+\nu$ and flow cost
\begin{equation}
\begin{aligned}
r{\left((\sigma S + 1 - \sigma)[-u(a_U)] + \sigma I[-u(a_{I_k}^*)]\right)} - \gamma\delta\sigma Iu_D + \nu \overline{C}_{\textnormal{vac}}I
\end{aligned}
\label{flowCOST2}
\end{equation}
where the state variable $(S,I)$ evolves according to 
\begin{equation}
\begin{aligned}
\dot{S} & = -\beta SIa_U ((1-\sigma) a_U + \sigma a_{I_k}^*)
\\ \dot{I} & = \beta SIa_U ((1-\sigma) a_U + \sigma a_{I_k}^*) - \gamma I.
\end{aligned}
\label{FinalSYS} 
\end{equation}
To produce the figures for the social planner's problem we solve a continuous-time control problem with state $(S,I)$, flow payoff \eqref{flowCOST2}, discount rate $r+ \nu$ and law of motion for the state variables \eqref{FinalSYS}.

\subsection{Perfect Bayesian equilibrium} \label{PBE_app}

We now state the recursive formulations of the problem of the unknown agents in the Perfect Bayesian equilibria. Let $V_h(S,I)$ be the value function of type $h=U,I_k$ agents, $\tilde{a}(S,I)$ be the policy function of unknown agents, and let partial derivatives be denoted by $\partial_S$ etc. In this case the Hamilton-Jacobi-Bellman equation for the unknown agent becomes
\begin{align}
(r+\nu)V_U(S, I) = & \max_{a_U \in [\ubar{a},1]} ru(a_U) \notag \\
& + \sigma \mu \beta a_U{\left((1-\sigma) \tilde{a}(S,I) + \sigma a_{I_k}^* \right)}I[V_{I_k}(S,I) - V_U(S,I)] \notag \\
&  -\beta S\tilde{a}(S,I){\left((1-\sigma) \tilde{a}(S,I) + \sigma a_{I_k}^* \right)} I\partial_S V_U(S,I) \notag \\
& + {\left(\beta S\tilde{a}(S,I){\left((1-\sigma) \tilde{a}(S,I) + \sigma a_{I_k}^* \right)} - \gamma \right)}I\partial_I V_U(S,I).
\label{PBEvalue2}
\end{align}

\subsection{Perfect recall Markov equilibrium}\label{CE2new}

We now consider an extension of the equilibrium concept explored in the main text in which we allow the agents to remember their previous actions so the beliefs of the agents enter as a separate state variable. In contrast with the Perfect Bayesian equilibrium notion in the main text, the evolution of an individual's belief will now depend directly upon their own actions. The evolution of the aggregate state variables and the value functions of the known infected, recovered, and dead agents are unchanged relative to the perfect Bayesian equilibrium, and we must therefore only alter the problem of the unknown agents. Now \eqref{eq:muz} will remain their belief \emph{in equilibrium}, but not if they deviated from the equilibrium action in the past.

If the average action of unknown agents is $\tilde{a}$, an agent with belief $\mu$ who chooses activity $a$ will believe that if susceptible, they become infected at rate $\beta a((1-\sigma)\tilde{a} + \sigma a^*_{I_k})I$. Their belief that they are susceptible is then the probability they were susceptible in the previous period multiplied by the probability they were not infected without diagnosis during the last period, or
\begin{equation}
\mu_{+\Delta} = \mu(1 - \Delta(1-\sigma) \beta a((1-\sigma) \tilde{a} + \sigma a^*_{I_k})I).
\label{beliefEVOL}
\end{equation} 
We then define a \emph{perfect recall Markov equilibrium} as an allocation in which the unknown agents optimize given the law of motion of the population shares \emph{and} their beliefs, and these beliefs are consistent with equilibrium behavior and Bayes' rule. Formally, we adopt the continuous-time formulation used in the numerical analysis and proceed as follows.  Rearranging \eqref{beliefEVOL} and sending $\Delta \to 0$ gives the evolution of beliefs
\begin{equation}
\dot{\mu} = - \mu(1-\sigma) \beta a{\left((1-\sigma) \tilde{a}(S,I) + \sigma a_{I_k}^* \right)}I.
\label{rhoLAW}
\end{equation}
The state variables for the unknown agents are now their health status and $(S,I,\mu)$, with the evolution of $(S,I)$ and $\mu$ given by \eqref{FinalSYS} and \eqref{rhoLAW}, respectively. Denote the value function of unknown agents by $V_U(S,I,\mu)$ and the value function of known infected agents by $V_{I_k}(S,I)$.  The value function for symptomatic agents is independent of the aggregate state variables, since their corresponding Hamilton-Jacobi-Bellman equation is 
$$
rV_{I_k}(S,I) = ru(a_{I_k}^*) + \gamma {\left(\delta[u_D - V_{I_k}(S,I)] + (1-\delta)[-V_{I_k}(S,I)] \right)},
$$
which rearranges to 
\begin{equation}
V_{I_k}(S,I) = \frac{1}{r+\gamma}[ru(a_{I_k}^*) + \gamma \delta u_D].
\label{Uvac}
\end{equation}
The Hamilton-Jacobi-Bellman for unknown agents is
\begin{equation}
\begin{aligned} 
(r+\nu)V_U(S,I,\mu) = & \max_{a_U \in [\underline{a},1]} ru(a_U) 
\\ & + \sigma \mu \beta a_U{\left((1-\sigma) \tilde{a}(S,I) + \sigma a_{I_k}^* \right)}I[V_{I_k}(S,I,\mu) - V_U(S,I,\mu)]
\\ &  -\beta S\tilde{a}(S,I){\left((1-\sigma) \tilde{a}(S,I) + \sigma a_{I_k}^* \right)} I \partial_S V_U(S,I,\mu)
\\ & + {\left(\beta S\tilde{a}(S,I){\left((1-\sigma) \tilde{a}(S,I) + \sigma a_{I_k}^* \right)} - \gamma \right)}I\partial_I V_U(S,I,\mu)
\\ & - \mu(1-\sigma)\beta a_U{\left((1-\sigma) \tilde{a}(S,I) + \sigma a_{I_k}^* \right)}I\partial_{\mu}V_U(S,I,\mu).
\end{aligned} 
\label{AGENTprob}
\end{equation}

Given an average action $\tilde{a}$ of unknown agents, there is an associated policy function $a(S,I,\mu; \tilde{a})$ solving the problem of a given unknown agent. Now define an operator $J(\tilde{a})(S,I) = a(S,I,S/[1-\sigma + \sigma S]; \tilde{a}) - \tilde{a}$. The equilibrium notion we adopt in this section is then that of a Markov perfect equilibrium, in which all agents solve their individual problems taking the aggregate law of motion as given, and the associated law of motion is consistent with individual behavior. 

\begin{defn}\label{Mdef}
A \emph{Perfect recall Markov equilibrium} (or Markov equilibrium for short) consists of value functions $V_U(S,I,\mu)$ and $V_{I_k}(S,I)$ for unknown agents and known infected agents together with a policy function $a(S,I,\mu)$ for unknown agents such that:
\begin{itemize}
\item The functions $V_U(S,I,\mu)$, $a(S,I,\mu)$ and $V_{I_k}(S,I)$ solve the problems of the unknown agents and known infected agents, respectively.
\item The law of motion of the aggregate state is consistent with the policy function of the unknown agents, or $J(\tilde{a}) = 0$.
\end{itemize}
\end{defn}

As Figure \ref{fig:PRME_PBE} shows, in the perfect recall Markov equilibrium agents take higher actions everywhere in the state space and the length of the pandemic is reduced. However, the difference in actions between the perfect Bayesian and perfect recall equilibria are small in the relevant part of the state space and the two models provide similar predictions.

\begin{figure}[!htb]
\centering
\includegraphics[width=0.48\linewidth]{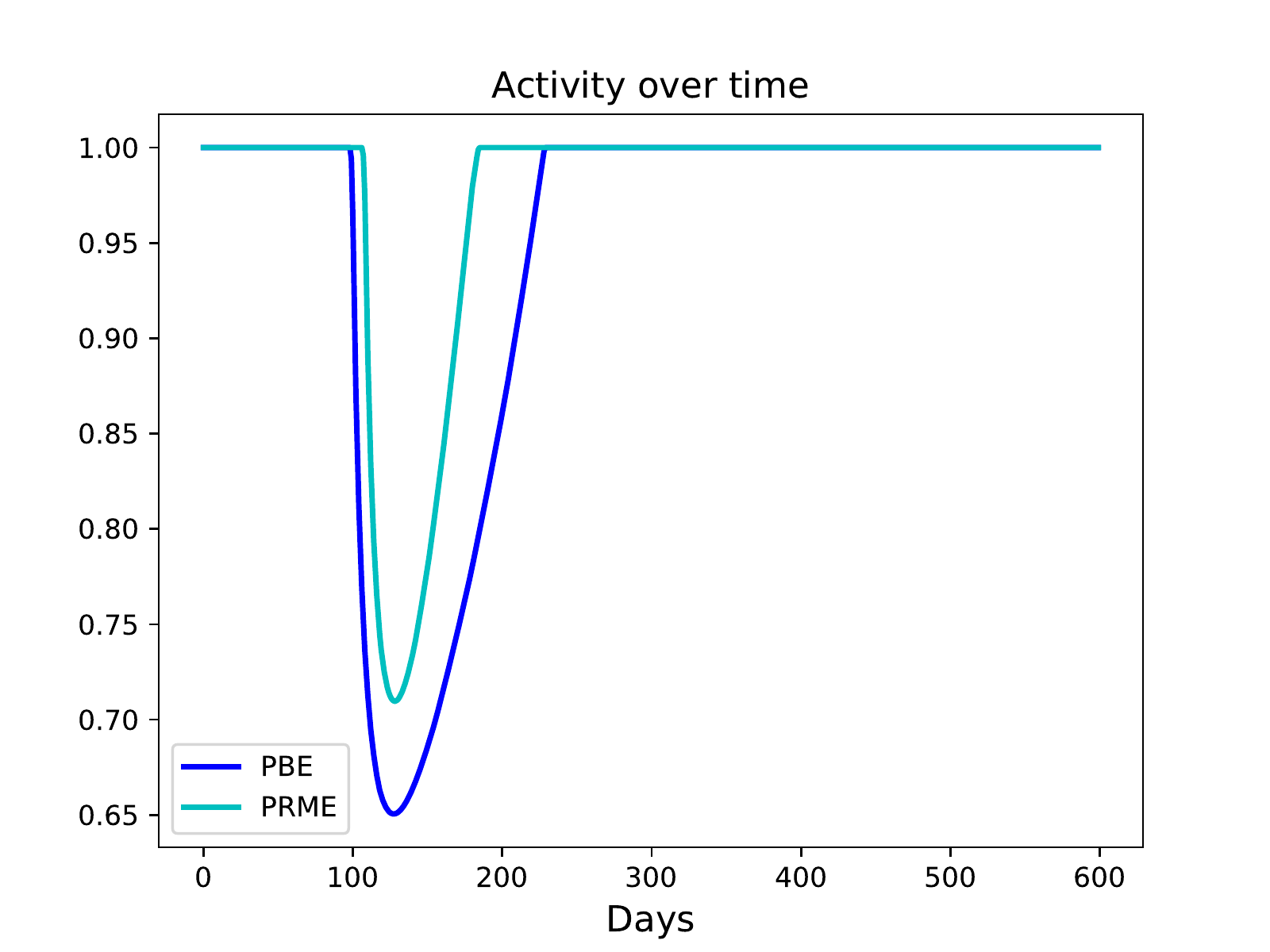}
\includegraphics[width=0.48\linewidth]{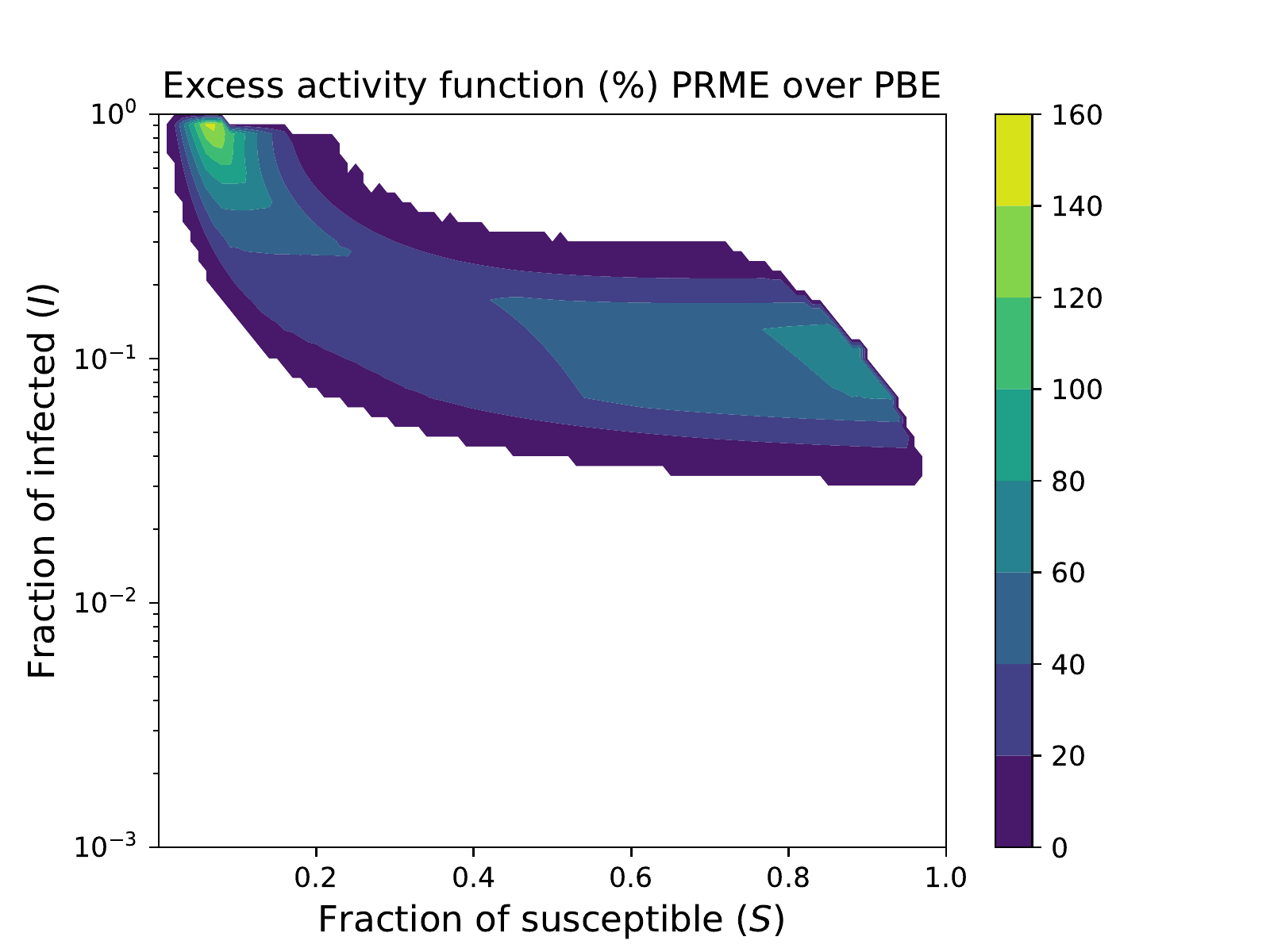}
\caption{Comparison of activity levels in perfect Bayesian and perfect recall Markov equilibria.}\label{fig:PRME_PBE}
\end{figure}

\section{Solution algorithm}\label{sec:solve}

In this section we outline the numerical method used to produce the figures in the main text. As noted in Appendix \ref{sec:recAPP}, we consider the continuous-time limit of our model with a continuum of agents and apply the finite-state Markov chain method of \cite{KushnerDupuis1992}. The local consistency requirements for the aggregate state are given by
\begin{subequations}\label{LC3}
\begin{align}
\E[\Delta S] & = - \Delta_t\beta SIa_U{\left((1-\sigma) a_U + \sigma a_{I_k}^*\right)} + o(\Delta_t)
\\ \E[\Delta I] & = \Delta_t{\left(\beta SIa_U{\left((1-\sigma) a_U + \sigma a_{I_k}^*\right)} -  \gamma I\right)} + o(\Delta_t).
\end{align}
\end{subequations}

\subsection{Exponential grid}\label{subsec:solve.grid}

When solving for an SIR model numerically, because the fraction of infected agents $I$ varies by many orders of magnitude (say between $10^{-6}$ and $10^{-1}$), it is important to use a grid that properly covers the relevant range of the state space, such as an exponential grid. Below, we describe the exponential grid proposed by \cite{Gouin-BonenfantToda2022QE}. In general, suppose we would like to construct an $N$-point exponential grid on a given interval $(a,b)$, where the lower bound $a$ is not necessarily positive. A natural idea to deal with such a case is as follows.

\begin{framed}
\begin{oneshot}[Constructing exponential grid]\quad 
\begin{enumerate}
\item Choose a shift parameter $s>-a$.
\item Construct an $N$-point evenly-spaced grid on $(\log(a+s),\log(b+s))$.
\item Take the exponential and subtract $s$.
\end{enumerate}
\end{oneshot}
\end{framed}

The remaining question is how to choose the shift parameter $s$. Suppose we would like to specify the median grid point as $c\in (a,b)$. Since the median of the evenly-spaced grid on $(\log(a+s),\log(b+s))$ is $\frac{1}{2}(\log(a+s)+\log(b+s))$, we need to take $s>-a$ such that
\begin{align*}
c=\exp\left(\frac{1}{2}(\log(a+s)+\log(b+s))\right)-s
&\iff s=\frac{c^2-ab}{a+b-2c}.
\end{align*}
Note that in this case $s+a=\frac{c^2-ab}{a+b-2c}+a=\frac{(c-a)^2}{a+b-2c}$, so $s+a$ is positive if and only if $c<\frac{a+b}{2}$. Therefore, for any $c\in \left(a,\frac{a+b}{2}\right)$, it is possible to construct an exponentially-spaced grid with end points $(a,b)$ and median point $c$.

To solve the model in Section \ref{sec:num} numerically, we construct finite grids for $S,I$ and define the minimal state space $Z$ in \eqref{eq:Z_min} by the Cartesian product of the $S,I$ grids. For the $S$-space (fraction of susceptible agents), we use a 100-point uniform grid on $[10^{-8},1]$. For the $I$-space we use a 400-point exponential grid on $[10^{-8},1]$ with a median point of $10^{-4}$. We now write these grids as $\Sigma_S= \set{0,1/N_S,\dots,1-1/N_S,1}$ and $\Sigma_I=\set{I_0,I_1,\dots,I_{N_I}}$, and define
\begin{align*}
\Delta_{Ii}^- & = I_i - I_{i-1} &  i & =1,\dots, N_I 
\\
\Delta_{Ii}^+ & = I_{i+1} - I_{i} & i & =0,\dots, N_I-1
\end{align*}
and declare $\Delta_{I0}^{-} = \Delta_{IN_I}^{+} = 0$ and $\Delta_S = 1/N_S$. We then write $\Sigma \coloneqq \Sigma_S \times \Sigma_I$. This will be the state space for both the social planner's problem and the problem of the unknown agent in the perfect Bayesian equilibrium. 

\subsection{Social planner's problem}\label{subsec:solve.e}

We now describe the numerical method we use for computing the solution to the social planner's problem. The value function of the planner is of the form $W(S,I,D) = Du_D - C(S,I)$, where $C$ is given in Lemma \ref{LemmaFORM}. We first construct a locally consistent Markov chain for the law of motion of the state variables. For an arbitrary $(S,I) \in \Sigma$ there are three possible transitions, to the points $(S-\Delta_S,I), (S,I-\Delta_I^-)$ and $(S,I+\Delta_I^+)$, with associated probabilities denoted $p^{-S}, p^{-I}$ and $p^{+I}$. The local consistency requirements \eqref{LC3} become
\begin{align*}
 - \Delta_Sp^{-S} & = -\Delta_t \beta SI a_U{\left((1-\sigma) a_U + \sigma a_{I_k}^*\right)} + o(\Delta_t),
\\
- \Delta_I^-p^{-I} + \Delta_I^+p^{+I} & = \Delta_t{\left(\beta Sa_U{\left((1-\sigma) a_U + \sigma a_{I_k}^*\right)} -  \gamma\right)} I + o(\Delta_t).
\end{align*}
Inspection reveals it will suffice to set 
\begin{subequations}\label{probSmain2}
\begin{align}
p^{-S} & = \frac{\Delta_t}{\Delta_S}\beta SIa_U{\left((1-\sigma) a_U + \sigma a_{I_k}^*\right)},
\\
p^{+ I} & = \frac{\Delta_t}{\Delta_I^+} \beta SIa_U{\left((1-\sigma) a_U + \sigma a_{I_k}^*\right)},
\\
p^{- I} & = \frac{\Delta_t}{\Delta_I^-} \gamma I.
\end{align} 
\end{subequations} 
To ensure the chain remains on the grid we declare that at $I=1$ we have $a_U \le \hat{a}(S)$, the point at which $0 = \beta Sa_U((1-\sigma) a_U + \sigma a_{I_k}^*) - \gamma$. For $S>0$ we have
\begin{equation*}
\hat{a}(S) = \frac{1}{2\beta S(1-\sigma)}{\left(-\beta S\sigma a_{I_k}^* + \sqrt{[\beta S\sigma a_{I_k}^*]^2 + 4\gamma \beta S(1-\sigma)} \right)}. 
\end{equation*}
We then have the Bellman equation 
\begin{align*}
C(S,I) = & \nu \Delta_t\bar{C}_{\mathrm{vac}}I + \min_{a_U\in [\ubar{a}, 1]}-\Delta_t r{\left((1-\sigma+\sigma S)u(a_U) + \sigma I u(a_{I_k}^*)\right)}  - \Delta_t \gamma\delta\sigma Iu_D 
\\
& + \e^{-(r+\nu)\Delta_t} {\left(p^{-S}C(S-\Delta_S,I) + p^{+I}C(S,I+\Delta_I^+) + p^{-I}C(S,I-\Delta_I^-)\right)} \\
& + \e^{-(r+\nu)\Delta_t} (1-p^{-S}-p^{+I} -p^{-I})C(S,I).
\end{align*}
Omitting terms independent of the control, using \eqref{probSmain2} and sending $\Delta_t \to 0$ gives
\begin{align*}
\min_{a_U\in [\ubar{a}, 1]} & - r(1-\sigma+\sigma S)u(a_U) + \beta SIa_U{\left((1-\sigma) a_U + \sigma a_{I_k}^*\right)}[-C^{BS}] \\
&  + \beta SIa_U{\left((1-\sigma) a_U + \sigma a_{I_k}^*\right)} C^{FI},
\end{align*}
where we abbreviated 
\begin{align*}
C^{BS} & = \frac{1}{\Delta_S}(C(S,I) - C(S-\Delta_S,I)),
\\ C^{FI} & = \frac{1}{\Delta_I^+}(C(S,I+\Delta_I^+) - C(S,I)),
\end{align*}
and the superscripts stand for ``backward in $S$'' and ``forward in $I$'', respectively. Dividing by $-r(1-\sigma+\sigma S)$ gives
\begin{equation}
\max_{a_U\in [\ubar{a}, 1]} u(a_U) + \beta Sa_U{\left((1-\sigma) a_U + \sigma a_{I_k}^*\right)}\frac{I(C^{BS}-C^{FI})}{r(1-\sigma+\sigma S)}.
\label{obj2}
\end{equation}
Note that \eqref{obj2} is of the form
\begin{equation}
G(a,b,c) \coloneqq \max_{a_U \in [a,b]} u(a_U) + ca_U{\left((1-\sigma)a_U + \sigma a_{I_k}^*\right)},
\label{maxFORM}
\end{equation} 
where $c \coloneqq \beta SI{\left(C^{BS}-C^{FI}\right)}/[r(1-\sigma+\sigma S)]$. If $c \ge 0$ then $a_U = 1$. Otherwise the objective function is concave, and we only need to evaluate the first-order conditions. The first-order condition condition for \eqref{maxFORM} when $u$ is CRRA is 
\begin{equation}
0 = a_U^{-\alpha} + c(\sigma a_{I_k}^* + 2(1-\sigma) a_U).
\label{GENfoc1}
\end{equation}
In the case of log utility the first-order condition \eqref{GENfoc1} reduces to the quadratic
\begin{equation}
0 = 1 + c\sigma a_{I_k}^*a_U + 2c(1-\sigma) a_U^2.
\label{GENfoc1LOG}
\end{equation}

\begin{lem}\label{lem:alg}
For any $0 < a< b$, the solution to \eqref{maxFORM} is 
\begin{equation*}
a_U(c) = b1_{c\ge 0} + (1- 1_{c\ge 0})\max\set{a,  \min\set{a_U^\mathrm{FOC}(c), b }},
\end{equation*}
where $a_U^\mathrm{FOC}(c)$ solves \eqref{GENfoc1}. 
\end{lem}

To avoid overflow in the code we divide all quantities by $\Delta_t$ and consider the limit of the above as $\Delta_t \to 0$. The linear system we solve at each stage is
\begin{align}
0 = & -r{\left((1-\sigma+\sigma S)u(a_U) + \sigma Iu(a_{I_k}^*)\right)} - \gamma\delta\sigma Iu_D + \nu \bar{C}_{\mathrm{vac}}I \notag \\
& + \bar{p}^{-S}C(S-\Delta_S,I)  + \bar{p}^{+I}C(S,I+\Delta_I^+) + \bar{p}^{-I}C(S,I-\Delta_I^-) \notag \\
& - (r +\nu + \bar{p}^{-S}+\bar{p}^{+I} + \bar{p}^{-I})C, \label{linSPP}
\end{align}
where for each transition probability we have $\bar{p} = p/\Delta_t$. The above system may be written in the form $0 = \textnormal{cost} + T_{\textnormal{SPP}}C$ for some matrix $T_{\textnormal{SPP}}$. Beginning with an arbitrary guess $a_U$, we solve the planner's problem by solving the linear system \eqref{linSPP}, replacing $a_U$ with the implied policy function $a_U$ in Lemma \ref{lem:alg}, and repeating until convergence.

\subsection{Perfect Bayesian Markov competitive equilibrium}\label{subsec:solve.eq}

We now outline the numerical method for computing the perfect Bayesian equilibrium. The Hamilton-Jacobi-Bellman equation for an unknown agent is given by \eqref{PBEvalue2}. We suppose that at an arbitrary $(S,I) \in \Sigma$ there are three possible transitions, to $(S-\Delta_S,I)$, $(S,I-\Delta_I^-)$ and $(S,I+\Delta_I^+)$, with associated transition probabilities $p^{-S}, p^{-I}$ and $p^{+I}$ given by
\begin{subequations}\label{probCE1}
\begin{align}
p^{-S} & = \frac{\Delta_t}{\Delta_S}\beta S\tilde{a}(S,I)({\left((1-\sigma) \tilde{a}(S,I) + \sigma a_{I_k}^* \right)}I,
\\
p^{\pm I} & = \frac{\Delta_t}{\Delta_I^{\pm}}\max\set{ \pm{\left[\beta S\tilde{a}(S,I){\left((1-\sigma) \tilde{a}(S,I) + \sigma a_{I_k}^* \right)} -  \gamma\right]} I, 0}.
\end{align}
\end{subequations}
The probability with which an unknown agent becomes known infected is then
\begin{equation}
p^{uk} = \Delta_t\sigma \mu\beta a_U{\left((1-\sigma) \tilde{a}(S,I) + \sigma a_{I_k}^* \right)}I, \label{eq:puk}
\end{equation}
where the superscript indicates this is the probability of transitioning from the unknown $u$ state to the known infected state $k$. The Bellman equation for unknown agents is then 
\begin{align*} 
& V_U(S,I) = \max_{a_U \in [\ubar{a},1]} \Delta_t ru(a_U)  + \e^{-(r+\nu)\Delta_t}p^{uk}V_{I_k}(S,I)
\\
& + \e^{-(r+\nu)\Delta_t}{\left(p^{-S} V_U(S-\Delta_S,I+\Delta_S) + p^{-I}V_U(S,I-\Delta_I^-) + p^{+I} V_U(S,I+\Delta_I^+)\right)}
\\
& + \e^{-(r+\nu)\Delta_t}{\left(1 - p^{-S} - p^{-I} - p^{+I} -  p^{uk}\right)}V_U(S,I). 
\end{align*}
Omitting terms independent of $a_U$, dividing by $r\Delta_t$ and sending $\Delta_t \to 0$, the maximization problem becomes
\begin{equation*}
\max_{a_U \in [\ubar{a},1]} u(a_U) + \bar{p}^{uk}r^{-1}[V_{I_k}(S,I) - V_U(S,I)]
\end{equation*}
where $\bar{p}^{uk} = p^{uk}/\Delta_t$. This problem is equivalent to
\begin{equation*}
\max_{a_U \in [\ubar{a},1]} u(a_U) + ca_U{\left((1-\sigma) \tilde{a}(S,I) + \sigma a_{I_k}^* \right)},
\end{equation*}
where
\begin{equation}
c = \sigma \mu\beta r^{-1}[V_{I_k}(S,I) - V_U(S,I)]I.
\label{2ndMAXPBE}
\end{equation} 
When the utility function takes the form \eqref{eq:CRRA}, the first-order condition rearranges to $a_U = a_U^\mathrm{FOC} \coloneqq [- b]^{-1/\alpha}$ for $b = c{\left((1-\sigma) \tilde{a}(S,I) + \sigma a_{I_k}^* \right)}$ and the optimal choice of $a_U$ is therefore
\begin{equation}
a_U = 1_{c \ge 0} + (1-1_{c \ge 0})\max\set{\ubar{a}, \min\set{1, a_U^\mathrm{FOC}}}. \label{aUFOC}
\end{equation}
The linear system characterizing the value function of the unknown agents is
\begin{align}
0 = & ru(a_U) + \bar{p}^{uk}V_{I_k}(S,I) \notag
\\
& + \bar{p}^{-I}V_U(S,I-\Delta_I^-) + \bar{p}^{+I}V_U(S,I+\Delta_I^+) + \bar{p}^{-S} V_U(S-\Delta_S,I+\Delta_S) \notag
\\
& - {\left(r + \nu + \bar{p}^{-S} + \bar{p}^{-I} + \bar{p}^{+I} + \bar{p}^{uk}\right)}V_U(S,I).
\label{linAPBE}
\end{align}
Which may be written in the form $0 = b + TV_S$ for some matrix $T$. We then iterate upon the policy function $a_U$. Beginning with an arbitrary guess $a_U$, we compute the value function of the unknown agent by solving the linear system \eqref{linAPBE}, replace $a_U$ with the implied policy function $a_U$ in \eqref{aUFOC}, and repeat until convergence.

\subsection{Perfect recall Markov equilibrium}\label{sec:PRME_num}

We now outline the numerical method for computing the perfect recall Markov equilibrium. We define a grid for beliefs $\Sigma_\mu = \{0, 1/N_\mu, \dots, 1-1/N_\mu, 1\}$ for some integer $N_\mu \geq 1$, and now let the state space be $\Sigma\coloneqq \Sigma_S\times \Sigma_I\times \Sigma_\mu$.

For unknown agents, we must specify the transition probabilities for their beliefs and the probability with which they become known infected. The latter quantity is simply \eqref{eq:puk}. We suppose that at an arbitrary $(S,I,\mu) \in \Sigma$, there are four possible transitions, to $(S-\Delta_S,I,\mu)$, $(S,I-\Delta_I^-,\mu)$, $(S,I+\Delta_I^+,\mu)$ and $(S,I,\mu-\Delta_{\mu})$, with associated probabilities $p^{-S}$, $p^{-I}$, $p^{+I}$ and $p^{-\mu}$. The local consistency requirements for the aggregate state are again satisfied if we choose $p^{-S},  p^{-I}$ and $p^{+I}$ according to \eqref{probCE1}. Using \eqref{rhoLAW}, the local consistency requirement for the belief variable is
\begin{equation*}
-\Delta_{\mu}p^{-\mu} = -\Delta_t (1-\sigma)\mu \beta a_U{\left((1-\sigma) \tilde{a}(S,I) + \sigma a_{I_k}^*\right)}I + o(\Delta_t),
\end{equation*}
which will be satisfied if we choose
\begin{equation}
p^{-\mu} = \frac{\Delta_t}{\Delta_{\mu}}\mu(1-\sigma)\beta a_U{\left((1-\sigma) \tilde{a}(S,I) + \sigma a_{I_k}^*\right)}I = \frac{1}{\Delta_{\mu}}(1/\sigma-1) p^{uk},
\label{probSf6}
\end{equation}
where $p^{uk}$ is given by \eqref{eq:puk}. The Bellman equation for unknown agents is then 
\begin{align*}
V_U(S,I,\mu) = & \max_{a_U \in [\ubar{a},1]} \Delta_t ru(a_U)  + \e^{-(r+\nu)\Delta_t}\frac{\Delta_{\mu} \sigma}{1-\sigma}p^{-\mu}[V_{I_k}(S,I) - V_U(S,I,\mu)]
\\ & + \e^{-(r+\nu)\Delta_t}{\left(p^{-S} V_U(S-\Delta_S,I+\Delta_S,\mu) + p^{-I}V_U(S,I-\Delta_I^-,\mu; u)\right)}
\\ & + \e^{-(r+\nu)\Delta_t}{\left(p^{+I} V_U(S,I+\Delta_I^+,\mu) +  p^{-\mu}V_U(S,I,\mu-\Delta_{\mu})\right)}
\\ & + \e^{-(r+\nu)\Delta_t}{\left(1 - p^{-S} - p^{-I} - p^{+I} -  p^{-\mu} - \frac{\Delta_{\mu} \sigma}{1-\sigma}p^{-\mu}\right)}V_U(S,I,\mu). 
\end{align*}
Omitting terms independent of $a_U$, dividing by $r\Delta_t$ and sending $\Delta_t \to 0$, the maximization problem becomes
\begin{multline*}
\max_{a_U \in [\ubar{a},1]} u(a_U) \\
+\bar{p}^{-\mu}r^{-1}\left(\frac{\Delta_{\mu}\sigma}{1-\sigma}[V_{I_k}(S,I) - V_U(S,I,\mu)] + V_U(S,I,\mu-\Delta_{\mu}) - V_U(S,I,\mu) \right),
\end{multline*}
where $\bar{p}^{-\mu} = p^{-\mu}/\Delta_t$. Using \eqref{probSf6}, this problem is equivalent to
\begin{equation*}
\max_{a_U \in [\ubar{a},1]} u(a_U) + ca_U{\left((1-\sigma) \tilde{a}(S,I) + \sigma a_{I_k}^* \right)},
\end{equation*}
where
\begin{equation}
c = \frac{\mu\beta}{r\Delta_{\mu}}{\left(\Delta_{\mu}\sigma[V_{I_k}(S,I) - V_U(S,I,\mu)] + (1 - \sigma)[V_U(S,I,\mu-\Delta_{\mu}) - V_U(S,I,\mu)] \right)} I.
\label{2ndMAX}
\end{equation} 
When the utility function takes the form \eqref{eq:CRRA}, the first-order condition rearranges to $a_U = a_U^\mathrm{FOC} \coloneqq [- b]^{-1/\alpha}$ for $b = c{\left((1-\sigma) \tilde{a}(S,I) + \sigma a_{I_k}^* \right)}$ and the optimal choice of $a_U$ is therefore
\begin{equation}
a_U = 1_{c \geq 0} + (1-1_{c \geq 0})\max\set{\ubar{a}, \min\set{1, a_U^\mathrm{FOC}}}.
\label{polCRRA}
\end{equation}
The fact that the belief $\mu$ is monotonically decreasing allows us to solve the problem of the unknown agent in a manner similar to backward induction. Specifically, given values for $V_U(S,I,\mu-\Delta_{\mu})$ for $(S,I) \in \Sigma_S \times \Sigma_I$, when the agent adheres to $a$ and other agents adhere to $\tilde{a}$, the values of $V_U(S,I,\mu)$ for $(S,I) \in \Sigma_S \times \Sigma_I$ solve
\begin{align}
0 = & ru(a_U) + \frac{\Delta_{\mu} \sigma}{1-\sigma}\bar{p}^{-\mu}V_{I_k}(S,I) + \bar{p}^{-\mu}V_U(S,I,\mu-\Delta_{\mu}) \notag \\
& + \bar{p}^{-I}V_U(S,I-\Delta_I^-,\mu) + \bar{p}^{+I}V_U(S,I+\Delta_I^+,\mu) + \bar{p}^{-S} V_U(S-\Delta_S,I+\Delta_S,\mu) \notag \\
& - {\left(r+\nu + \bar{p}^{-S} + \bar{p}^{-I} + \bar{p}^{+I} +  \bar{p}^{-\mu} + \frac{\Delta_{\mu} \sigma}{1-\sigma}\bar{p}^{-\mu}\right)}V_U(S,I,\mu). \label{linA}
\end{align}
In this way the solution to the problem of the unknown agent may be found by repeatedly solving two-dimensional dynamic programming problems, using the fact that $V_U(S,I,0) = V_{R_k} = u(1) = 0$.

We then iterate upon the policy function $\tilde{a}$. Beginning with an arbitrary guess $\tilde{a}$, we compute the value function of the unknown agent by solving the above discrete-time Bellman equation. If $a(S,I,\mu;\tilde{a})$ is the resulting policy function, we then replace $\tilde{a}$ with $\tilde{a}' = a(S,I,S/(\sigma S + 1 - \sigma); \tilde{a})$ and repeat until convergence.

\section{Calibrating $u_D$}\label{sec:uD}

We calibrate the flow utility from death $u_D$ based on the case study from Sweden, which did not introduce mandatory lockdowns. For this purpose, we first obtain the daily cumulative number of reported cases and deaths for Sweden from Johns Hopkins University CSSE (Footnote \ref{fn:CSSE}). Let $N_{C,t},N_{D,t}$ be the cases and deaths up to date $t$. We divide these numbers by the population $N=10.38\times 10^6$ to obtain the population shares of reported cases $C_t\coloneqq N_{C,t}/N$ and deaths $D_t\coloneqq N_{D,t}/N$. We compute the case fatality rate as $\delta_{\mathrm{CFR},t}\coloneqq D_t/C_t$. Figure \ref{fig:CFR_SW} shows the evolution of the case fatality rate. The fact that the CFR peaked at around 12\% in May 2020 and has settled down below 2\% at the time of writing suggests that the diagnosis rate has significantly changed over time and that the reported cases are unreliable. 

We thus estimate the fraction of infected population $I_t$ using the accounting equation from the SIR model
\begin{equation}
D_{t+1}-D_t=\gamma \delta_{\mathrm{IFR}}I_t\iff I_t=\frac{D_{t+1}-D_t}{\gamma\delta_{\mathrm{IFR}}}.\label{eq:I_estim}
\end{equation}
To control for the day-of-week effect (deaths seem to be unreported over the weekend), we take the 7-day moving average of \eqref{eq:I_estim}, which is plotted in Figure \ref{fig:I_SW.pdf}. The estimated peak prevalence is thus $\max_t I_t=0.0663$. Finally, we choose the value for $u_D$ to match the peak prevalence in the model and obtain $u_D=\uD$.

\begin{figure}[!htb]
\centering
\begin{subfigure}{0.48\linewidth}
\includegraphics[width=\linewidth]{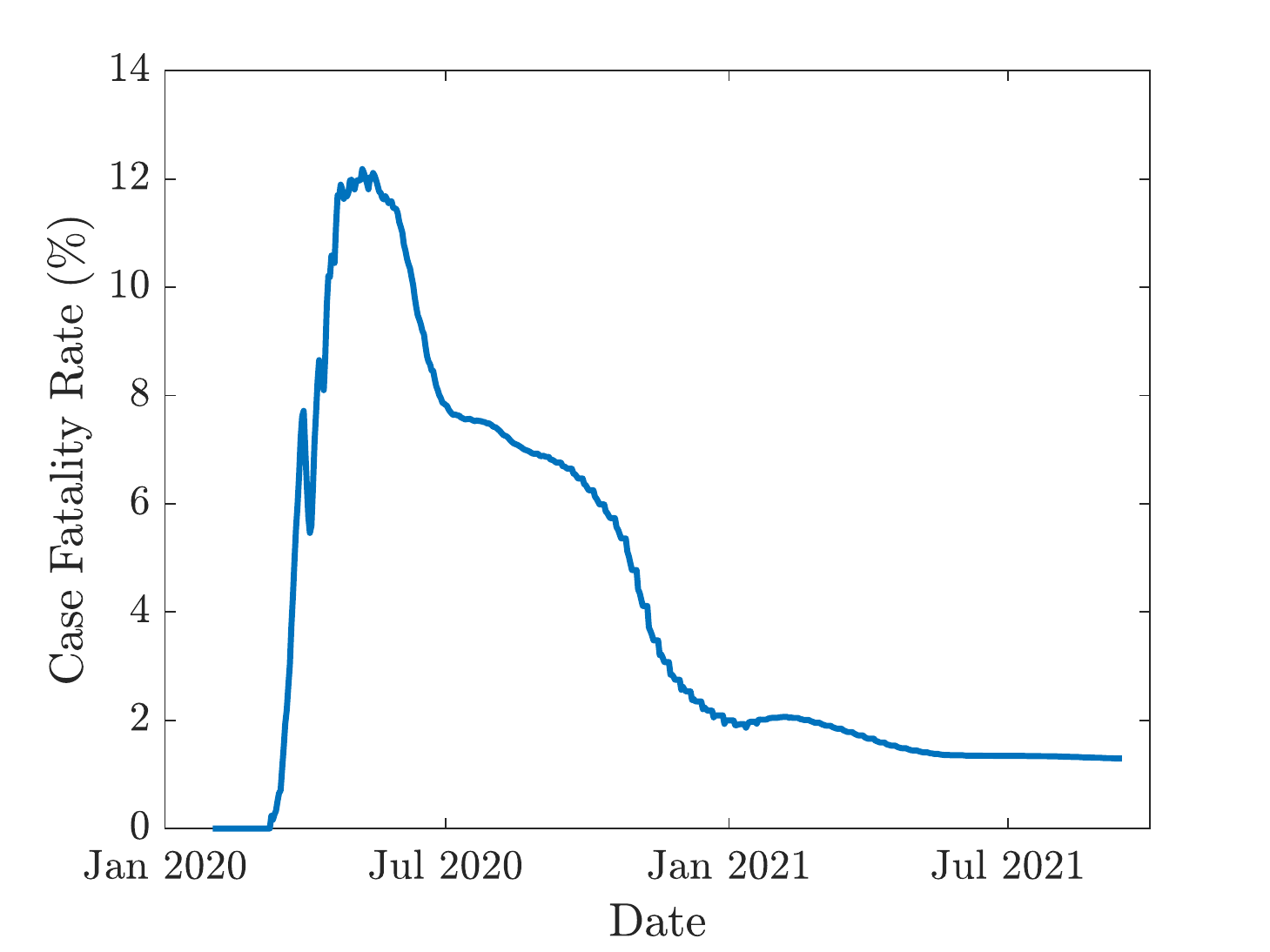}
\caption{Case Fatality Rate (\%) over time.}\label{fig:CFR_SW}
\end{subfigure}
\begin{subfigure}{0.48\linewidth}
\includegraphics[width=\linewidth]{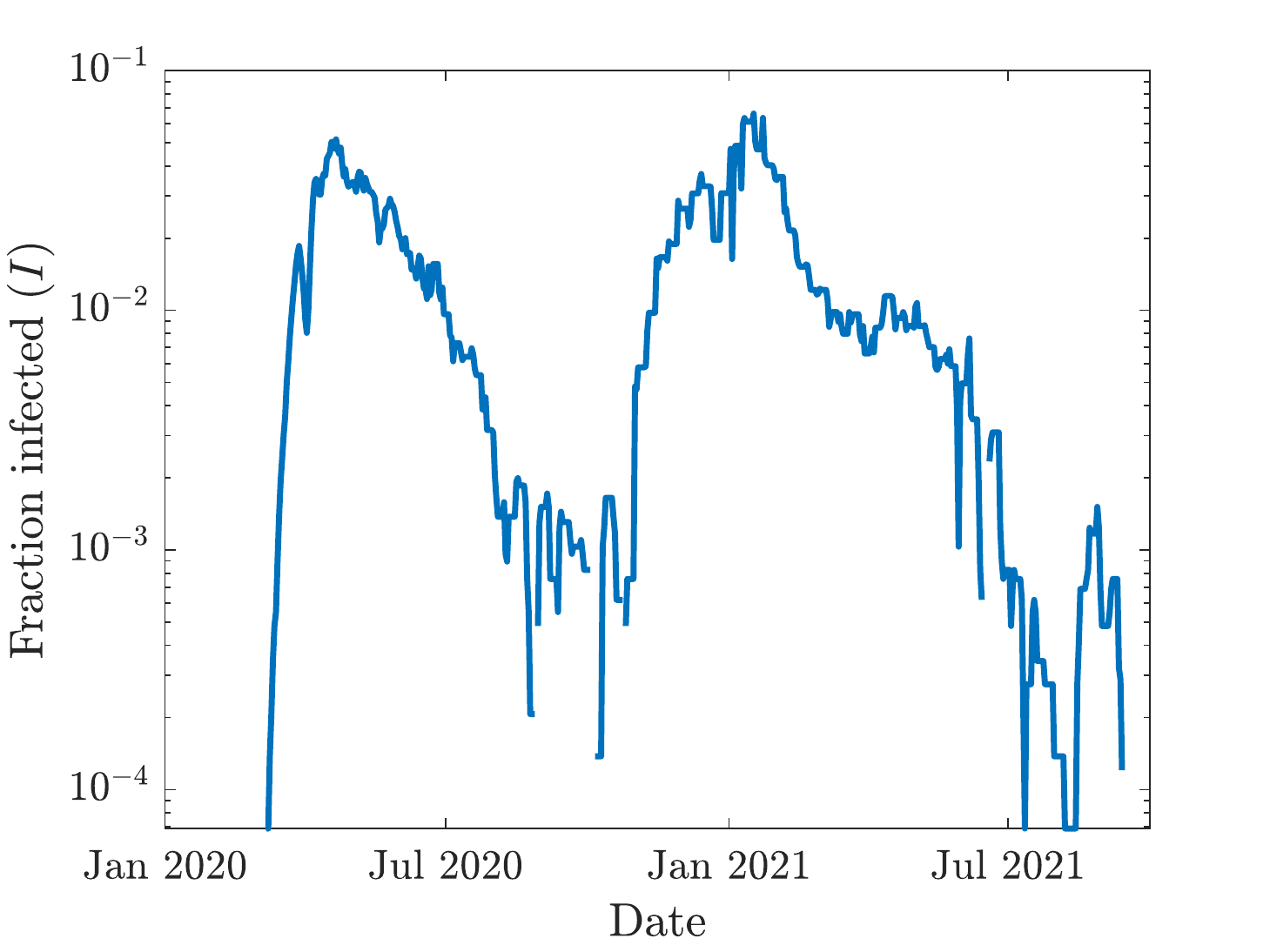}
\caption{Estimated prevalence over time.}\label{fig:I_SW.pdf}
\end{subfigure}
\caption{Case Fatality Rate and estimated prevalence in Sweden.}\label{fig:Sweden}
\end{figure}

As a robustness check, we also compute $u_D$ from the value of statistical life. Consider an individual consuming a constant flow normalized to 1 and facing a small probability $d$ of death for one period. Letting $\beta\in (0,1)$ be the agent's discount factor, $p$ be the willingness to pay to avoid the possibility of death, $V$ be the continuation value of being alive, and $V_D$ be the continuation value of being dead, by definition we have
\begin{equation}
(1-\beta)u(1-p)+\beta V=(1-\beta)u(1)+\beta((1-d)V+dV_D). \label{eq:wtp}
\end{equation}
Using $V=u(1)$ and $V_D=u_D$, \eqref{eq:wtp} simplifies to
\begin{equation}
(1-\beta)u(1-p)=(1-\beta-\beta d)u(1)+\beta du_D.\label{eq:wtp2}
\end{equation}
The \emph{value of statistical life} $v$ is defined by the
willingness to pay $p$ scaled such that one life is saved on average, so $v=p/d$. Since
\begin{equation*}
u(1-p)\approx u(1)-u'(1)p=u(1)-u'(1)vd
\end{equation*}
by the Taylor approximation, solving \eqref{eq:wtp2} for $u_D$, we obtain
\begin{equation}
u_D\approx u(1)-\frac{1-\beta}{\beta}u'(1)v.\label{eq:uD}
\end{equation}
Since $u(1)=0$ and $u'(1)=1$ for the CRRA utility \eqref{eq:CRRA}, \eqref{eq:uD} further simplifies to
\begin{equation}
u_D\approx -\frac{1-\beta}{\beta}v.\label{eq:uD_CRRA}
\end{equation}
\cite{HallJonesKlenow2020} note that The Environmental Protection Agency (EPA) recommends using $v=\text{\$7.4 million}$ in 2006,\footnote{\url{https://www.epa.gov/environmental-economics/mortality-risk-valuation}} when per capita consumption was about \$31,000. Thus the value of life in units of one period consumption is
\begin{equation*}
v=\frac{\text{7,400,000}}{\text{31,000}}=238.7,
\end{equation*}
and with 5\% annual discounting, the flow utility from death in \eqref{eq:uD_CRRA} is
\begin{equation*}
u_D=-0.05\times 238.7=-11.94,
\end{equation*}
which is very similar to the value we obtained from the case study from Sweden.

\end{document}